\documentclass{l4dc2024}

\usepackage{wrapfig}
\usepackage{hyperref}
\usepackage{cleveref}
\usepackage{algorithm}
\usepackage{algorithmicx}
\usepackage[noend]{algpseudocode}
\usepackage{graphicx}
\usepackage{Shorthands}
\usepackage{enumitem}
\usepackage{ifthen}
\newcounter{l4dc}

\makeatletter
\renewcommand{\jmlrvolume}[1]{}
\renewcommand{\jmlryear}[1]{}
\renewcommand{\jmlrproceedings}[2]{}
\def\ps@jmlrtps{%
  \def\@oddhead{}% Clear header for odd pages
  \def\@evenhead{}% Clear header for even pages
  \def\@oddfoot{}% Clear footer for odd pages
  \def\@evenfoot{}% Clear footer for even pages
}
\makeatother
\usepackage[noindentafter]{titlesec}
\titlespacing\section{0pt}{10pt}{4pt}
\titlespacing\subsection{0pt}{4pt}{3pt}

\usepackage{times}
\title[Nonasymptotic Regret Analysis of Adaptive LQ Control with Model Misspecification]{Nonasymptotic Regret Analysis of Adaptive Linear Quadratic Control with Model Misspecification}

\author{%
 \Name{Bruce D. Lee$^1$} \Email{brucele@seas.upenn.edu} \\
 \Name{Anders Rantzer$^2$} \Email{anders.rantzer@control.lth.se} \\
 \Name{Nikolai Matni$^1$} \Email{nmatni@seas.upenn.edu} \\
 \addr $^1$ Department of Electrical and Systems Engineering, University of Pennsylvania \\
 $^2$ Department of Automatic Control, Lund University
}

\date{August 2023}
\begin{document}
\maketitle

% \vspace{-0pt}
\begin{abstract}
    The strategy of pre-training a large model on a diverse dataset, then fine-tuning for a particular application has yielded impressive results in computer vision, natural language processing, and robotic control. This strategy has vast potential in adaptive control, where it is necessary to rapidly adapt to changing conditions with limited data. Toward concretely understanding the benefit of pre-training for adaptive control, we study the adaptive linear quadratic control problem  in the setting where the learner has prior knowledge of a collection of basis matrices for the dynamics. This basis is misspecified in the sense that it cannot perfectly represent the dynamics of the underlying data generating process. We propose an algorithm that uses this prior knowledge, and prove upper bounds on the expected regret after $T$ interactions with the system. In the regime where $T$ is small, the upper bounds are dominated by a term that scales with either $\texttt{poly}(\log T)$ or $\sqrt{T}$, depending on the prior knowledge available to the learner.  When $T$ is large, the regret is dominated by a term that grows with $\delta T$, where $\delta$ quantifies the level of misspecification. This linear term arises due to the inability to perfectly estimate the underlying dynamics using the misspecified basis, and is therefore unavoidable unless the basis matrices are also adapted online. However, it only dominates for large $T$, after the sublinear terms arising due to the error in estimating the weights for the basis matrices become negligible. We provide simulations that validate our analysis. Our simulations also show that offline data from a collection of related systems can be used as part of a pre-training stage to estimate a misspecified dynamics basis, which is in turn used by our adaptive controller.
\end{abstract}

\section{Introduction}

Transfer learning, whereby a model is pre-trained on a large dataset, and then finetuned for a specific application, has exhibited great success in computer vision \citep{dosovitskiy2020image} and natural language processing \citep{devlin2018bert}. Efforts to apply these methods to control have shown exciting preliminary results, particularly in robotics \citep{dasari2019robonet}. 
The principle underpinning the success of transfer learning is to use diverse datasets to extract compressed, broadly useful features, which can be used in conjunction with comparatively simple models for downstream objectives. These simple models can be finetuned with relatively little data from the downstream task. However, errors in the pre-training stage may cause this two-step strategy to underperform learning from scratch when ample task-specific data is available. This tradeoff may be acceptable in settings such as adaptive control, where the learner must rapidly adapt to changes with limited data. %However, efforts to analyze adaptive control from a nonasymptotic perspective do not consider misspecified models, such as those resulting from the pre-training stage of transfer learning. 

Driven by the potential of pre-training in adaptive control, we study the adaptive linear quadratic regulator (LQR) in a setting where imperfect prior information about the system is available. The adaptive LQR problem consists of a learner interacting with an unknown linear system
\begin{align}
    \label{eq: dynamics}
    x_{t+1} = A^\star x_t + B^\star u_t + w_t,
\end{align} 
with state $x_t$, input $u_t$, and noise $w_t$  assuming values in $\R^{\dx}$, $\R^{\du}$, and $\R^{\dx}$, respectively. The learner is evaluated by its ability to minimize its \emph{regret}, which compares the cost incurred by playing the learner for $T$ time steps with the cost attained by the optimal LQR controller. Prior work has studied the adaptive LQR problem in the settings where $A^\star$, $B^\star$, or both $A^\star$ and $B^\star$ are fully unknown to the learner and with either stochastic or bounded adversarial noise processes \citep{abbasi2011regret, hazan2020nonstochastic, cassel2020logarithmic}. In this work, we analyze a setting where the learner has a set of basis matrices for the system dynamics; however, $\bmat{A^\star & B^\star}$ is not in the subspace spanned by this basis, leading to misspecification. %Such an estimate for a misspecified basis may arise from learning a representation using offline data from a collection of related systems.

\subsection{Related Work}

\paragraph{Adaptive Control} Originating from autopilot development for high-performance aircraft in the 1950s \citep{gregory1959proceedings}, adaptive control addressed the need for controllers to adapt to varying altitude, speed, and flight configuration \citep{stein1980adaptive}. Interest in adaptive control theory grew over the subsequent decades, notably advanced by the landmark paper by \cite{aastrom1973self}, who studied self tuning regulators. %The theoretical foundation paved the way for application of adaptive control in various domains, such as robotic manipulation \citep{slotine1987adaptive}. 
The history of adaptive control is documented in numerous texts \citep{aastrom2013adaptive, ioannou1996robust, narendra2012stable}. 
\vspace{-6pt}
\paragraph{Nonasymptotic Adaptive LQR} The non-asymptotic study of the adaptive LQR problem
%, where a learner interacts with a linear dynamical system under stochastic disturbances, 
was pioneered by \cite{abbasi2011regret}. Subsequent works \citep{dean2018regret, cohen2019learning, mania2019certainty} developed computationally efficient algorithms with upper bounds on the regret scaling with $\sqrt{T}$. \cite{simchowitz2020naive} provide lower bounds which show that the rate $\sqrt{\du^2\dx T}$ is optimal when the system is entirely unknown. %, where $\dx$ and $\du$ are the state and input dimensions, respectively. 
The dependence of the lower bounds on system theoretic constants is refined by \cite{ziemann2022regret}, enabling \cite{tsiamis2022learning} to show the regret may depend exponentially on $\dx$ for specific classes of systems. If either $A^\star$ or $B^\star$ are known, then the regret bounds may be improved to $\texttt{poly}(\log T)$ \citep{cassel2020logarithmic, jedra2022minimal}. Alternative formulations of the adaptive LQR problem consider bounded non-stochastic disturbances \citep{hazan2020nonstochastic, simchowitz2020improper} and minimax settings \citep{rantzer2021minimax, cederberg2022synthesis, renganathan2023online}. %The results in this setting demonstrate an $\calH_\infty$ norm bound on the system which can be used in conjunction with the small-gain theorem to guarantee stability of the closed-loop system in the face of bounded unmodeled dynamics. 
In contrast to existing work studying the adaptive LQR problem from a nonasymptotic perspective, we consider bounded misspecification between a representation estimate for the dynamics and the data generating process. 
\vspace{-20pt}
\paragraph{Multi-task Representation Learning} The source of mispecification we consider is inspired by theoretical work studying multi-task representation learning in the linear regression setting \citep{du2020few,tripuraneni2020theory}. These papers have been followed by a collection of work studying the use of multi-task representation learning in the presence of data correlated across time, which arises  in system identification \citep{modi2021joint, zhang2023meta} and imitation learning \citep{zhang2023multi}. All of these works show that by pre-training a shared representation on a set of source tasks, the sample complexity of learning the target task may be reduced. 
% In this setting, we borrow this notion of an approximate representation, but we consider how such a representation may be used in online adaptive control. 

\subsection{Contributions}

We introduce a notion of misspecification between an estimate for the basis of the dynamics and the data generating process. We then propose an adaptive control algorithm that uses this misspecified basis, and subsequently analyze the regret of this algorithm. This leads to the following insights:
\begin{itemize}[noitemsep, nolistsep, leftmargin=*]
    \item Our results generalize the understanding of when logarithmic regret is possible in adaptive LQR beyond the cases of known $A^\star$ or $B^\star$ studied by \cite{cassel2020logarithmic, jedra2022minimal}. 
    %demonstrate that logarithmic regret is possible when either the system matrix $A$ or the control matrix $B$ of a linear model $x_{t+1} = Ax_t + Bu_t +w_t$, are fully known. We show that logarthmic regret is possible for more general settings of unknown parameters for which the optimal policy is persistently exciting. %\Bruce{Cite this one from hassibi etc that's kinda wrong?}
    \item When misspecification is present, the regret incurs a term that is linear in $T$. The coefficient for this term decays gracefully as the level of misspecification diminishes. As a result, small misspecification means the regret is dominated by sublinear terms in the low data regime of interest for adaptive control. These terms are favorable to those possible in the absence of prior knowledge.
\end{itemize}
We validate our theory with numerical experiments, and show the benefit using a dynamics representation determined  by pre-training on offline data from related systems for adaptive control.

% \Bruce{State an informal theorem here? Seems challenging before introducing the model assumptions in the problem formulation.}

\subsection{Notation}
The Euclidean norm of a vector $x$ is denoted $\norm{x}$. For a matrix $A$, the spectral norm is denoted $\norm{A}$, and the Frobenius norm is denoted $\norm{A}_F$. We use $\dagger$ to denote the Moore-Penrose pseudo-inverse. The spectral radius of a square matrix is denoted $\rho(A)$. \ifnum\value{l4dc}>0{The minimum eigenvalue of a symmetric, positive definite matrix $A$ is denoted $\lambda_{\min}(A)$}\else{A symmetric, positive semi-definite (psd) matrix $A = A^\top$ is denoted $A \succeq 0$.  Similarly $A \succeq B$ denotes that $A-B$ is positive semidefinite. The minimum eigenvalue of a symmetric, positive definite matrix $A$ is denoted $\lambda_{\min}(A)$. }\fi  %We denote the normal distribution with mean $\mu$ and covariance $\Sigma$ by  $\calN(\mu, \Sigma)$. 
 For $f,g :D \to \mathbb{R}$, we write $f \lesssim g$ if for some $c > 0$, $f(x) \leq c g(x)\,\forall x \in D$. We denote the solutions to the discrete Lyapunov equation by $\dlyap(A,Q)$ and the discrete algebraic Riccati equation by $\dare(A,B,Q,R)$. \ifnum\value{l4dc}>0{}\else{We use $x \vee y$ to denote the maximum of $x$ and $y$.}\fi
%For an autonomous system $x_{t+1} = Ax_t$ and a symmetric matrix $Q$, we denote hte solution $P$ to the discrete Lyapunov equation $P = A^\top P A + Q$ by $\dlyap(A,Q)$. Similarly, for a controlled system $x_{t+1} = Ax_t + Bu_t$ and symmetric matrices $Q,R$ of compatible sizes, we denote the solution $P$ to the discrete algebraic Riccati equation by $P = A^\top P B(B^\top P B+R)^{-1} B^\top P A + A^\top P A +Q$ by $\dare(A,B,Q,R)$.
\section{Problem Formulation}
\vspace{-4pt}
\subsection{System model}
% \Bruce{Refer back to the linear system defined in the intro instead, and just place the dynamics assumptions and noise assumptions on it here. }

We consider the system \eqref{eq: dynamics} where the noise $w_t$ has independent identically distributed elements that are mean zero and $\sigma^2$-sub-Gaussian for some $\sigma^2\in \R$ with $\sigma^2 \geq 1$. We additionally assume that the noise has identity covariance: $\E \brac{w_t  w_t^\top} = I$.\footnote{Noise that enters the process through a non-singular matrix $H$ can be addressed by rescaling the dynamics by $H^{-1}$.} 
We suppose the dynamics admit the  decomposition
\begin{align}
    \label{eq: low dim structure}
    \bmat{A^\star & B^\star} &= \VEC^{-1}\paren{\Phi^\star \theta^\star}, 
\end{align}
where $\Phi^\star\in \R^{\dx(\dx+\du)\times \dtheta}$ specifies the model structure, and has orthonormal columns. Meanwhile, $\theta^\star \in\R^{\dtheta}$ specifies the parameters. The operator $\VEC^{-1}$ maps a vector in $\R^{\dx(\dx+\du)}$ into a matrix in $\R^{\dx \times (\dx+\du)}$ by stacking length $\dx$ blocks of the the vector into columns of the matrix, working top to bottom and left to right. %The dynamics will be represented concisely by defining $\bmat{A^\star & B^\star} \triangleq \VEC^{-1}\paren{\Phi^\star \theta^\star}$.
We can write this as a linear combination of basis matrices:
\begin{align*}
    \bmat{A^\star & B^\star} = \sum_{i=1}^{\dtheta} \theta^\star_i \bmat{\Phi^{A,\star}_i & \Phi^{B,\star}_i}, \mbox{ where } \bmat{\Phi^{A,\star}_i & \Phi^{B,\star}_i} = \VEC^{-1} \Phi^\star_i,
\end{align*}
and $\Phi^\star_i$ is the $i^{\mathrm{th}}$ column of $\Phi^\star$. This decomposition of the data generating process is a natural extension of the low-dimensional linear representations considered in \cite{du2020few} to the setting of multiple related dynamical systems with shared structure determined by $\Phi^\star$. It  captures many practically relevant settings, such as when $A^\star$ and $B^\star$ depend on a few physical parameters $\theta^\star$, and $\Phi^\star$ describes the structure through which these physical parameters enter the dynamics.

We assume that both $\Phi^\star$ and $\theta^\star$ are unknown; however, we have an estimate $\hat \Phi \in \R^{\dx(\dx+\du)\times \dtheta}$, also with orthonormal columns, for $\Phi^\star$. Such an estimate may be obtained by performing a pre-training step
%multi-task system identification 
on offline data from a collection of systems related to \eqref{eq: dynamics} by the shared matrix $\Phi^\star$ in \eqref{eq: low dim structure}. Due to the noise present in the offline data, this estimate will be imperfect, resulting  in misspecification.\footnote{Sample complexity bounds for learning $\hat \Phi$ are provided by \cite{zhang2023meta}, so this step is not studied here.} To quantify the level of misspecification between this estimate and the underlying data generating process, we use the following subspace distance metric. % from \cite{zhang2023meta}.

\begin{definition}[\citep{stewart1990matrix}] %zhang2023meta, 
    \label{def: representation error} 
    Let $\Phi^\star_\perp$ complete the basis of $\Phi^\star$ such that $\bmat{\Phi^\star & \Phi^\star_\perp}$ is an orthogonal matrix. Then the \emph{subspace distance} between $\Phi^\star$ and $\hat \Phi$ is 
    $
        d(\Phi^\star, \hat \Phi) \triangleq \norm{\hat \Phi^\top \Phi^\star_\perp}.
    $
\end{definition}
Note that the above distance is small when the range of the matrices $\hat \Phi$ and $\Phi^\star$ is similar.\footnote{This distance may be small when $\norm{\hat \Phi - \Phi^\star}$ is not. However,  small $\norm{\hat \Phi - \Phi_\star}$ implies small subspace distance. } 
% An estimate with a small subspace distance may be obtained by fitting a model on data generated by systems which share the structure $\Phi^\star$, e.g. using multi-task representation learning \citep{zhang2023multi}.

As long as $d(\Phi^\star, \hat \Phi)$ is sufficiently small and the dimension $\dtheta < \dx (\dx+\du)$, the estimate $\hat \Phi$ allows the learner to fit a model with less data than would be required to estimate $\bmat{A^\star & B^\star}$ from scratch. This benefit comes at the cost of a bias in the learner's model that grows with $d(\Phi^\star, \hat \Phi)$.

\subsection{Learning objective}
The goal of the learner is to interact with system \eqref{eq: dynamics} while keeping the cumulative cost small, where the cumulative cost is defined for matrices $Q \succeq I$ and $R = I$ as\footnote{Generalizing to arbitrary $Q\succ 0$ and $R\succ 0$ can be performed by scaling the cost and changing the input basis. } 
\begin{align}
    \label{eq: cost}
    C_T \triangleq \sum_{t=1}^T c_t, \mbox{ where } c_t \triangleq x_t^\top Q x_t  + u_t^\top R u_t.
\end{align} 
To define an algorithm that keeps the cost small, we first introduce the infinite horizon LQR cost:
\begin{align}
    \label{eq: lqr cost}
    \calJ(K) \triangleq \limsup_{T\to\infty} \frac{1}{T} \E^K C_T,
\end{align}
where the superscript of $K$ on the expectation denotes that the cost is evaluated under the state feedback controller $u_t = Kx_t$. To ensure that there exists a controller such that \eqref{eq: lqr cost} has finite cost, we assume $(A^\star, B^\star)$ is stabilizable. Under this assumption, 
% The following standard assumption ensures that there exists a controller $K$ for which $\calJ(K)$ is finite, and that any such controller is stabilizing  \cite{zhou1996robust}:
% \begin{assumption}
%    The pair $(A^\star,B^\star)$ is stabilizable and the pair $(A^\star, Q^{1/2})$ is detectable. 
% \end{assumption}
% \sloppy Under the above assumption, 
\eqref{eq: lqr cost} is minimized by the LQR controller $K_{\infty}(A^\star,B^\star)$, where
% \begin{align*}
    $
    K_{\infty}(A,B) \triangleq - (B^\top P_{\infty}(A,B) B + R)^{-1} B^\top P_{\infty}(A,B) A,
    $
% \end{align*}
and $P_{\infty}(A, B) \triangleq \dare(A,B,Q,R)$. We define the shorthands $P^\star \triangleq P_{\infty}(A^\star, B^\star)$ and $K^\star \triangleq K_\infty(A^\star, B^\star)$. To characterize the infinite horizon LQR cost of an arbitrary stabilizing controller $K$, we additionally define the solution $P_K$ to the lyapunov equation for the closed loop system under an arbitrary $K$ such that $\rho(A^\star + B^\star K) < 1$: $P_K \triangleq \dlyap(A^\star + B^\star K, Q+ K^\top R K)$.
For a controller $K$ satisfying $\rho(A^\star + B^\star K) < 1$, $\calJ(K) = \trace(P_K)$. We have that $P_{K^\star} = P^\star$. 

The infinite horizon LQR controller provides a baseline level of performance that our learner cannot surpass in the limit as $T \to \infty$. Borrowing the notion of regret from online learning, as in \cite{abbasi2011regret}, we quantify the performance of our learning algorithm by comparing the cumulative cost $C_T$ to the scaled infinite horizon cost attained by the LQR controller if the system matrices $\bmat{A^\star & B^\star}$ were known:
\begin{align}
    \label{eq: regret}
    \mathbf{R}_T \triangleq C_T - T \calJ(K^\star).
\end{align}
% \Bruce{Make this more crisp}
In light of the above reformulation, the goal of the learner is to interact with the system \eqref{eq: dynamics} to maximize the information about the relevant parameters for control while simultaneously regulating the system to minimize $\mathbf{R}_T$. A reasonable strategy to do so is for the learner to use its history of  interaction with the system to construct a model for the dynamics, e.g. by determining estimates $\hat A$
and $\hat B$. It may then use these estimates as part of a \emph{certainty equivalent} (CE) design by synthesizing controllers $\hat K = K_\infty(\hat A, \hat B)$. Prior work has shown that if the model estimate is sufficiently close to the true dynamics, then the cost of playing the controller $\hat K$ exceeds the cost of playing $K^\star$ by a quantity that is quadratic in the estimation error \citep{mania2019certainty, simchowitz2020naive}. \ifnum\value{l4dc}>0{}\else{This fact will be used to characterize the performance of the learner in terms of the error in the estimation of the system matrices.}\fi %Here, sufficiently close may be defined in terms of a ``safe'' ball around dynamics $\bmat{A^\star & B^\star}$.

\begin{lemma}[Theorem 3 of \cite{simchowitz2020naive}]
    \label{lem: CE closeness main body} 
    \sloppy Define $\varepsilon \triangleq \frac{1}{2916 \norm{P^\star}^{10}}$. If
    % \begin{align}
    % \label{eq: control safety condition}
    $
        \norm{\bmat{\hat A & \hat B}-\bmat{A^\star & B^\star}}_F^2 \leq \varepsilon, 
    $
    then $\calJ(\hat K)  - \calJ(K^\star) \leq 142 \norm{P^\star}^8 \norm{\bmat{\hat A & \hat B}-\bmat{A^\star & B^\star}}_F^2$.
    %%$P_{\hat K} \preceq \frac{21}{20} P^\star$, $\norm{\hat K - K^\star}\leq \frac{1}{6 \norm{P^\star}^{3/2}}$, and
\end{lemma}

\subsection{Algorithm description}

 \begin{algorithm}[t]
 \caption{Certainty Equivalent Control with Continual Exploration} 
 \label{alg: ce with exploration}
\begin{algorithmic}[1]
\State \textbf{Input: } Stabilizing controller $K_0$, initial epoch length $\epoch$, number of epochs $k_{\fin}$, exploration sequence $\sigma_1^2, \sigma_2^2, \sigma_3^2, \dots \sigma_{k_{\fin}}^2$, state bound $x_b$, controller bound $K_b$
\State \textbf{Initialize: } $\hat K_1 \gets K_0$,  $\tau_0 \gets 0$, $T \gets \tau_1 2^{k_{\fin}-1}$.
\For{$k=1,2, \dots, k_{\fin}$}
     \For{$t=\tau_{k-1}+1, \dots, \tau_{k}$}
            \If{$\norm{x_t}^2 \geq x_b^2 \log T$ or $\norm{\hat K_k}\geq K_b$} Abort and play $K_0$ forever
            \EndIf
            \State Play $u_t = \hat K_k x_t + \sigma_k g_t$, where $g_t \sim \calN(0,I)$ %and $\sigma_k^2:= \min\curly{1, \sigma_{\mathsf{in}, 1}^2 \tau_k^{-1/2} + \sigma_{\mathsf{in},2}^2 \varepsilon}$ 
        \EndFor
    \State $\hat \theta_{k} \gets \texttt{LS}(\hat \Phi, x_{\tau_{k-1}+1:\tau_k+1}, u_{\tau_{k-1}+1: \tau_k})$ \Comment{\Cref{alg: least squares}}
    \State $\bmat{\hat A_{k} & \hat B_{k}} \gets \VEC^{-1} \paren{\hat \Phi \hat \theta_{k}} $
    \State $\hat K_{k+1} \gets K_\infty(\hat A_{k}, \hat B_{k})$
    \State $\tau_{k+1} \gets 2 \tau_k$
\EndFor
\end{algorithmic}
\end{algorithm}
 
Our proposed algorithm, \Cref{alg: ce with exploration}, is a CE algorithm akin to that proposed in \cite{ cassel2020logarithmic}. The algorithm takes a stabilizing controller $K_0$ as an input.\ifnum\value{l4dc}>0{}\else{\footnote{While providing a stabilizing controller is common in prior work, not all analyses require it \citep{lale2022reinforcement}. } }\fi
% \begin{assumption}
%     The learner has access to a stabilizing controller $K_0$ such that $\rho(A^\star + B^\star K_0) < 1$.
% \end{assumption} 
 Starting from this controller, \Cref{alg: ce with exploration} follows a doubling epochs strategy. At the end of each epoch, it uses the data collected during the epoch along with the estimate for $\hat \Phi$ to obtain an estimate for $\hat \theta$ by solving a least squares problem (as detailed in \Cref{alg: least squares}). The estimated parameters $\hat \theta$ are combined with the estimate $\hat \Phi$ to obtain the dynamics estimate $\bmat{\hat A & \hat B}$. This estimate is used to synthesize a CE controller $\hat K = K_\infty(\hat A, \hat B)$. In the next epoch, the learner plays the resultant controller with exploratory noise added. Before playing each input, the algorithm checks whether the state or the controller exceed bounds determined by algorithm inputs $x_b$ and $K_b$. If they do, it aborts the certainty equivalent scheme, and plays the initial stabilizing controller for all time.\ifnum\value{l4dc}>0{}\else{\footnote{Alternatively, one could use a hysteresis mechanism \citep{jedra2022minimal}; however, the naive strategy proposed suffices for our analysis. } }\fi Doing so enables bounding of the regret during unlikely events where the CE controller fails. The key difference from the CE algorithms proposed in prior work is that the system identification step solves a least squares problem defined in terms of the estimate $\hat \Phi$ to estimate the unknown parameters. %This is in contrast to prior work, where the least squares problem is solved to directly estimate $A^\star$ and $B^\star$. 

\vspace{-2pt}
\begin{algorithm}
\caption{Least squares: $\texttt{LS}(\hat \Phi, x_{1:t+1}, u_{1:t+1})$} 
\label{alg: least squares}
\begin{algorithmic}[1]
\State \textbf{Input:} Model structure estimate $\hat \Phi$, state data $x_{1:t+1}$, input data $u_{1:t}$
\State \textbf{Return:}
$
    \hat \theta \!=\! \Lambda^\dagger \paren{\sum_{s=1}^{t}\! \hat \Phi^\top \paren{ \bmat{x_s \\ u_s} \! \otimes\! I_{\dx}} x_{s+1} },  \, \mbox{where} \, \Lambda = \sum_{s=1}^{t } \! \hat \Phi^\top \! \paren{ \bmat{x_s \\ u_s}\!\bmat{x_s \\ u_s} ^\top \!\otimes \! I_{\dx}}\hat \Phi.
$
\end{algorithmic}
\end{algorithm}
\vspace{-5pt}
\section{Regret Bounds}
\label{s: analysis}

We now present our bounds on the expected\footnote{In contrast to high probability regret bounds, expected regret provides an understanding of what happens in the unlikely events where controller performs poorly.} % It also prevents terms of order $\sqrt{T}$ that arise in high probability regret bounds due to the process noise. } 
regret incurred by \Cref{alg: ce with exploration}. \ifnum\value{l4dc}>0{Further discussion and complete proofs may be found in \cite{lee2023nonasymptotic}. 

}\else{}\fi
Consider running \Cref{alg: ce with exploration} for $T = \tau_1 2^{k_\fin-1} $ timesteps, where $\tau_1$ is the initial epoch length and $k_{\fin}$ is the number of epochs. To bound the regret incurred by this algorithm, we decompose the regret into that achieved by the algorithm under a high probability success event, and that incurred during a failure event under which the state or controller bound in line 5 of \Cref{alg: ce with exploration} are violated. To ensure the failure event occurs with a small probability, we make the following assumption on the state and controller bounds, which uses the shorthand $\Psi_{B^\star} \triangleq \max\curly{1,\norm{B^\star}}$.

\begin{assumption}
    \label{asmp: state and controller bounds}
    We assume that $x_b   \geq 400 \norm{P_{K_0}}^2 \Psi_{B^\star} \sigma \sqrt{\dx+\du} \mbox{ and }
        K_b   \geq \sqrt{2 \norm{P_{K_0}}}$.
\end{assumption} 
%Note that the choices for $x_b$ and $K_b$ can be made using crude upper bounds on the system theoretic quantities $P_{K_0}$, $\Psi_{B^\star}$ and $\sigma$ from prior knowledge. 

We make additional assumptions about the remaining arguments supplied as inputs to the algorithm in two cases: one where no additional assumptions about the dynamics are made (\Cref{s: continual exp}), and one where we assume the system structure estimate is such that the initial controller and the optimal controller provide sufficient excitation to identify the unknown parameters (\Cref{s: no exp}).  

\subsection{Certainty equivalent control with continual exploration}
\label{s: continual exp}

To ensure an estimate satisfying the condition in \Cref{lem: CE closeness main body} is attainable, the gap between the model structure estimate and the ground truth cannot be too large, leading to the following assumption. 
\begin{assumption}
    \label{asmp: upper bound on representation error exp}
    Let $\varepsilon$ be as in \Cref{lem: CE closeness main body} and $K_0$ be an initial stabilizing controller. Define
    % \begin{align*}
    $
        \beta_1 \triangleq C_{\mathsf{bias},1} \sigma^4 \norm{P_{K_0}}^{12} \Psi_{B^\star}^8 \norm{\theta^\star}^2 (\dx+\du) \sqrt{\frac{\dtheta}{\du}}
    $
    % \end{align*}
    for a sufficiently large universal constant $C_{\mathsf{bias},1}$. 
    We assume our representation error satisfies $
        d(\hat \Phi, \Phi^\star) \leq \frac{\varepsilon^2 }{4 \beta_1^2}.$
\end{assumption}
The requirement above arises from the way that misspecification enters our bounds on the estimation error $\norm{\bmat{\hat A & \hat B} - \bmat{A^\star & B^\star}}_F^2$. See the definition of $\calE_{\mathsf{est},1}$ in \Cref{s: proof sketch}.  

In this setting, we run \Cref{alg: ce with exploration} with exploratory inputs injected to ensure identifiability of the unknown parameters. Doing so provides the regret guarantees in the following theorem.
\begin{theorem}
    \label{thm: regret bound naive exploration}
     Consider applying \Cref{alg: ce with exploration} with initial stabilizing controller $K_0$ for $T = \tau_1 2^{k_{\fin}-1}$ timesteps for some positive integers $k_{\fin}$, and $\tau_1$. %=\tau_{\mathsf{warm\,up}} \log^2 T$ for some $\tau_{\mathsf{warm\,up}}\in\R$.
     %, where $\tau_{\mathsf{warm\,up}}$ satisfies
    % $$\tau_{\mathsf{warm\,up}}\geq  C_{\mathsf{warm\,up}} \sigma^4 \norm{P_{K_0}}^3\paren{ \Psi_{B^\star}^2 \dx(\dx+\du) \vee x_b^2 \vee \frac{\log \frac{1}{\norm{P^\star}}}{\log\paren{1- \frac{1}{ \norm{P^\star}}}} \vee \paren{\frac{\sqrt{\dtheta \du} }{\varepsilon}}^2},$$
    % for some sufficiently large universal constant $C_{\mathsf{warm\,up}}$. 
    Suppose that for some $\exploration \geq 1$, the exploration sequence is given by
          $\sigma_k^2 = \max\curly{\frac{\sqrt{\du/\dtheta}}{\sqrt{\epoch 2^{k-1}}} , \exploration d(\hat \Phi, \Phi^\star)^{1/2}} \, \forall k \geq 1.$\footnote{The $\gamma$ allows the sequence to be defined with a bound on the level of misspecification, rather than precise knowledge.} Suppose the state bound $x_b$ and the controller bound $K_b$ satisfy \Cref{asmp: state and controller bounds} and that $\hat \Phi $ satisfies \Cref{asmp: upper bound on representation error exp}. Let $\varepsilon$ be as in \Cref{lem: CE closeness main body}.
    There exists a universal positive constant $C_{\mathsf{warm\,up}}$ such that if $\tau_1 = \tau_{\mathsf{warm\,up}} \log^2 T$ for 
    $$\tau_{\mathsf{warm\,up}}\geq  C_{\mathsf{warm\,up}} \sigma^4 \norm{P_{K_0}}^3\max\curly{ \Psi_{B^\star}^2 (\dx+\du), x_b^2, -\log_{\paren{1- \frac{1}{ \norm{P^\star}}}} \norm{P^\star}, \paren{\sqrt{\dtheta \du/\varepsilon} }^2},$$
    then the expected regret satisfies
     %Fix some $\delta \in (0,\frac{1}{2})$ and $T\in\mathbb  N$ such that $T \geq \frac{1}{\delta}$. With probability at least $1-\delta$, the regret satisfies 
     \begin{align*}
         \E \brac{ \mathbf{R}_T }  \leq  c_0 \log^2(T) + c_1 \sqrt{\dtheta \du} \sqrt{T} \log T + c_2 \sqrt{d(\hat \Phi, \Phi^\star)} T,
     \end{align*}
    where
     %$c_0 \lesssim \tau_{\mathsf{warm\,up}}\max\curly{\dx,\du} \norm{P_{K_0}}\Psi_{B^\star}^2 + x_b^2   \norm{P^\star} +    K_b^2 \paren{1 +\norm{P_{K_0}} \norm{\theta^\star}^2},$
     $c_0 = \texttt{poly}(\dx, \du,\norm{P_{K_0}}, \Psi_{B^\star},\tau_{\mathsf{warm\,up}}, x_b, K_b)$, 
     % \Bruce{XXX}\\  
         % $c_1 \lesssim    \norm{P^\star}^8  \norm{P_{K_0}} \Psi_{B^\star}^2,$ 
         $c_1 =\texttt{poly}(\norm{P_{K_0}}, \Psi_{B^\star}, \sigma)$
         and
         % $c_2 \lesssim\du \norm{P^\star} \Psi_{B^\star}^2 + \norm{P^\star}^8 \sigma^4 \norm{P_{K_0}}^{12} \Psi_{B^\star}^8 \norm{\theta^\star} \dx(\dx+\du) \sqrt{\frac{\dtheta}{\du}}.$ 
          $c_2 = \texttt{poly}(\du, \dx, \dtheta, \norm{P_{K_0}},\Psi_{B^\star},  \norm{\theta^\star}, \sigma, \gamma).$
         %\Bruce{Just say poly().}
\end{theorem}

The constants $c_0$ and $c_2$ in the above bound depend on system dimensions, system-theoretic quantities, and algorithm parameters including the state and controller bounds, the initial epoch length, and the initial controller. In contrast, the constant $c_1$, does not depend on system dimension. It is presented as such to emphasize that the dimensional dependence of the order $\sqrt{T}$ term is $\sqrt{d_{\theta} \du}$. This elucidates the dependence on the system and parameter dimensions in the regime where the $\sqrt{T}$ term is dominant. Consider the result in the absence of misspecification: $d(\hat \Phi, \Phi^\star) = 0$. In this case, the dominant term grows with $\sqrt{\dtheta \du}\sqrt{T} \log T$. As long as $\dtheta \leq \dx\du$, this is smaller than the dependence of $\sqrt{\du^2 \dx}$ which appears in the lower bounds for the regret of learning to control a system with entirely unknown $A^\star$ and $B^\star$ \citep{simchowitz2020naive}. If the misspecification is nonzero, then the regret bound incurs an additional term that grows linearly with $T$. However, as long as $d(\hat \Phi, \Phi^\star)$ is sufficiently small, there exists a regime of $T$ for which the $\sqrt{T}$ term dominates, and using the misspecified basis provides a benefit over learning from scratch. 

\ifnum\value{l4dc}>0{}\else{One could modify the procedure to simultaneously use the collected data to estimate $A^\star$ and $B^\star$ from scratch at every epoch, while maintaining confidence sets for these estimates \citep{abbasi2011regret}. If the high confidence set is smaller than the bound on the level of misspecification, the learner could switch to relying upon the estimates of learning from scratch. It may therefore be possible to eliminate the term in the regret that is linear in $T$, while maintaining the benefit of the prior knowledge for small $T$. We leave investigation of such procedures to future work. }\fi

\subsection{Certainty equivalent control without additional exploration}
\label{s: no exp}

In this section, we analyze the regret attained under the additional assumption that the process noise fully excites the relevant modes of the system under $K^\star$ and $K_0$. This may be guaranteed as follows. 
\begin{assumption}
    \label{asmpt: persisent excitation}
    Let $\hat \Phi$ be the estimate for dynamics representation, and let $\alpha$ be a number satisfying $\alpha \geq \frac{1}{3 \norm{P^\star}^{3/2}}$. We assume that 
    $\lambda_{\min}\paren{\hat\Phi^\top \paren{\bmat{I \\ K} \bmat{I \\K}^\top \otimes I_{\dx}} \hat \Phi} \geq \alpha^2$ for $K = K_0, K^\star$. 
    % Assume that 
    % % \begin{align*}
    %     $\min_{v: \norm{v}=1} \norm{\sum_{i=1}^{\dtheta} v_i(\hat \Phi^A_i + \hat \Phi^B_i K)}_F^2 \geq \alpha^2 \quad$ for  $K = K_0, K^\star,$
    %    % \min_{i \in [\dtheta]} \norm{\hat A_i + \hat B_i K^\star}_F^2 \geq \alpha^2 \quad\mbox{and} \quad
    %    % \min_{i \in [\dtheta]} \norm{\hat A_i + \hat B_i K_0}_F^2 \geq \alpha^2.
    % % \end{align*}
    % where $\bmat{\hat \Phi^A_i & \hat \Phi^B_i} = \VEC^{-1}\hat \Phi_i$, and $\hat \Phi_i$ is the $i^{\mathsf{th}}$ column of $\hat \Phi$.
\end{assumption}

The above assumption captures a setting where playing either the initial controller $K_0$ or the optimal controller $K^\star$ provides persistence of excitation without any exploratory input. This can be seen by noting that the matrix $\hat\Phi^\top \paren{\bmat{I \\ K} \bmat{I \\K}^\top \otimes I_{\dx}} \hat \Phi$ is a lower bound (in Loewner order) for the covariance matrix formed by taking the expectation of $\Lambda/t$ in \Cref{alg: least squares} when $u_s = K x_s$. 

Under the above assumption, we may run \Cref{alg: ce with exploration} without an additional exploratory input injected. As in \Cref{s: continual exp}, we require that the representation error is small enough to guarantee the closeness condition in \Cref{lem: CE closeness main body} may be satisfied with our estimated model. 
\begin{assumption}
    \label{asmp: upper bound on representation error no exp}
    Let $\varepsilon$ be as in \Cref{lem: CE closeness main body}, $K_0$ be a stabilizing controller, and $\alpha$ be a positive number such that \Cref{asmpt: persisent excitation} holds. We assume our representation error satisfies $
        d(\hat \Phi, \Phi^\star) \leq \sqrt{\frac{\varepsilon }{2 \beta_2}}$, where $
        \beta_2 \triangleq C_{\mathsf{bias},2} \frac{\varepsilon \norm{P_{K_0}}^9 \Psi_{B^\star}^8 \norm{\theta^\star}^2 (\dx+\du)}{\dtheta \min\curly{\alpha^2, \alpha^4}}$ and $C_{\mathsf{bias},2}$ is a sufficiently large universal constant . 
\end{assumption}
This requirement again arises from dependence of the estimation error bounds on the misspecification.  See the definition of $\calE_{\mathsf{est},2}$ in \Cref{s: proof sketch}.  
% Note that due to the square root in the bound on $d(\hat \Phi, \Phi^\star)$ in the above assumption this condition is much less stringent than that in \Cref{asmp: upper bound on representation error exp}. 
Under these assumptions, our regret bound may be improved to that in the following theorem.

\begin{theorem}
    \label{thm: regret bound no exploration}
     Consider applying \Cref{alg: ce with exploration} with initial stabilizing controller $K_0$ for $T = \tau_1 2^{k_{\fin}}$ timeteps for some positive integers $k_{\fin}$, and $\tau_1$. %, where $\tau_1 =\tau_{\mathsf{warm\,up}} \log T $ for  for some sufficiently large universal constant $C_{\mathsf{warm\,up}}$. 
     Additionally suppose the exploration sequence is zero for all time: $\sigma_k^2 = 0$ for $k = 1, \dots, k_{\fin}$. Suppose the state bound $x_b$ and the controller bound $K_b$ satisfy \Cref{asmp: state and controller bounds}, and that $\hat \Phi$ satisfies \Cref{asmpt: persisent excitation} and \Cref{asmp: upper bound on representation error no exp}. Let $\varepsilon$ be as in \Cref{lem: CE closeness main body}. There exists a positive universal constant $C_{\mathsf{warm\,up}}$ such that if $\tau_1 = \tau_{\mathsf{warm\,up}} \log T$, for 
     $$\tau_{\mathsf{warm\,up}}\geq  C_{\mathsf{warm\,up}}\sigma^4 \norm{P_{K_0}}^3 \Psi_{B^\star}^2 \max\curly{ (\dx+\du), x_b^2,-\log_{\paren{1- \frac{1}{2 \norm{P^\star}}}}\norm{P^\star}, \dtheta/(2 \varepsilon \alpha^2)},$$
     then the expected regret satisfies
     \begin{align*}
         \E \brac{\mathbf{R}_T} \leq  c_1 \log^2(T) + c_2 d(\hat \Phi, \Phi^\star)^2 T,
     \end{align*}
     where 
     % $c_1 \lesssim \tau_{\mathsf{warm\,up}} \max\curly{\dx,\du} \norm{P_{K_0}}\Psi_{B^\star}^2 +  \norm{P^\star}^8 \frac{ \sigma^2 \dtheta }{  \alpha^2} + 2  x_b^2  \norm{P^\star} +  K_b^2 \paren{1 +\norm{P_{K_0}} \norm{\theta^\star}^2}$ 
     $c_1 = \texttt{poly}(\dx,\du, \dtheta,\norm{P_{K_0}},\Psi_{B^\star},\norm{\theta^\star}, \sigma,  \alpha^{-1},\tau_{\mathsf{warm\,up}},K_b, x_b)$ 
     and \\
     %$c_2 \lesssim \frac{\norm{P_{K_0}}^9 \Psi_{B^\star}^8 \norm{\theta^\star}^2 \dx(\dx+\du)}{\dtheta \min\curly{\alpha^2, \alpha^4}}.$
     $c_2 = \texttt{poly}( \dx,\du,\dtheta,\norm{P_{K_0}}, \Psi_{B^\star}, \norm{\theta^\star}, \alpha^{-1}).$ 
\end{theorem}

When the misspecification is zero, the expected regret grows with $\log^2 T$. Prior work \citep{cassel2020logarithmic, jedra2022minimal} has shown that such rates are possible if either $A^\star$ or $B^\star$ are known to the learner. %These situations can be embedded into the setting of \eqref{eq: dynamics} by instead considering an affine model of the unknown parameters, which is then subtracted off when solving the least squares problem using \Cref{alg: least squares}. 
See \ifnum\value{l4dc}>0{\cite{lee2023nonasymptotic} }\else{\Cref{s: log regret with known A or B} }\fi for details about how to obtain logarithmic regret in these settings using \Cref{thm: regret bound no exploration}. The above result expands on prior work by generalizing conditions for prior knowledge that are sufficient to achieve logarithmic regret. 

% \footnote{We choose not to include the affine component in our analysis, because it obfuscates the question of misspecification in the exposition. In particular, this affine term would need to be sufficiently close to ground truth in matrix norm, rather than subspace distance. \Bruce{Write an appendix section about generalizations, and include this result there. Also include the discussion about the hassibi paper.}}. 
% When this embedding is performed, we can turn to \Cref{asmpt: persisent excitation} to determine whether the guarantees of \Cref{thm: regret bound no exploration} are possible. What we find is that if $A$ is known to the learner, then the quantity in \Cref{asmpt: persisent excitation} can only be zero if the matrix $K^\top$ has a nontrivial nullspace, which corresponds to the condition $K K^\top \succ 0$ in \cite{jedra2022minimal}. When $B$ is known to the learner, the condition in \Cref{asmpt: persisent excitation} is always satisfied, recovering the situation from \cite{jedra2022minimal}. See \Cref{s: generalizations} for further discussions about when $\log T$ regret is attainable. 

With misspecification present, the above theorem has a term growing linearly $T$. In contrast to \Cref{thm: regret bound naive exploration}, the coefficient for this term is proportional to the level of misspecification squared, which is smaller than the square root dependence in \Cref{thm: regret bound naive exploration}. This result shows that if some coarse system knowledge depending on a few unknown parameters is available in advance and the unknown parameters are easily identifiable in the sense of \Cref{asmpt: persisent excitation}, then there exists a substantial regime of $T$ for which the regret incurred is much smaller than that attained by learning to control the system from scratch with fully unknown $A^\star$ and $B^\star$ (where the regret scales as $\sqrt{T}$). 

\subsection{Proof sketch}
\label{s: proof sketch}
Our main result proceeds by first defining a success events for which the certainty equivalent control scheme never aborts, and generates dynamics estimates $\bmat{\hat A_k & \hat B_k}$ which are sufficiently close to the true dynamics $\bmat{A^\star & B^\star}$ at all times.\ifnum\value{l4dc}>0{}\else{\footnote{This will suffice for $(\hat A_k, \hat B_k)$ to be  stabilizable, and for the resulting controller $\hat K_k$ to stabilize $(A^\star, B^\star)$.}}\fi The success events for \Cref{s: continual exp} and \Cref{s: no exp} are $\calE_{\mathsf{success},1} =  \calE_{\mathsf{bound}} \cap  \calE_{\mathsf{est},1}$ and $\calE_{\mathsf{success},2}= \calE_{\mathsf{bound}} \cap \calE_{\mathsf{est},2}$ respectively, where
\begin{align*}
    \calE_{\mathsf{bound}} &= \curly{\norm{x_t}^2  \leq x_b^2 \log T \quad \forall t=1,\dots , T} \cap \curly{\norm{\hat K_k} \leq K_b \,, \forall k=1,\dots, k_{\fin}}, \\
    \calE_{\mathsf{est},1} &= \curly{\norm{\bmat{\hat A_k & \hat B_k} - \bmat{A^\star & B^\star}}_F^2 \leq C_{\mathsf{est},1} \frac{\sigma^2 \sqrt{\dtheta \du} \norm{P_{K_0}}}{\sqrt{\tau_{k}}} \log T + \beta_1 \sqrt{d(\hat \Phi, \Phi^\star)}}, \\
     \calE_{\mathsf{est},2} &= \curly{\norm{\bmat{\hat A_k & \hat B_k} - \bmat{A^\star & B^\star}}_F^2 \leq C_{\mathsf{est},2} \frac{\sigma^2  \dtheta   }{ \tau_{k} \alpha^2 }\log T+ \beta_2 d(\hat \Phi, \Phi^\star)^2 },
\end{align*}
and $C_{\mathsf{est},1}$ and $C_{\mathsf{est},2}$ are positive universal constants.

We use the success events to decompose the expected regret as in \cite{cassel2020logarithmic}: 
% \begin{align*}
$
    \E \brac{\mathbf{R}_T} = R_1 + R_2+ R_3 - T\calJ(K^\star), %\underbrace{\sum_{t=1}^{\epoch} x_t^\top Q x_t  + u_t^\top R u_t - \epoch \calJ^\star}_{\mbox{First round regret}} + \underbrace{\sum_{t=\epoch+1}^T \paren{ x_t^\top Q x_t  + u_t^\top R u_t} - (T-\epoch) \calJ^\star.}_{\mbox{Subsequent round regret}}
$
%\end{align*}
where for $\calE_{\mathsf {success}} = \calE_{\mathsf{success},1}$ or $\calE_{\mathsf {success}} =\calE_{\mathsf{success}, 2}$,
\begin{align}
    \label{eq: regret decomposition}
    R_1 = \E \brac{\mathbf{1}(\calE_{\mathsf{success}}) \sum_{k=2}^{k_\fin} J_k}, \quad R_2 =  \E \brac{\mathbf{1}(\calE_{\mathsf{success}}^c) \sum_{t=\tau_1 + 1}^{T} c_t}, \quad  \mbox{ and } \quad R_3 = \E \brac{\sum_{t=1}^{\tau_1} c_t},
\end{align}
are the costs due to success, failure, and the first epoch respectively. Here, $J_k$ are the epoch costs defined as $J_k = \sum_{t=\tau_k}^{\tau_{k+1}} c_t$. In the settings of both \Cref{s: continual exp} and \Cref{s: no exp}, $R_3$ is given by $\tau_1 \trace\paren{P_{K_0} (I+ \sigma_1^2 B^\star (B^\star)^\top)}$, while $R_2$ is controlled using the upper bounds on the state and controller to obtain a bound on the cost, which is then multiplied by the small probability of the failure event. To control $R_1$, we show that under the success event, the closeness condition \Cref{lem: CE closeness main body} is satisfied. As a result, the cost of each epoch $J_k$ is $(\tau_k-\tau_{k-1}) \calJ(K^\star)$ in addition to a term proportional to $(\tau_k-\tau_{k-1})\paren{ \norm{\bmat{\hat A_k & \hat B_k} - \bmat{A^\star & B^\star}}_F^2 + \sigma_k^2}.$ Using the estimation error bounds of events $\calE_{\mathsf{est},1}$ and $\calE_{\mathsf{est},2}$ along with the choices for $\sigma_k^2$ in the two settings, we find that the quantity $R_1 - T\calJ(K^\star)$ is proportional to $\sqrt{T} \log T + \sqrt{d(\hat \Phi, \Phi^\star)} T$ in the setting of \Cref{s: continual exp}, and $ \log^2{T} + d(\hat \Phi, \Phi^\star)^2 T$ in the setting of \Cref{s: no exp}. Combining terms provides the expected regret bounds in Theorems~\ref{thm: regret bound naive exploration} and \ref{thm: regret bound no exploration}\ifnum\value{l4dc}>0{}\else{\footnote{High probability estimation error bounds are proved in \Cref{s: estimation error bounds}. These are used to establish high probability bounds on the success events in \Cref{s: success event bounds}. The bounds on $R_1$, $R_2$ and $R_3$ are established in \Cref{s: r1 r2 r3 bounds}. }}\fi.

\section{Numerical Example}
\label{s: numerical example} 

To validate the trends predicted by our bounds, we run \Cref{alg: ce with exploration} on the system \eqref{eq: dynamics} where $A^\star$ and $B^\star$ 
are  obtained by linearizing and discretizing the cartpole dynamics defined by the equations
% \begin{align*}
$(M + m) \ddot{x} + m\ell(\ddot{\theta}\cos(\theta) - \dot{\theta}^2\sin(\theta)) = u,$ and $
m(\ddot{x}\cos(\theta) + \ell\ddot{\theta} - g\sin(\theta)) = 0,$
% \end{align*}
for cart mass $M=1$, pole mass $m=1$, pole length $\ell=1$, and gravitiy $g=1$. Discretization uses Euler's approach with stepsize $0.25$. The disturbance signal is generated as $w_t \sim \calN(0, 0.01 I)$. 

% We let $Q = I$ and $R = I$. Then the optimal LQR controller for this system is given by $K_\star = 
\sloppy We consider various inputs for the representation estimate $\hat \Phi$ and the exploration sequence $\sigma_1^2, \dots, \sigma_{k_{\fin}}^2$. The remaining parameters are discussed in \ifnum\value{l4dc}>0{\cite{lee2023nonasymptotic}. }\else{\Cref{s: additional experimental details}. Note that the requirements on the size of the misspecification and on the length of the initial epoch in Theorems~\ref{thm: regret bound naive exploration} and \ref{thm: regret bound no exploration} are quite stringent. In the numerical examples, we choose values that do not satisfy these assumptions. However, the predicted trends still hold. }\fi 

\begin{wrapfigure}[10]{r}{0.43\textwidth}
  \vspace{-14pt}
  \includegraphics[width=0.42\textwidth]{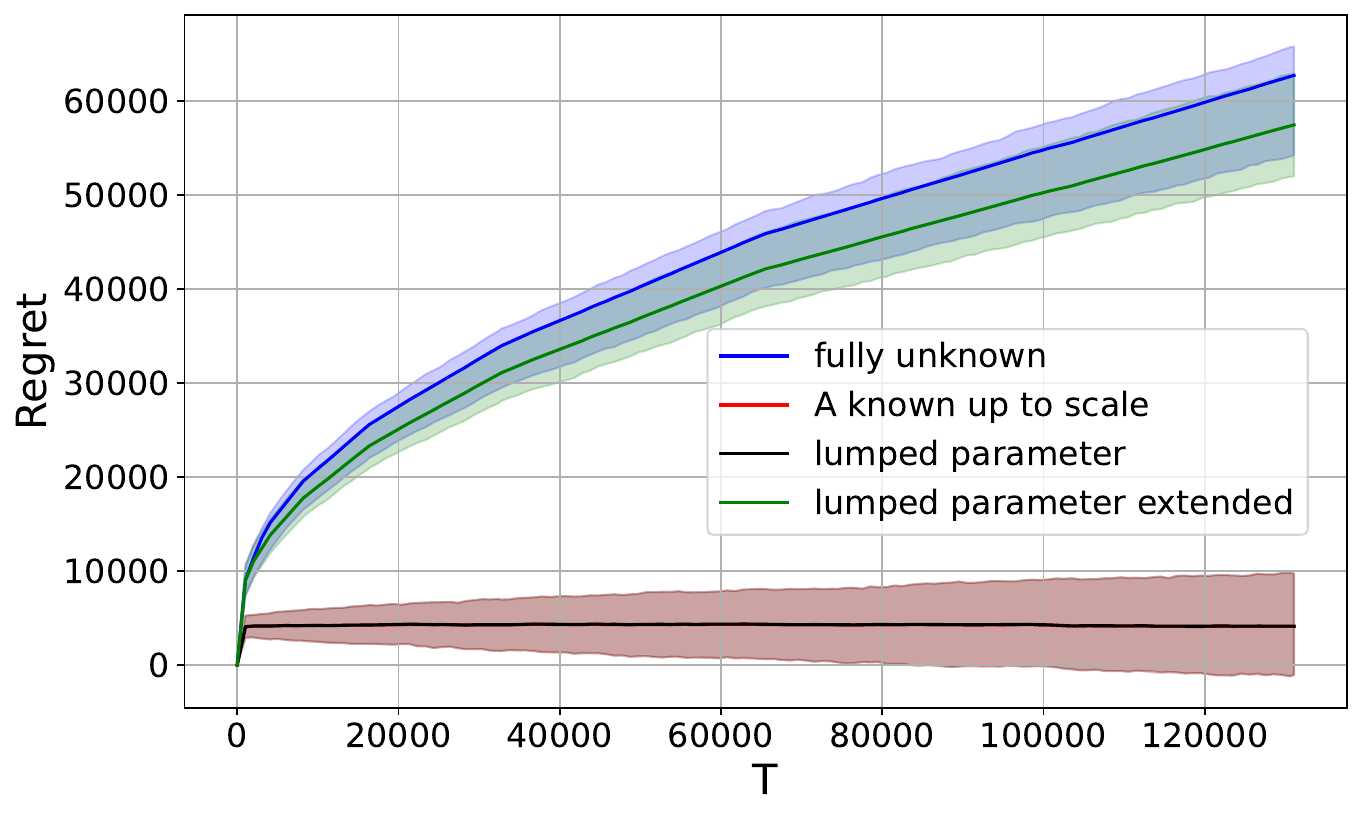}
  \vspace{-16pt}
  \caption{Regret of \Cref{alg: ce with exploration} with various choices for $\hat\Phi$. }
  \label{fig: no misspecification}
\end{wrapfigure}

\paragraph{No Misspecification: } In \Cref{fig: no misspecification}, we plot the regret of \Cref{alg: ce with exploration}  in the absence of misspecification, i.e. $d(\hat\Phi, \Phi^\star)=0$. We consider several instances for the representation: one for fully unknown $A^\star$ and $B^\star$, one which encodes a setting where the $A^\star$ matrix is known up to an unknown scaling, and one which captures a lumped parameter model where the discretized and linearized cartpole structure is known up to scale, but the values of the entries which vary with cart mass, pole mass, and pole length are unknown. We additionally consider extending the lumped parameter representation by adding a basis vector that ensures the condition in \Cref{asmpt: persisent excitation} is violated. For the representation capturing fully unknown $A^\star$ and $B^\star$, and the extended lumped parameter representation, the condition in \Cref{asmpt: persisent excitation} is not satisfied, so we run \Cref{alg: ce with exploration} with exploration noise scaling as $\sigma_k^2 \propto \frac{1}{\sqrt{2^k}}$, and incur $\sqrt{T}$ regret, as predicted by \Cref{thm: regret bound naive exploration}.  The extended lumped parameter representation incurs regret at a slower rate than the setting when the system is fully unknown. This is predicted by \Cref{thm: regret bound naive exploration} due to the fact that the extended lumped parameter model has $\dtheta = 6 < 20 = \dx(\dx+\du)$, so the coefficient on the $\sqrt T$ term is smaller. For the remaining settings, \Cref{asmpt: persisent excitation} is satisfied, so we run the algorithm with no additional exploration and incur logarithimic regret, as predicted in \Cref{thm: regret bound no exploration}.

\begin{figure}
    \centering
    \subfigure[Continual exploration]
        {\label{fig: continual exp}\includegraphics[width=0.42\textwidth]{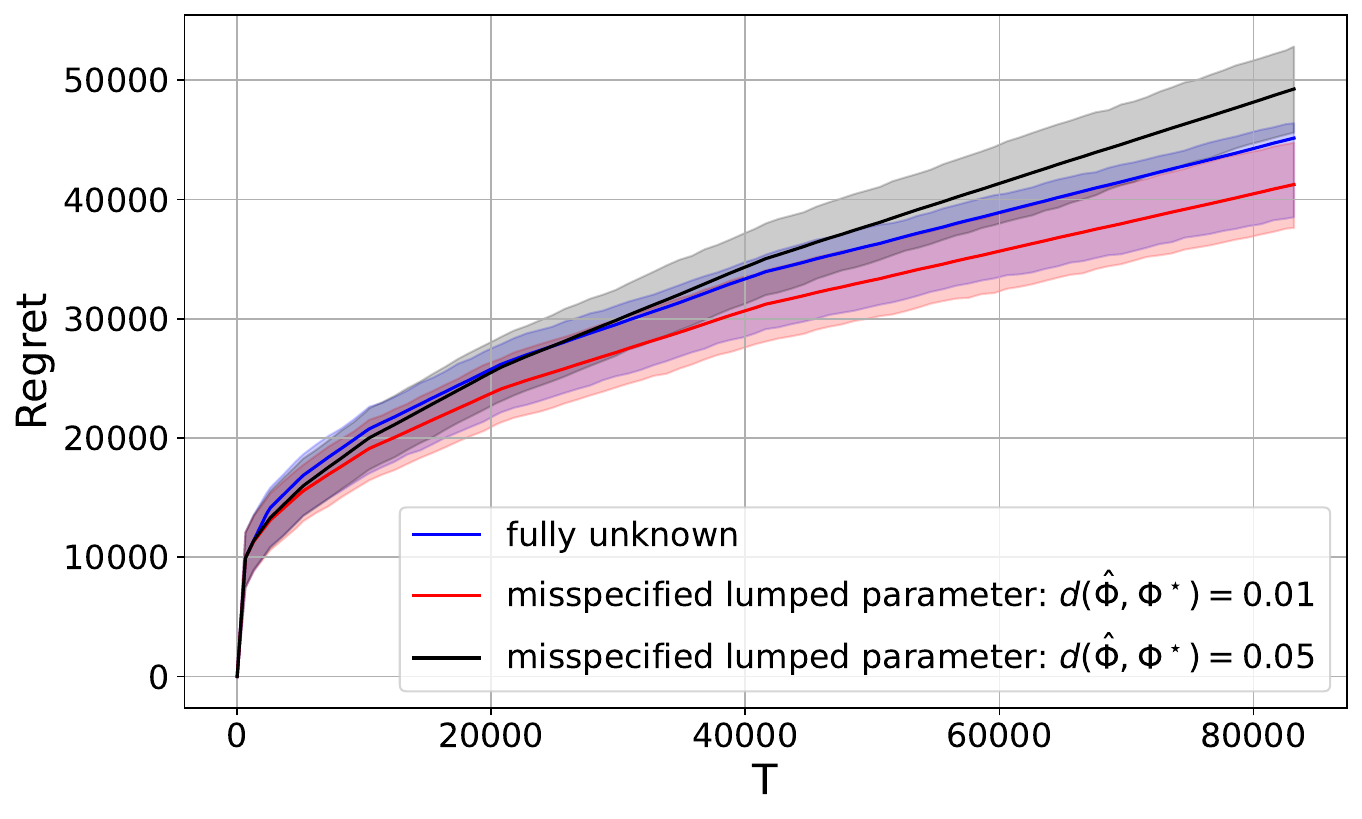}}
    \subfigure[No exploration]{\label{fig: no exp}
        \includegraphics[width=0.42\textwidth]{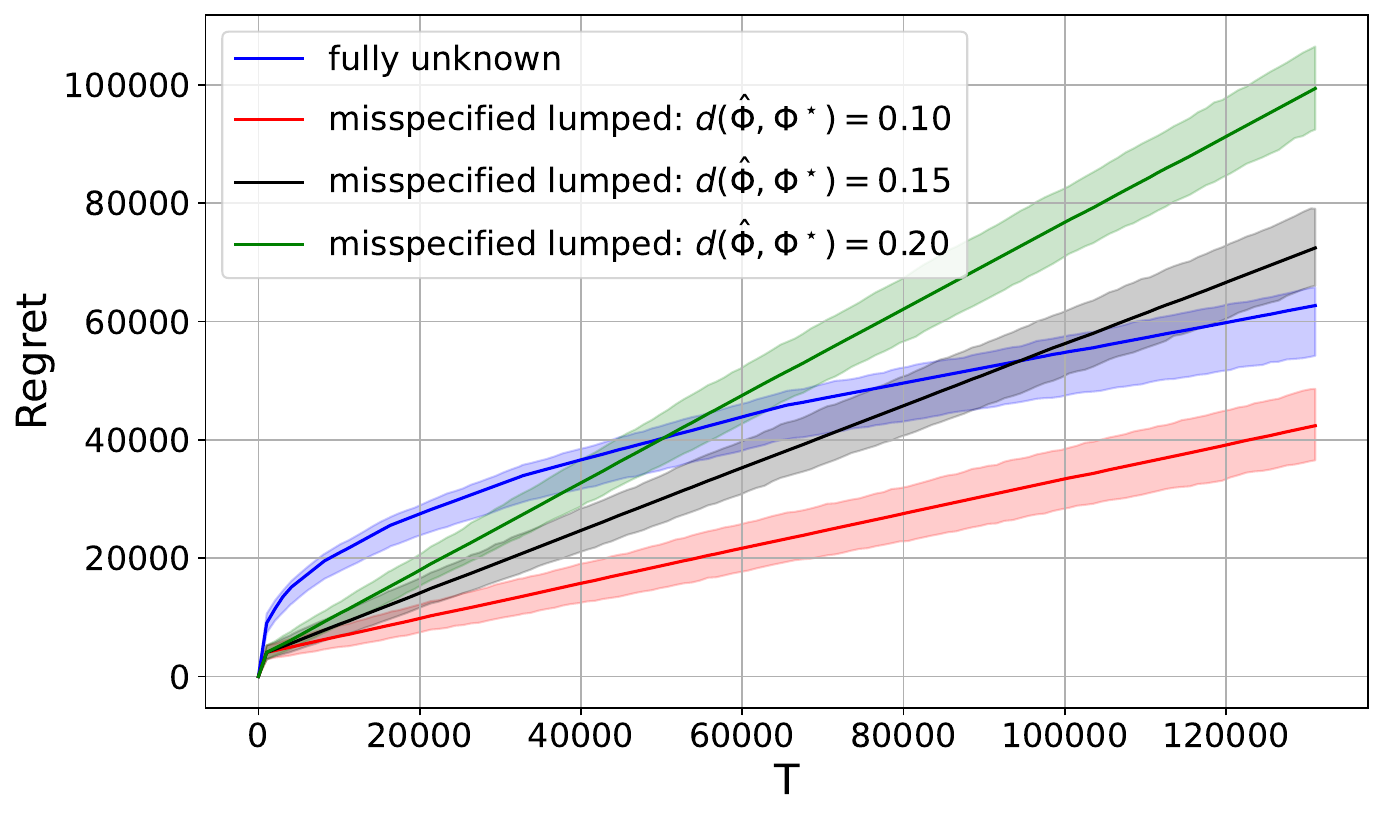}
    }
    \vspace{-12pt}
    \caption{We plot the regret of \Cref{alg: ce with exploration} with $\hat \Phi$ describing a lumped parameter model (right), or a lumped parameter model and extended such that the condition \Cref{asmpt: persisent excitation} is violated (left). In both settings the representations are perturbed, resulting in a misspecification between the true representation $\Phi^\star$ and the representation estimate $\hat \Phi$. The regret is compared to that incurred by running \Cref{alg: ce with exploration} with a fully unknown $A^\star$ and $B^\star$.}
    \vspace{-6pt}
    \label{fig: artificial misspecification}
\end{figure}

\paragraph{Artificial Misspecification: } In \Cref{fig: artificial misspecification}, we compare the regret from the fully unknown setting to the regret with a misspecified lumped parameter representation. \Cref{fig: continual exp}  considers the lumped parameter representation that is extended such that \Cref{asmpt: persisent excitation} is violated. Therefore the learner must continually inject noise to the system in order to explore. We artificially create misspecification by adding small perturbations to the true representation such that $d(\hat \Phi, \Phi^\star) > 0$. We see that in the low data regime, the regret incurred is less than that incurred when the model is fully unknown. When the misspecification level is $0.05$, the bias in the identification due to the misspecification causes the regret to rapidly overtake the regret from the fully unkown setting. When the misspecification is small, the regret remains less than that from the fully unknown setting for the entire horizon of $T$ values that are plotted. \Cref{fig: no exp} considers the lumped parameter setting without the extension, for which \Cref{asmpt: persisent excitation} is satisfied. As in \Cref{fig: continual exp}, we add a perturbation to the representation to create misspecification. In this setting, we consider much larger perturbations, such that $d(\hat \Phi, \Phi^\star)$ is $0.1$, $0.15$, or $0.20$. For all three such situations, the regret begins much smaller than that of the fully unknown model, but overtakes it as $T$ becomes large. 

\begin{wrapfigure}[10]{r}{0.43\textwidth}
  \vspace{-12pt}
  \includegraphics[width=0.42\textwidth]{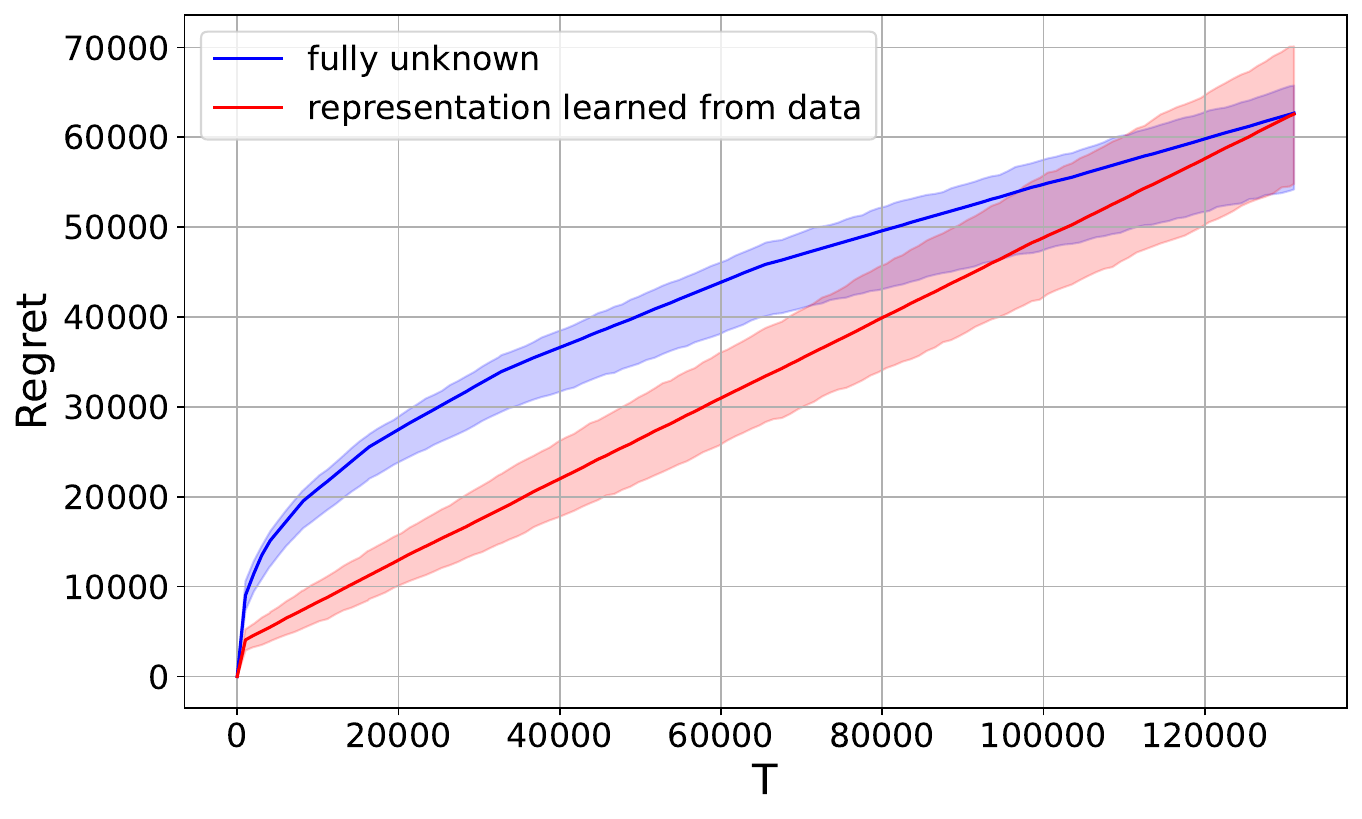}
  \vspace{-18pt}
  \caption{Regret of Alg. \ref{alg: ce with exploration} with $\hat \Phi$ learned offline from related systems.}
  \label{fig: learned representation}
\end{wrapfigure}\paragraph{Learned Representation: } In \Cref{fig: learned representation}, we consider a setting motivated by multi-task learning, in which a representation is learned using offline data from several systems related to the system of interest. In particular, we collect trajectories of length $1200$ from five discretized and linearized cartpole systems generated with various values of the parameters $(M, m,\ell)$. The resulting data is in turn used to fit a a representation $\hat \Phi$.\footnote{See \ifnum\value{l4dc}>0{\cite{lee2023nonasymptotic} }\else{\Cref{s: additional experimental details} }\fi for details.} Once the representation is obtained, we run \Cref{alg: ce with exploration} in the absence of exploratory input. We see that the regret incurred is much lower than the setting in which the dynamics are fully unknown for the small data regime, but overtakes it as $T$ becomes large. This aligns with the results from the artificial misspecification experiment. By computing the distance between the learned and true lumped parameter representations, we find that $d(\hat \Phi, \Phi^\star) = 0.2041$. 

\ifnum\value{l4dc}>0{}\else{It is worth noting that, although our focus was on the statistical benefit of using a pre-trained representation, there are clear computation benefits as well. In particular, each epoch solves a linear regression problem of dimension $d_{\theta}$ as opposed to $\dx (\dx + \du)$ in the setting of learning from scratch. The computational complexity of linear regression is $d^3$. Therefore, in our examples where $d_{\theta}=5$ and $\dx (\dx +\du) = 20$, our estimation procedure yields a $64$ times computational speed up. }\fi
\section{Conclusion}

We studied adaptive LQR in the presence of misspecification between the learner's simple model structure for the dynamics, and the true dynamics. Our proposed algorithm performs well in both experiments and theory as long as this misspecification is sufficiently small. Our analysis shows a phase shift in the problem that depends on whether the learner's prior knowledge  enables identification of the unknown parameters without additional exploration, thus allowing regret that is logarithmic in $T$. %A consequence of our analysis that may be of independent interest is an improved understanding of when logarithmic regret is possible in adaptive linear quadratic control. 
There are many interesting avenues for future work, including extending the analysis to settings where a collection of nonlinear basis functions for the dynamics are misspecified, e.g., by modifying the model studied in \cite{kakade2020information}. Another direction is to analyze the setting where the learner's simple model is updated online by sharing data from a collection of related systems, as in recent work on federated learning in dynamical systems \citep{wang2023model}.

\acks{We thank Leo Toso, Olle Kjellqvist, Thomas Zhang, and Ingvar Ziemann for helpful comments and discussions. BL and NM gratefully acknowledge support from NSF Award SLES 2331880, NSF CAREER award ECCS 2045834, and NSF EECS 2231349.}

\bibliography{refs}

\begin{thebibliography}{39}
\providecommand{\natexlab}[1]{#1}
\providecommand{\url}[1]{\texttt{#1}}
\expandafter\ifx\csname urlstyle\endcsname\relax
  \providecommand{\doi}[1]{doi: #1}\else
  \providecommand{\doi}{doi: \begingroup \urlstyle{rm}\Url}\fi

\bibitem[Abbasi-Yadkori and Szepesv{\'a}ri(2011)]{abbasi2011regret}
Yasin Abbasi-Yadkori and Csaba Szepesv{\'a}ri.
\newblock Regret bounds for the adaptive control of linear quadratic systems.
\newblock In \emph{Proceedings of the 24th Annual Conference on Learning
  Theory}, pages 1--26. JMLR Workshop and Conference Proceedings, 2011.

\bibitem[{\AA}str{\"o}m and Wittenmark(2013)]{aastrom2013adaptive}
Karl~J {\AA}str{\"o}m and Bj{\"o}rn Wittenmark.
\newblock \emph{Adaptive control}.
\newblock Courier Corporation, 2013.

\bibitem[{\AA}str{\"o}m and Wittenmark(1973)]{aastrom1973self}
Karl~Johan {\AA}str{\"o}m and Bj{\"o}rn Wittenmark.
\newblock On self tuning regulators.
\newblock \emph{Automatica}, 9\penalty0 (2):\penalty0 185--199, 1973.

\bibitem[Cassel et~al.(2020)Cassel, Cohen, and Koren]{cassel2020logarithmic}
Asaf Cassel, Alon Cohen, and Tomer Koren.
\newblock Logarithmic regret for learning linear quadratic regulators
  efficiently.
\newblock In \emph{International Conference on Machine Learning}, pages
  1328--1337. PMLR, 2020.

\bibitem[Cederberg et~al.(2022)Cederberg, Hansson, and
  Rantzer]{cederberg2022synthesis}
Daniel Cederberg, Anders Hansson, and Anders Rantzer.
\newblock Synthesis of minimax adaptive controller for a finite set of linear
  systems.
\newblock In \emph{2022 IEEE 61st Conference on Decision and Control (CDC)},
  pages 1380--1384. IEEE, 2022.

\bibitem[Cohen et~al.(2019)Cohen, Koren, and Mansour]{cohen2019learning}
Alon Cohen, Tomer Koren, and Yishay Mansour.
\newblock Learning linear-quadratic regulators efficiently with only $\sqrt t$
  regret.
\newblock In \emph{International Conference on Machine Learning}, pages
  1300--1309. PMLR, 2019.

\bibitem[Dasari et~al.(2019)Dasari, Ebert, Tian, Nair, Bucher, Schmeckpeper,
  Singh, Levine, and Finn]{dasari2019robonet}
Sudeep Dasari, Frederik Ebert, Stephen Tian, Suraj Nair, Bernadette Bucher,
  Karl Schmeckpeper, Siddharth Singh, Sergey Levine, and Chelsea Finn.
\newblock Robonet: Large-scale multi-robot learning.
\newblock \emph{arXiv preprint arXiv:1910.11215}, 2019.

\bibitem[Dean et~al.(2018)Dean, Mania, Matni, Recht, and Tu]{dean2018regret}
Sarah Dean, Horia Mania, Nikolai Matni, Benjamin Recht, and Stephen Tu.
\newblock Regret bounds for robust adaptive control of the linear quadratic
  regulator.
\newblock \emph{Advances in Neural Information Processing Systems}, 31, 2018.

\bibitem[Devlin et~al.(2018)Devlin, Chang, Lee, and Toutanova]{devlin2018bert}
Jacob Devlin, Ming-Wei Chang, Kenton Lee, and Kristina Toutanova.
\newblock Bert: Pre-training of deep bidirectional transformers for language
  understanding.
\newblock \emph{arXiv preprint arXiv:1810.04805}, 2018.

\bibitem[Dosovitskiy et~al.(2020)Dosovitskiy, Beyer, Kolesnikov, Weissenborn,
  Zhai, Unterthiner, Dehghani, Minderer, Heigold, Gelly,
  et~al.]{dosovitskiy2020image}
Alexey Dosovitskiy, Lucas Beyer, Alexander Kolesnikov, Dirk Weissenborn,
  Xiaohua Zhai, Thomas Unterthiner, Mostafa Dehghani, Matthias Minderer, Georg
  Heigold, Sylvain Gelly, et~al.
\newblock An image is worth 16x16 words: Transformers for image recognition at
  scale.
\newblock \emph{arXiv preprint arXiv:2010.11929}, 2020.

\bibitem[Du et~al.(2020)Du, Hu, Kakade, Lee, and Lei]{du2020few}
Simon~S Du, Wei Hu, Sham~M Kakade, Jason~D Lee, and Qi~Lei.
\newblock Few-shot learning via learning the representation, provably.
\newblock \emph{arXiv preprint arXiv:2002.09434}, 2020.

\bibitem[Gregory(1959)]{gregory1959proceedings}
PC~Gregory.
\newblock \emph{Proceedings of the Self Adaptive Flight Control Systems
  Symposium}, volume~59.
\newblock Wright Air Development Center, Air Research and Development Command,
  United~…, 1959.

\bibitem[Hazan et~al.(2020)Hazan, Kakade, and Singh]{hazan2020nonstochastic}
Elad Hazan, Sham Kakade, and Karan Singh.
\newblock The nonstochastic control problem.
\newblock In \emph{Algorithmic Learning Theory}, pages 408--421. PMLR, 2020.

\bibitem[Ioannou and Sun(1996)]{ioannou1996robust}
Petros~A Ioannou and Jing Sun.
\newblock \emph{Robust adaptive control}, volume~1.
\newblock PTR Prentice-Hall Upper Saddle River, NJ, 1996.

\bibitem[Jedra and Proutiere(2020)]{jedra2020finite}
Yassir Jedra and Alexandre Proutiere.
\newblock Finite-time identification of stable linear systems optimality of the
  least-squares estimator.
\newblock In \emph{2020 59th IEEE Conference on Decision and Control (CDC)},
  pages 996--1001. IEEE, 2020.

\bibitem[Jedra and Proutiere(2022)]{jedra2022minimal}
Yassir Jedra and Alexandre Proutiere.
\newblock Minimal expected regret in linear quadratic control.
\newblock In \emph{International Conference on Artificial Intelligence and
  Statistics}, pages 10234--10321. PMLR, 2022.

\bibitem[Kakade et~al.(2020)Kakade, Krishnamurthy, Lowrey, Ohnishi, and
  Sun]{kakade2020information}
Sham Kakade, Akshay Krishnamurthy, Kendall Lowrey, Motoya Ohnishi, and Wen Sun.
\newblock Information theoretic regret bounds for online nonlinear control.
\newblock \emph{Advances in Neural Information Processing Systems},
  33:\penalty0 15312--15325, 2020.

\bibitem[Lale et~al.(2020)Lale, Azizzadenesheli, Hassibi, and
  Anandkumar]{lale2020logarithmic}
Sahin Lale, Kamyar Azizzadenesheli, Babak Hassibi, and Anima Anandkumar.
\newblock Logarithmic regret bound in partially observable linear dynamical
  systems.
\newblock \emph{Advances in Neural Information Processing Systems},
  33:\penalty0 20876--20888, 2020.

\bibitem[Lale et~al.(2022)Lale, Azizzadenesheli, Hassibi, and
  Anandkumar]{lale2022reinforcement}
Sahin Lale, Kamyar Azizzadenesheli, Babak Hassibi, and Animashree Anandkumar.
\newblock Reinforcement learning with fast stabilization in linear dynamical
  systems.
\newblock In \emph{International Conference on Artificial Intelligence and
  Statistics}, pages 5354--5390. PMLR, 2022.

\bibitem[Mania et~al.(2019)Mania, Tu, and Recht]{mania2019certainty}
Horia Mania, Stephen Tu, and Benjamin Recht.
\newblock Certainty equivalence is efficient for linear quadratic control.
\newblock \emph{Advances in Neural Information Processing Systems}, 32, 2019.

\bibitem[Modi et~al.(2021)Modi, Faradonbeh, Tewari, and
  Michailidis]{modi2021joint}
Aditya Modi, Mohamad Kazem~Shirani Faradonbeh, Ambuj Tewari, and George
  Michailidis.
\newblock Joint learning of linear time-invariant dynamical systems.
\newblock \emph{arXiv preprint arXiv:2112.10955}, 2021.

\bibitem[Narendra and Annaswamy(2012)]{narendra2012stable}
Kumpati~S Narendra and Anuradha~M Annaswamy.
\newblock \emph{Stable adaptive systems}.
\newblock Courier Corporation, 2012.

\bibitem[Pe{\~n}a et~al.(2009)Pe{\~n}a, Lai, and Shao]{pena2009self}
Victor~H Pe{\~n}a, Tze~Leung Lai, and Qi-Man Shao.
\newblock \emph{Self-normalized processes: Limit theory and Statistical
  Applications}.
\newblock Springer, 2009.

\bibitem[Rantzer(2021)]{rantzer2021minimax}
Anders Rantzer.
\newblock Minimax adaptive control for a finite set of linear systems.
\newblock In \emph{Learning for Dynamics and Control}, pages 893--904. PMLR,
  2021.

\bibitem[Renganathan et~al.(2023)Renganathan, Iannelli, and
  Rantzer]{renganathan2023online}
Venkatraman Renganathan, Andrea Iannelli, and Anders Rantzer.
\newblock An online learning analysis of minimax adaptive control.
\newblock In \emph{2023 62nd IEEE Conference on Decision and Control (CDC)},
  pages 1034--1039. IEEE, 2023.

\bibitem[Simchowitz and Foster(2020)]{simchowitz2020naive}
Max Simchowitz and Dylan Foster.
\newblock Naive exploration is optimal for online lqr.
\newblock In \emph{International Conference on Machine Learning}, pages
  8937--8948. PMLR, 2020.

\bibitem[Simchowitz et~al.(2020)Simchowitz, Singh, and
  Hazan]{simchowitz2020improper}
Max Simchowitz, Karan Singh, and Elad Hazan.
\newblock Improper learning for non-stochastic control.
\newblock In \emph{Conference on Learning Theory}, pages 3320--3436. PMLR,
  2020.

\bibitem[Stein(1980)]{stein1980adaptive}
Gunter Stein.
\newblock Adaptive flight control: A pragmatic view.
\newblock In \emph{Applications of Adaptive Control}, pages 291--312. Elsevier,
  1980.

\bibitem[Stewart and Sun(1990)]{stewart1990matrix}
Gilbert~W Stewart and Ji-guang Sun.
\newblock Matrix perturbation theory.
\newblock \emph{(No Title)}, 1990.

\bibitem[Tilli(1998)]{tilli1998singular}
Paolo Tilli.
\newblock Singular values and eigenvalues of non-hermitian block toeplitz
  matrices.
\newblock \emph{Linear algebra and its applications}, 272\penalty0
  (1-3):\penalty0 59--89, 1998.

\bibitem[Tripuraneni et~al.(2020)Tripuraneni, Jordan, and
  Jin]{tripuraneni2020theory}
Nilesh Tripuraneni, Michael Jordan, and Chi Jin.
\newblock On the theory of transfer learning: The importance of task diversity.
\newblock \emph{Advances in neural information processing systems},
  33:\penalty0 7852--7862, 2020.

\bibitem[Tsiamis et~al.(2022)Tsiamis, Ziemann, Morari, Matni, and
  Pappas]{tsiamis2022learning}
Anastasios Tsiamis, Ingvar~M Ziemann, Manfred Morari, Nikolai Matni, and
  George~J Pappas.
\newblock Learning to control linear systems can be hard.
\newblock In \emph{Conference on Learning Theory}, pages 3820--3857. PMLR,
  2022.

\bibitem[Tsiamis et~al.(2023)Tsiamis, Ziemann, Matni, and
  Pappas]{tsiamis2023statistical}
Anastasios Tsiamis, Ingvar Ziemann, Nikolai Matni, and George~J Pappas.
\newblock Statistical learning theory for control: A finite-sample perspective.
\newblock \emph{IEEE Control Systems Magazine}, 43\penalty0 (6):\penalty0
  67--97, 2023.

\bibitem[Vershynin(2020)]{vershynin2020high}
Roman Vershynin.
\newblock High-dimensional probability.
\newblock \emph{University of California, Irvine}, 2020.

\bibitem[Wang et~al.(2023)Wang, Toso, Mitra, and Anderson]{wang2023model}
Han Wang, Leonardo~F Toso, Aritra Mitra, and James Anderson.
\newblock Model-free learning with heterogeneous dynamical systems: A federated
  lqr approach.
\newblock \emph{arXiv preprint arXiv:2308.11743}, 2023.

\bibitem[Zhang et~al.(2023{\natexlab{a}})Zhang, Kang, Lee, Tomlin, Levine, Tu,
  and Matni]{zhang2023multi}
Thomas~T Zhang, Katie Kang, Bruce~D Lee, Claire Tomlin, Sergey Levine, Stephen
  Tu, and Nikolai Matni.
\newblock Multi-task imitation learning for linear dynamical systems.
\newblock In \emph{Learning for Dynamics and Control Conference}, pages
  586--599. PMLR, 2023{\natexlab{a}}.

\bibitem[Zhang et~al.(2023{\natexlab{b}})Zhang, Toso, Anderson, and
  Matni]{zhang2023meta}
Thomas~TCK Zhang, Leonardo~Felipe Toso, James Anderson, and Nikolai Matni.
\newblock Sample-efficient linear representation learning from non-iid
  non-isotropic data.
\newblock In \emph{The Twelfth International Conference on Learning
  Representations}, 2023{\natexlab{b}}.

\bibitem[Ziemann and Sandberg(2022)]{ziemann2022regret}
Ingvar Ziemann and Henrik Sandberg.
\newblock Regret lower bounds for learning linear quadratic gaussian systems.
\newblock \emph{arXiv preprint arXiv:2201.01680}, 2022.

\bibitem[Ziemann et~al.(2023)Ziemann, Tsiamis, Lee, Jedra, Matni, and
  Pappas]{ziemann2023tutorial}
Ingvar Ziemann, Anastasios Tsiamis, Bruce Lee, Yassir Jedra, Nikolai Matni, and
  George~J Pappas.
\newblock A tutorial on the non-asymptotic theory of system identification.
\newblock In \emph{2023 62nd IEEE Conference on Decision and Control (CDC)},
  pages 8921--8939. IEEE, 2023.

\end{thebibliography}
\clearpage
\appendix

\section{Additional Experimental Details}
\label{s: additional experimental details}
The matrices $A^\star$ and $B^\star$ considered in \Cref{s: numerical example} are given as follows: 
\begin{align}
    % \label{eq: lin cartpole}
    % x_{t+1} = \bmat{1 & 0.25 & 0 & 0 \\ 0 & 1 & -0.25 & 0 \\ 0 & 0 & 1 & 0 \\ 0 & 0 & 0.5 & 1} x_t + \bmat{0 \\ 0.25 \\ 0 \\ -0.25} u_t + \sqrt{0.01} I w_t \triangleq Ax_t + Bu_t + \Sigma_w^{1/2} w_t,
    A^\star = \bmat{1 & 0.25 & 0 & 0 \\ 0 & 1 & -0.25 & 0 \\ 0 & 0 & 1 & 0 \\ 0 & 0 & 0.5 & 1}, \, B^\star = \bmat{0 \\ 0.25 \\ 0 \\ -0.25}.
\end{align}

We consider $Q=I$, and $R=I$. We keep the values of $K_0$, $\tau_1$, $x_b$, and $K_b$ fixed across all experiments as $\bmat{0.37 & 1.64 & 4.49 & 3.89}$, $1024$, $50$, and $15$. We replace the condition $\norm{x_t}^2 \geq x_b^2 \log T$ from line 5 of \Cref{alg: ce with exploration} with $\norm{x_t}^2 \geq x_b$. These modifications are made because the guarantees in \Cref{s: analysis} are not anytime, but rather require the number of timesteps $T$ which the algorithm is to be run as an input to the algorithm. For sake of efficiently plotting the regret against time, we remove the dependence of the algorithm on $T$, and keep these parameters fixed across time. Despite the discrepancy, we observe that the predicted trends still hold.

\paragraph{Learned Representation}
The parameter values used for the systems which generate the offline data are
$(M,m,\ell) =(0.4, 1.0, 1.0)$, $(1.6, 1.3, 0.3)$, $(1.3, 0.7, 0.65)$, $(0.2, 0.06, 1.36)$ and, $(0.2, 0.47, 1.83)$. The data from the is collected with process noise $w_t\sim\calN(0,0.01 I)$, and inputs $u_t = K_0 x_t + \eta_t$, where $\eta_t \sim \calN(0, I)$. The representation is determined by jointly optimizing $\Phi$ along with the system specific parameter values $\theta$ to solve the least squares problem
\begin{align*}
    \hat \Phi, \hat \theta^{(1)}, \dots, \hat \theta^{(5)} = \argmin_{\Phi, \theta^{(1)}, \dots, \theta^{(5)}} \sum_{i=1}^5 \sum_{t=1}^{1199} \norm{x_{t+1}^{(i)} - \VEC^{-1} \paren{\Phi \theta^{(i)}} \bmat{x_t^{(i)} \\ u_t^{(i)}} }^2,
\end{align*}
 where the superscript $(i)$ indicates which of the five systems the data is collected from. 
% The parameters for the fie discretized and linearized cartpole systems  used to collect the offline dataset are 

\section{Estimation Error Bounds}
\label{s: estimation error bounds}
We begin by considering a general estimation problem in which the system is excited by an arbitrary stabilizing controller $K$ and excitation level defined by $\sigma_u$.  In particular, we consider the evolution of the following system: 
\begin{equation}
\label{eq: rollout under K}
\begin{aligned}
    x_{t+1} &= A^\star x_t + B^\star u_t +  w_t \\
    u_t &= K x_t + \sigma_u g_t,
\end{aligned}
\end{equation}
where $g_t \overset{i.i.d.}{\sim} \calN(0,I)$, and $w_t$ is a random variable with $\sigma^2$-sub-Gaussian entries satifying $\E[w_t w_t^\top] =I  $. We assume that $\sigma_u^2\leq 1$ and that $x_1$ is a random variable. 
We consider generating the estimates $\hat \theta, \Lambda = \texttt{LS}(\hat \Phi, x_{1:t+1}, u_{1:t})$. 
We present a bound on the estimation error $\norm{\hat \Phi \hat \theta - \Phi^\star \theta^\star}^2$ in terms of the true system parameters as well as the amount of data, $t$.  

Before presenting our main estimation result, we define the shorthand $A_K  \triangleq A^\star + B^\star K$. We also use the $\calH_\infty$ norm, defined below. 
\begin{definition}
    The $\calH_\infty$ norm of a stable matrix $A$ is given by 
    \begin{align*}
        \norm{A}_{\calH_{\infty}} = \sup_{z \in \mathbb{C}, \abs{z} = 1} \norm{\paren{zI - A}^{-1}}.
    \end{align*}
\end{definition}
\begin{theorem}[least squares estimation error]
    \label{thm: ls estimation error}
    Let $\delta\in(0,1/2)$. Suppose $t \geq c \tau_{\mathsf{ls}}(K, \norm{x_1}^2, \delta)$ for 
      \begin{align*}
      \tau_{\ls}(K, \bar x, \delta) & \triangleq  \max\bigg\{\sigma^4  \norm{P_K}^3 \Psi_{B^\star}^2 \paren{\dx+\du + \log \frac{1}{\delta}},  \bar x \times \norm{P_K} + 1\bigg\}
      \end{align*}
      and a sufficiently large universal constant $c>0$. There exists an event $\calE_{\ls}$ which holds with probability at least $1-\delta$ under which the estimation error satisfies 
    %  \begin{align*}
    %     &\norm{\hat \Phi \hat \theta  - \Phi^\star \theta^\star}^2 \leq \frac{16 \dtheta \sigma^2  }{t \lambda_{\min} \paren{\hat \Phi^\top \paren{\bar \Sigma^t(K,\sigma_u,x_1) \otimes I_{\dx}} \hat \Phi}} \log\paren{ \frac{128}{ \delta} } \\
    %     &\quad+ \paren{1+ 2000000000\frac{\sigma^4 \norm{P_K}^7 \Psi_{B^\star}^4 \paren{\dx(\dx+\du) + \log \frac{4}{\delta}}}{t\lambda_{\min}(\hat \Phi^\top \paren{\bar \Sigma^t(K,\sigma_u,x_1) \otimes I_{\dx}} \hat \Phi)^2 } } \times \frac{10\norm{P_K}^2 \Psi_{B^\star} d(\hat \Phi, \Phi^\star)^2 \norm{\theta^\star}^2}{\lambda_{\min}(\hat \Phi^\top \paren{\bar \Sigma^t(K,\sigma_u,x_1) \otimes I_{\dx}} \hat \Phi)} .
    % \end{align*}
    \begin{align*}
        &\norm{\hat \Phi \hat \theta  - \Phi^\star \theta^\star}^2 \lesssim \frac{\dtheta \sigma^2  }{t \lambda_{\min} \paren{\bar \Delta^t(\sigma_u, K)}} \log\paren{ \frac{1}{ \delta} } \\
        &+ \paren{1+\frac{\sigma^4 \norm{P_K}^7 \Psi_{B^\star}^6 \paren{\dx+\du + \log \frac{1}{\delta}}}{t\lambda_{\min}(\bar \Delta^t(\sigma_u, K))^2 } } \frac{\norm{P_K}^2 \Psi_{B^\star}^2 d(\hat \Phi, \Phi^\star)^2 \norm{\theta^\star}^2}{\lambda_{\min}(\bar \Delta^t(\sigma_u, K))} .
    \end{align*}
    where
    \begin{align*}
        \bar \Delta^t(\sigma_u, K) \triangleq \hat \Phi^\top \paren{\frac{1}{t}\sum_{s=0}^{t-2} \sum_{j=0}^s \bmat{I \\ K} A_{K}^j (\sigma_u^2 B^\star \paren{B^\star}^\top + I)   \paren{A_K^j}^\top \bmat{I \\  K}^\top + \bmat{0 & 0 \\ 0 & \sigma_u^2 I_{\du}}} \hat \Phi.
    \end{align*}
    % \begin{align*}
    %     \bar \Sigma^t(K,\sigma_u,x_1) = \frac{1}{t}\sum_{s=0}^{t-2} \sum_{j=0}^s \bmat{I \\ K} A_{K}^j (\sigma_u^2 B^\star \paren{B^\star}^\top + \Sigma_w)   \paren{A_K^j}^\top \bmat{I \\  K}^\top + \bmat{0 & 0 \\ 0 & \sigma_u^2 I_{\du}}.\\ \mbox{\Bruce{Use different notation for this matrix}}
    % \end{align*}
\end{theorem}

Toward proving the above theorem, we first define a mapping from the noise sequence and the state $x_1$ to the stacked sequence of states that is be used throughout the section. In particular, 
\begin{align}
    \label{eq: state sequence}
        \bmat{x_1 \\ x_2\\ \vdots \\ x_t} = \left[\begin{array}{c|c} F_1^t(K,\sigma_u) & F_2^t(K,\sigma_u) \end{array} \right]  \left [\begin{array}{c} x_1 \\ \hline \eta_t \end{array} \right],
    \end{align}
    where $\eta_t \triangleq  \bmat{\bmat{w_1^\top & g_1^\top} & \dots &\bmat{w_{t-1}^\top & g_{t-1}^\top}}^\top$ and 
    \begin{align*}
        \left[\begin{array}{c|c} F_1^t(K,\sigma_u) & F_2^t(K,\sigma_u) \end{array} \right] &\triangleq 
         \left[\begin{array}{c|cccc} I & 0 & 0 & \dots & 0 \\ A_K & \bmat{ I& B^\star \sigma_u} & 0 &  \dots & 0 \\ \vdots & & \ddots \\ A_K^{t-1} &  A_K^{t-2} \bmat{ I& B^\star \sigma_u}  & \dots & & \bmat{ I & B^\star \sigma_u} \end{array} \right].
    \end{align*}
     Using the fact that $\bmat{x_t \\ u_t} = \bmat{I \\ K} x_t + \bmat{0 \\ \sigma_u g_t}$, we may express the stacked states and inputs as 
    \begin{align}
    \nonumber
     \bmat{x_1 \\ u_1 \\ \vdots \\ x_t \\ u_t} &= \paren{I_{t} \otimes \bmat{I \\ K}} F_1^t(K,\sigma_u) x_1 +  \paren{\paren{I_t \otimes \bmat{I \\ K}} F_2^t(K,\sigma_u)   + \paren{I_t \otimes \bmat{0 & 0 \\ 0 & \sigma_u I_{\du}}}} \eta \\
        &\triangleq \Xi_1^t(K,\sigma_u) x_1 + \Xi_2^t(K,\sigma_u) \eta.
        \label{eq: state and input sequence}
    \end{align}
We additionally define two quantities related to the empirical covariance matrix. In particular, we define the expected value of the empirical covariance matrix conditioned on the state at the start of the sequence, $x_1$, along with the centered counterpart: 
\begin{align*}
    \Sigma^t(K,\sigma_u,x_1) &\triangleq \E \brac{ \frac{1}{t} \sum_{s=1}^t \bmat{x_s \\ u_s} \bmat{x_s \\ u_s}^\top \vert x_1} \mbox{ and }  \\
    \bar \Sigma^t(K,\sigma_u,x_1) &\triangleq \E \brac{ \frac{1}{t} \sum_{s=1}^{t} \paren{\bmat{x_s \\ u_s} - \E \brac{\bmat{x_s \\ u_s} \vert x_1} }\paren{\bmat{x_s \\ u_s} - \E \brac{\bmat{x_s \\ u_s }\vert x_{1} }}^\top}.
\end{align*}
% Note that $\Sigma^t(K,\sigma_u,x_1)$ is a random variable, while $\bar \Sigma^t(K,\sigma_u,x_1)$ is deterministic. 

\begin{lemma}[Epoch-wise covariance bounds]
    \label{lem: covariance facts}
    For $t \geq 2$, we have
    \begin{enumerate}
        \item $\bar \Sigma^t(K,\sigma_u,x_1) = \frac{1}{t}\sum_{s=0}^{t-2} \sum_{j=0}^s \bmat{I \\ K} A_{K}^j (\sigma_u^2 B^\star \paren{B^\star}^\top +  I)   \paren{A_K^j}^\top \bmat{I \\  K}^\top + \bmat{0 & 0 \\ 0 & \sigma_u^2 I_{\du}}$
        \item $ \Sigma^t(K,\sigma_u,x_1) = \bar \Sigma_k + \frac{1}{t} \sum_{s=0}^{t-1} \bmat{I \\ K} A_{K}^{s} x_1 x_1^\top \paren{A_K^{s}}^\top \bmat{I \\  K}^\top$
        \item $\Sigma^t(K,\sigma_u,x_1)  \succeq \bar \Sigma^t(K,\sigma_u,x_1)  \succeq \frac{\sigma_u^2 }{2(1+2\norm{K}^2 + \sigma_u^2)}  I$ 
    \end{enumerate}
    %If additionally $K = K^\star(\hat A, \hat B)$, where $\norm{\bmat{\hat A & \hat B} - \bmat{A^\star & B^\star}}_F^2 \leq \frac{1}{C_{\mathsf{safe}}(A^\star,B^\star)^2}$ and $\sigma_u^2 \leq 1$, then
    \begin{enumerate}\addtocounter{enumi}{3}
        %\item  $\bar \Sigma^t(K,\sigma_u, x_1) \succeq \frac{\sigma_u^2}{6.1 \norm{P^\star}} I$ \Bruce{Unnecessary?}
        \item $ \Sigma^t(K,\sigma_u,x_1)  \preceq (1+  \norm{P_K}\frac{\norm{ x_1^2}}{ t-1 }) \bar \Sigma^t(K,\sigma_u,x_1)$  
        
        \item $\norm{\bar \Sigma^t(K,\sigma_u,x_1)} \leq 5 \norm{P_K}^2  \Psi_{B^\star }^2$. 
        %\item $\norm{\Sigma_k} \lesssim (1+\frac{\norm{x_{\tau_{k-1}}}^2}{\tau_{k-1}}) \norm{P^\star}^2 \Psi_{B^\star}$. \Bruce{Remove this point. It gets superceded by point 4.}
    \end{enumerate}
\end{lemma}
\begin{proof}
\paragraph{Points 1 and 2: } These results follow by expressing the states and inputs in terms of the noise sequence and $x_1$, then using the fact that the noise terms are independent, with mean zero and identity covariance. 

\paragraph{Point 3:} The first inequality is immediate from point one by the fact that 
    \[
        \frac{1}{t} \sum_{s=0}^{t-1} \bmat{I \\ K} A_K^{s} x_1 x_1^\top \paren{A_K^{s}}^\top \bmat{I \\ K}^\top \succeq 0.
    \]
    For the second inequality, note that
    \begin{align*}
        \bar \Sigma^t(K,\sigma_u, x_1)&=\frac{1}{t}\sum_{s=0}^{t-2} \sum_{j=0}^s \bmat{I \\ K} A_{K}^j (\sigma_u^2 B^\star \paren{B^\star}^\top + I)   \paren{A_K^j}^\top \bmat{I \\  K}^\top + \bmat{0 & 0 \\ 0 & \sigma_u^2 I_{\du}}  \\
        &\succeq \frac{1}{2} \bmat{I \\ K} \bmat{I \\ K}^\top + \bmat{0 & 0 \\ 0 & \sigma_u^2 I_{\du}}.
    \end{align*}
    By Lemma F.6 of \cite{dean2018regret}, 
    \begin{align*}
        \lambda_{\min}\paren{\frac{\sigma_w^2}{2}\bmat{I \\ K} \bmat{I \\ K}^\top + \bmat{0 & 0 \\ 0 & \sigma_u^2 I_{\du}}} \geq  \sigma_u^2 \min\curly{\frac{1}{2}, \frac{1}{4  \norm{K}^2 + 2 \sigma_u^2}}.% \geq \sigma_u^2 \min\curly{\frac{1}{2}, \frac{1}{2 \norm{K}^2 + 2 \sigma_u^2}}
    \end{align*}
    To conclude, we pull out $\frac{1}{2}$, and we see that $\min\curly{1, \frac{1}{2\norm{K}^2 +\sigma_u^2}} \geq \frac{1}{1+2 \norm{K}^2 + \sigma_u^2}$.
    % \paragraph{Point 4: } Under the condition that $K = K^\star(\hat A, \hat B)\norm{\bmat{\hat A & \hat B} - \bmat{A^\star & B^\star}}_F^2 \leq \frac{1}{C_{\mathsf{safe}}(A^\star,B^\star)^2}$, \Cref{lem: consequences of safe event} tells us that
    % \begin{align*}
    %     2\norm{K}^2 \leq 2.1 \norm{P^\star}.
    % \end{align*}
    % Therefore $2 + 2 \norm{K}^2 + 2\sigma_u^2 \leq 6.1 \norm{P^\star}$ by the fact that $\sigma_u^2 \leq 1 \leq \norm{P^\star}$.  
    
    \paragraph{Point 4: } Using the fact that for $vv^\top \preceq \norm{v}^2 I$, we find 
    \begin{align*}
        &\frac{1}{t} \sum_{s=0}^{t-1} \bmat{I \\ K} A_K^{s} x_1 x_1^\top \paren{A_K^{s}}^\top \bmat{I \\ K}^\top \\
        &\preceq \norm{x_1}^2 \frac{1}{t} \sum_{s=0}^{t-1} \bmat{I \\ K} A_K^{s} \paren{A_K^{s}}^\top \bmat{I \\ K}^\top \\
        & \preceq \norm{x_1}^2 \norm{\sum_{s=0}^{t-1}  A_{K}^{s} \paren{A_{K}^{s}}^\top} \frac{1}{t}  \bmat{I \\ K}\bmat{I \\ K}^\top \\
        & \preceq  \frac{\norm{x_{1}^2}}{t-1} \norm{\dlyap[A_K]} \frac{t-1}{t}  \bmat{I \\ K}\bmat{I \\ K}^\top \\
        &\preceq  \frac{\norm{x_{1}^2}}{ t-1 } \norm{P_K} \bar \Sigma^t(K,\sigma_u, x_1).% \\
        %&\precsim  \frac{\norm{x_{1}^2}}{t-1} \norm{P^\star} \bar \Sigma^t(K,\sigma_u, x_1)
    \end{align*}
    %where the final inequality follows from point 4 of \Cref{lem: consequences of safe event}.

    \paragraph{Point 5:}
    By the fact that $\sigma_u^2 \leq 1$ we have from the triangle inequality and submultiplicativity that 
    \begin{align*}
        \norm{\bar \Sigma^t(K,\sigma_u, x_1)} &= \norm{ \frac{1}{t-1}\sum_{s=0}^{t-2} \sum_{j=0}^s \bmat{I \\ K} A_{K}^j (\sigma_u^2 B^\star \paren{B^\star}^\top +   I)   \paren{A_K^j}^\top \bmat{I \\  K}^\top + \bmat{0 & 0 \\ 0 & \sigma_u^2 I_{\du}}} \\
        &\leq 1 + \paren{1 + \norm{K}^2} \paren{\norm{B^\star}^2 + 1 } \norm{\frac{1}{t}\sum_{t=0}^{t-2}\sum_{j=0}^t A_K^j    \paren{A_K^j}^\top} \\
        &\leq 1 + \paren{1 + \norm{K}^2} \paren{\norm{B^\star}^2 + 1} \norm{\sum_{j=0}^\infty A_K^j    \paren{A_K^j}^\top}\\
        &= 1 + \paren{1 + \norm{K}^2} \paren{\norm{B^\star}^2 + 1} \norm{\dlyap[A_K]} \\
        &\leq 1 + \paren{1 + \norm{K}^2} \paren{\norm{B^\star}^2 + 1} \norm{P_K} \\
        & \leq  5 \norm{P_K}^2  \Psi_{B^\star }^2, 
    \end{align*}
    where the final point follows from the fact that $\norm{K}^2 \leq \norm{Q+ K^\top R K} \leq \norm{P_K}$, $\norm{P_K}\geq 1$, and the definition of $\Psi_{B^\star }$. 
    % \paragraph{Point 7: } 
    % We have from points 1 and 2 that 
    % \begin{align*}
    %     \norm{\Sigma_k} &\leq \norm{\bar \Sigma_k} + \norm{\frac{1}{\tau_{k-1}} \sum_{t=0}^{\tau_{k-1}-1} \bmat{I \\ \hat K_k} A_{\cl,k}^{t} x_{\tau_{k-1}}x_{\tau_{k-1}}^\top \paren{A_{\cl,k}^{t}}^\top \bmat{I \\ \hat K_k}^\top} \\
    %     & \leq 1 + \paren{1+ \norm{\hat K_k}^2}\paren{\norm{B^\star}^2+1 + \frac{\norm{x_{\tau_{k-1}}}^2}{\tau_{k-1}} } \norm{\sum_{j=0}^\infty A_{\cl,k}^j \paren{A_{\cl,k}^j}^\top} \\
    %     &\lesssim \norm{P^\star}^2 \paren{\Psi_{B^\star} + \frac{\norm{x_{\tau_{k-1}}}^2}{\tau_{k-1}} } \\ 
    %      &\leq \norm{P^\star}^2 \Psi_{B^\star}\paren{1 + \frac{\norm{x_{\tau_{k-1}}}^2}{\tau_{k-1}}}.
    % \end{align*}
\end{proof}

To analyze the estimate $\hat \theta$, we first verify that it solves the least squares problem of interst. 
\begin{lemma}
    \label{lem: ls solution}
    Consider applying \Cref{alg: least squares} to an estimated model structure $\hat \Phi$, state data $x_{ 1:t+1}$ and input data $u_{ 1: t}$: $\hat \theta , \Lambda  = LS(\hat \Phi, x_{ 1:t+1}, u_{1:t})$. As long as
    $  \Lambda   \succ 0,$ then $\hat \theta $ is the unique solution to  
    \begin{align}
        \label{eq: ls problem}
        \argmin_{\theta} \sum_{s=1}^{t} \norm{x_{s+1} - \VEC^{-1}(\hat \Phi \theta) \bmat{x_s \\ u_s}}^2.
    \end{align}
\end{lemma}
\begin{proof}
    By the vectorization identity $\VEC(XYZ) = (Z^\top \otimes X)\VEC(Y)$, we may write 
    \begin{align*}
        \VEC^{-1}(\hat \Phi \theta) \bmat{x_s \\ u_s} = \paren{\bmat{x_s \\ u_s}^\top \otimes I_{\dx} }\hat \Phi \theta.
    \end{align*}
    Minimizing $
        \sum_{s= 1}^{t} \norm{x_{s+1} - \paren{\bmat{x_s \\ u_s}^\top \otimes I_{\dx}}\hat \Phi \theta }^2$
    over $\theta$ yields $\hat \theta $.     
\end{proof}

We now split the estimation error into a part arising from the variance and a part arising from the bias. The variance decays as $t$ grows sufficiently large. The bias is due to the misspecificiation between the model structure estimate $\hat \Phi$ and the true model structure $\Phi^\star$. 
\begin{lemma}
    \label{lem: bias variance decomposition}
     Consider applying \Cref{alg: least squares} to an estimated model structure $\hat \Phi$, state data $x_{1:t+1}$ and input data $u_{1:t}$: $\hat \theta, \Lambda = LS(\hat \Phi, x_{1:t+1}, u_{1:t})$. Suppose that $\Lambda \succ 0$. Let
    \begin{align}
        \label{eq: optimal projection}
        \bar \theta = \argmin_{\theta} (\hat \Phi \theta - \Phi^\star \theta^\star)^\top \paren{\Sigma^t(K,\sigma_u,x_1) \otimes I_{\dx}} (\hat \Phi \theta - \Phi^\star\theta^\star).
    \end{align}
    We have that 
    \begin{align}
        \label{eq: bias variance decomposition}
         \norm{\hat \Phi \hat \theta - \Phi^\star \theta^\star}^2 \leq \underbrace{2 \norm{\hat \Phi \hat \theta - \hat \Phi \bar \theta}^2}_{\mbox{variance}} + \underbrace{2 \norm{\hat \Phi \bar \theta - \Phi^\star \theta^\star}^2}_{\mbox{bias}}. 
    \end{align}

\end{lemma}
\begin{proof}
    The lemma folllows by application of the triangle inequality.
\end{proof}

The variance term in the above lemma is bounded in \Cref{lem: variance bound}, while the bias term is bounded in \Cref{lem: bias bound}. Before presenting these lemmas, we present several covariance concentration guarantees that will be useful in obtaining presenting the variance bound.

\subsection{Covariance Concentration}
The following $\varepsilon$-net argument is standard. A proof may be found in Chapter 4 of \cite{vershynin2020high}.
\begin{lemma}
    \label{lem: symmetric covering}
    Let $W$ be and $d\times d $ random matrix, and $\varepsilon \in [0, 1/2)$. Furthermore, let $\calN$ be an $\varepsilon$-net of $\calS^{d-1}$ 
    %\st{$\calS^{d-1}$ is not defined thus far}
    with minimal cardinality. Then for all $\rho > 0$,
    \[
        \P\brac{\norm{W} > \rho} \leq \paren{\frac{2}{\varepsilon} + 1}^d \max_{x \in \calN} \P\brac{\abs{x^\top W x} > (1-2\epsilon)\rho}.
    \]
    % \st{You can replace the $\max_{x \in \calN}$ on the RHS
    % with $\max_{x \in \calS^{d-1}}$, and then no need to even reference any
\end{lemma}

The following lemma helps to bound system constants in term of a few key quantities.
\begin{lemma}[Lemma B.11 of \cite{simchowitz2020naive}]
    \label{lem: powers of closed loop by lyap}
    For a stabilizing controller $K$, we have that for any $t$, $\norm{A_K^t} \leq \sqrt{\paren{1 - \frac{1}{\norm{P_K}}}^t \norm{P_K}}$ and $\sum_{t=0}^\infty \norm{A_K^t} \leq 2 \norm{P_K}^{3/2}$. As a consequence, we also have $\norm{A_K}_{\calH_\infty} \leq 2 \norm{P_K}^{3/2}$.
\end{lemma}
\begin{proof}
    For the first point, observe that
    \begin{align*}
        \norm{A_K^t} &= \sqrt{\norm{(A_K^t)^\top A_K^t}} \leq \sqrt{\norm{(A_K^t)^\top P_K A_K^t}}  \leq \sqrt{\paren{1 - \frac{1}{\norm{P_K}}}^t \norm{P_K}}.
    \end{align*}
    For we may use the first point to show that 
    \begin{align*}
        \sum_{t=0}^\infty \norm{A_K^t} 
        &\leq \norm{P_K}^{1/2} \sum_{t=0}^\infty \paren{\sqrt{1-\frac{1}{\norm{P_K}}}}^t \\ & = \norm{P_K}^{1/2} \frac{1}{1-\sqrt{1 - \frac{1}{\norm{P_K}}}} \\
         &= \norm{P_K}^{1/2} \frac{1 + \sqrt{1 - \frac{1}{\norm{P_K}}}}{1/\norm{P_K}} \\&\leq 2 \norm{P_K}^{3/2}.
    \end{align*}
    The result on the $\calH_\infty$ norm of $A_K$ follows from the previous point by noting that $\norm{A_K}_{\calH_\infty} \leq \sum_{t=0}^\infty \norm{A_K^t}$.
\end{proof}

We now present a result about the concentration of the empirical covariance to the true covariance. The result is adapted from Lemma 2 of \cite{jedra2020finite} and Lemma A.2 of \cite{zhang2023multi}, which generalized the approximate isometries argument in \cite{jedra2020finite} to account for projecting the covariance onto a lower dimensional space by hitting it on either side with a low rank orthonormal matrix. The primary difference of the below result from that in \cite{zhang2023multi} is the inclusion of a nonzero initial state, resulting in an additional sub-gaussian concentration step. 
\begin{lemma}[Approximate Isometries]
    \label{lem: approximate isometries}
    Let data $x_{1:t}$, $u_{1:t}$ be generated by \eqref{eq: rollout under K}. 
    Let $G \in \R^{\dc (\dx+\du)\times d_G}$ be a matrix with orthonormal columns. Define
    \begin{align*}
        M \triangleq\paren{G^\top \paren{\Sigma^t(K,\sigma_u,x_1) \otimes I_{\dc}} G}^{-1/2}.
    \end{align*}
    Then for some positive universal constant $C$, and for any $\rho<1$, the inequality 
    \begin{align*}
        \norm{M G^\top \paren{\frac{1}{\tau_{k-1}} \sum_{t=\tau_{k-1}}^{\tau_k - 1} \bmat{x_t \\ u_t} \bmat{x_t \\ u_t}^\top \otimes I_{\dc}} G M - I} > \rho
    \end{align*}
    holds with probability at most
     \begin{align*}4 \times 9^{d_G} \exp\paren{-C \frac{ \min\curly{\rho,\rho^2} t}{\sigma^4 \norm{P_K}^3  \Psi_{B^\star}^2}},
    \end{align*}
\end{lemma}
\begin{proof}
    Define the shorthand 
    \[
        X \triangleq \frac{1}{\sqrt{t}} \bmat{x_1^\top & u_1^\top \\ \vdots & \vdots \\ x_t^\top  & u_t^\top } \otimes I_{\dc}
    \]
    such that 
    \[
        \paren{\frac{1}{t} \sum_{s=1}^{t} \bmat{x_s \\ u_s} \bmat{x_s \\ u_s}^\top \otimes I_{\dc}} =  X^\top X.
    \]
    First note that $\E\brac{G^\top X^\top X G \vert x_1} = M^{-2}$.
    Then 
    \begin{align}   
        \label{eq: quadratic form quantity of interest}
         \norm{M G^\top X^\top X G M - I} 
         &= \sup_{v \in \calS^{d_G}} \abs{v^\top \paren{M G^\top X^\top X G M - I}v }\\
         \nonumber
         &= \sup_{v \in \calS^{d_G}} \abs{v^\top M G^\top X^\top X G M v - \E \brac{v^\top M G^\top X^\top X G M v \vert x_1  } } \\
         \nonumber
          &= \sup_{v \in \calS^{d_G}} \abs{ \norm{X G M v}^2 - \E \brac{\norm{X G Mv}^2 \vert x_1 }}.
    \end{align}
    Using the identity $\VEC(WYZ) = (Z^\top \otimes W) \VEC(Y)$, we may write 
    \begin{align}
        \label{eq: vectorizing quadratic form}
        \VEC\paren{v^\top M^\top G^\top X^\top} &= (I_{t \times \dx}\otimes (G M v)^\top ) \VEC(X^\top ) \triangleq \sigma_{GM v}  \VEC(X^\top ).
    \end{align}

    Next, we construct matrices $\Gamma_1^t(K,\sigma_u)$ and $\Gamma_2^t(K,\sigma_u)$ such that $\VEC(X^\top)=\Gamma_1^t(K,\sigma_u) x_1 + \Gamma_2^t(K,\sigma_u) \eta_t$.
    To construct these matrices, first recall the definition of $\Xi_1^t(K,\sigma_u)$ and $\Xi_2^t(K,\sigma_u)$ given in \eqref{eq: state and input sequence} and of $F_1^t(K,\sigma_u)$ and $F_2^t(K,\sigma_u)$ given in \eqref{eq: state sequence}. To reach $\VEC\paren{X^\top} $ from the vector of stacked states and inputs, we can apply the following transformation:
    \begin{align*}
        \VEC(X^\top) = \frac{1}{\sqrt{t}} \paren{I_t \otimes \bmat{I_{\dx+\du} \otimes e_1 \\ \vdots \\I_{\dx + \du} \otimes e_{\dc}}}   \bmat{x_1 \\ u_1 \\ \vdots \\ x_t \\ u_t} \triangleq \frac{1}{\sqrt{t}} H\bmat{x_1 \\ u_1 \\ \vdots \\ x_t \\ u_t}.
    \end{align*}
    Combining the above sequence of transformation implies that by defining $ \Gamma_1^t(K,\sigma_u) \triangleq \frac{1}{t} H \Xi_1^t(K,\sigma_u)$ and $\Gamma_2^t(K,\sigma_u) \triangleq \frac{1}{t}H \Xi_2^t(K,\sigma_u)$, we have $\VEC(X^\top)=\Gamma_1 x_1 + \Gamma_2^t(K,\sigma_u) \eta_t$.
    
    Using the above expression, our quantity of interest in \eqref{eq: quadratic form quantity of interest} may be expressed as 
    \begin{align*}
        &\norm{M G^\top X^\top X G M - I}  \\
        &=\sup_{v \in \calS^{d_G}} \abs{\norm{\sigma_{GM v} \paren{\Gamma_1^t(K,\sigma_u) x_1  + \Gamma_2^t(K,\sigma_u) \eta_t}}^2 
        -\E \brac{ \norm{\sigma_{GM v} \paren{\Gamma_1^t(K,\sigma_u) x_1  + \Gamma_2^t(K,\sigma_u) \eta_t}}^2 \vert x_1}} \\
        &=  \sup_{v \in \calS^{d_G}} \Big|\norm{\sigma_{GM v} \Gamma_2^t(K,\sigma_u) \eta_t}^2  + \norm{ \sigma_{GM v} \Gamma_1^t(K,\sigma_u)  x_1}^2  + 2 \langle \sigma_{GM v} \Gamma_1^t(K,\sigma_u)  x_1, \sigma_{GM v} \Gamma_2^t(K,\sigma_u) \eta_t \rangle\\ 
        &\quad  -\E \brac{ \norm{\sigma_{GM v} \Gamma_2^t(K,\sigma_u) \eta_t}^2  + \norm{ \sigma_{GM v} \Gamma_1^t(K,\sigma_u)  x_1}^2 + 2 \langle \sigma_{GM v} \Gamma_1^t(K,\sigma_u)  x_1, \sigma_{GM v} \Gamma_2^t(K,\sigma_u) \eta_t \rangle \vert x_1} \Big| \\
        &\overset{(i)}{=} \sup_{v \in \calS^{d_G}} \abs{\norm{\sigma_{GM v} \Gamma_2^t(K,\sigma_u \eta_t}^2  + 2 \langle \sigma_{GM v} \Gamma_1^t(K,\sigma_u  x_1, \sigma_{GM v} \Gamma_2^t(K,\sigma_u \eta_t \rangle  -\E \brac{ \norm{\sigma_{GM v} \Gamma_2^t(K,\sigma_u) \eta_t}^2 }} \\
         &\overset{(ii)}{=} \sup_{v \in \calS^{d_G}} \abs{\norm{\sigma_{GM v} \Gamma_2^t(K,\sigma_u) \eta_t}^2  + 2 \langle \sigma_{GM v} \Gamma_1^t(K,\sigma_u)  x_1, \sigma_{GM v} \Gamma_2^t(K,\sigma_u) \eta_t \rangle  - \norm{\sigma_{GM v} \Gamma_2^t(K,\sigma_u) }_F^2}\\
         &\leq  \sup_{v \in \calS^{d_G}} \underbrace{\abs{\norm{\sigma_{GM v} \Gamma_2^t(K,\sigma_u) \eta_t}^2   - \norm{\sigma_{GM v} \Gamma_2^t(K,\sigma_u) }_F^2}}_{\mbox{quadratic form concentration}} +\underbrace{\abs{2 \langle \sigma_{GM v} \Gamma_1^t(K,\sigma_u)  x_1, \sigma_{GM v} \Gamma_2^t(K,\sigma_u) \eta_t \rangle}}_{\mbox{cross term}}
    \end{align*}
    where $(i)$ follows from linearity of expectation and the substition rule of conditional expectation, combined with the fact that $\eta$ has mean zero and is independent of $x_1$. Equality $(ii)$ follows from the fact that $\eta_t$ is mean zero with identity covariance.

    To bound the concentration of the quadratic form, we invoke Theorem 6.3.2 of \cite{vershynin2020high}. In particular, there exists a universal constant $C>0$ such that %there exists some universal positive constant $C_1$ such that \Bruce{to track the constants, instead apply the Hanson wright inequality from the tutorial. }
    \begin{equation}
    \label{eq: quadratic form concentration}
    \begin{aligned}
        &\Pr\brac{\abs{\norm{\sigma_{GM v} \Gamma_2^t(K,\sigma_u) \eta_t}^2   - \norm{\sigma_{GM v} \Gamma_2^t(K,\sigma_u) }_F^2} \geq \frac{\rho}{2}} \\
        &\qquad \leq 2 \exp\paren{-C \frac{\min\curly{\rho^2/\norm{\sigma_{GM v} \Gamma_2^t(K,\sigma_u)}_F^2, \rho}}{\sigma^4 \norm{\sigma_{GM_k v} \Gamma_2^t(K,\sigma_u)}^2 }}.
    \end{aligned}
    \end{equation}
    
    For the cross term, we invoke a standard sub-Gaussian tail bound \citep{vershynin2020high} to show that for any fixed $v \in \calS^{d_G}$,
    \begin{equation}
    \begin{aligned}
        \label{eq: cross term concentration}
        &\Pr\brac{\abs{2 \langle \sigma_{GM v} \Gamma_1^t(K,\sigma_u)  x_1, \sigma_{GM v} \Gamma_2^t(K,\sigma_u) \eta_t \rangle} > \frac{\rho}{2}} \\
        &\qquad \leq 2 \exp\paren{-\frac{1}{16} \rho^2 \frac{1}{\sigma^2 \norm{\Gamma_2^t(K,\sigma_u)^\top  \sigma_{GM v}^\top  \sigma_{GM v} \Gamma_1^t(K,\sigma_u) x_1 }^2}}.
    \end{aligned}
    \end{equation}

    We now seek to simplify the denominators appearing in the exponents for each of the probability bounds above. First,  note that we have 
    \begin{align*}
        & \norm{\sigma_{GM v} \Gamma_1^t(K,\sigma_u) x_1 }^2 + \norm{\sigma_{GM v} \Gamma_2^t(K,\sigma_u)}_F^2 \\
        &= \trace\paren{\sigma_{GM v} \Gamma_1^t(K,\sigma_u) x_1x_1^\top \Gamma_1^t(K,\sigma_u)^\top \sigma_{GM v}^\top } + \trace\paren{\sigma_{GM v} \Gamma_2^t(K,\sigma_u)  \Gamma_2^t(K,\sigma_u)^\top \sigma_{GM v}^\top} \\
        &= \trace\paren{\sigma_{GM v} \paren{\Gamma_1^t(K,\sigma_u) x_1x_1^\top \Gamma_1^t(K,\sigma_u)^\top  + \Gamma_2^t(K,\sigma_u)  \Gamma_2^t(K,\sigma_u)^\top } \sigma_{GM v}^\top } \\
        %&\overset{(i)}{=} \E \brac{\trace\paren{\sigma_{GM v} \paren{\Gamma_1^t(K,\sigma_u) x_1x_1^\top \Gamma_1^t(K,\sigma_u)^\top + \sym\paren{\Gamma_1^t(K,\sigma_u) x_1\eta_t^\top \Gamma_2^t(K,\sigma_u)^\top}  + \Gamma_2^t(K,\sigma_u) \eta_t \eta^_t\top \Gamma_2^t(K,\sigma_u)^\top } \sigma_{GM v}^\top } \vert x_1} \\
        &=  \E \brac{\trace\paren{\sigma_{GM v} \VEC(X^\top) \VEC(X^\top)^\top  \sigma_{GM v}^\top }\vert x_1 },
    \end{align*} 
    where the final equality follows from the fact that $\eta_t$ is mean zero with identity covariance and independent form $x_1$ combined with the construction of $\Gamma_1^t(K,\sigma_u)$ and $\Gamma_2^t(K,\sigma_u)$. Employing the identity in \eqref{eq: vectorizing quadratic form}, we find
    \begin{align*}
        \E \brac{\trace\paren{\sigma_{GM v} \VEC(X^\top) \VEC(X^\top)^\top  \sigma_{GM v}^\top }\vert x_1} =  \E \brac{ v^\top M G^\top X^\top X G M v\vert x_1} = 1,
    \end{align*}
    where the final equality follows by the definition of $M$ and the fact that $\norm{v}^2=1$. The above computations provide identity $
        \norm{\sigma_{GM v} \Gamma_1^t(K,\sigma_u) x_1 }^2 + \norm{\sigma_{GM v} \Gamma_2^t(K,\sigma_u)}_F^2  = 1.$
    We therefore have $
        \norm{\sigma_{GM v} \Gamma_2^t(K,\sigma_u)}_F^2  \leq 1$ and $\norm{\sigma_{GM v} \Gamma_1^t(K,\sigma_u) x_1 }^2 \leq 1$.
    As a result, the quantity in the denominator of the exponential in \eqref{eq: cross term concentration} is bounded as
    \begin{align*}
        \norm{\Gamma_2^t(K,\sigma_u)^\top  \sigma_{GM v}^\top  \sigma_{GM v} \Gamma_1^t(K,\sigma_u) x_1 } & \leq  \norm{\sigma_{GM v} \Gamma_2^t(K,\sigma_u)} \norm{\sigma_{GM v} \Gamma_1^t(K,\sigma_u) x_1 } {\leq}   \norm{\sigma_{GM v} \Gamma_2^t(K,\sigma_u)}.
    \end{align*}
    Therefore, the probability bounds in \eqref{eq: quadratic form concentration} and \eqref{eq: cross term concentration} may be modified to read
    \begin{align*}
        \Pr\brac{\abs{\norm{\sigma_{GM v} \Gamma_2^t(K,\sigma_u) \eta_t}^2   - \norm{\sigma_{GM v} \Gamma_2^t(K,\sigma_u) }_F^2} \geq \frac{\rho}{2}} &\leq 2 \exp\paren{-C \frac{\min\curly{\rho, \rho^2}}{\sigma^4 \norm{\sigma_{GM v} \Gamma_2^t(K,\sigma_u)}^2}} \\
        \Pr\brac{\abs{2 \langle \sigma_{GM v} \Gamma_1^t(K,\sigma_u)  x_1, \sigma_{GM v} \Gamma_2^t(K,\sigma_u) \eta_t \rangle} > \frac{\rho}{2}} &\leq 2 \exp\paren{-\frac{1}{16} \rho^2 \frac{1}{\sigma^2 \norm{\sigma_{GM v} \Gamma_2^t(K,\sigma_u)}^2}}.
    \end{align*}
    Union bounding over the above two events, and using the assumption that $\sigma^2 \geq 1$ implies that for some universal constant $C > 0$, 
    \begin{align*}
        &\Pr\brac{\abs{\norm{\sigma_{GM v} \Gamma_2^t(K,\sigma_u) \eta_t}^2   - \norm{\sigma_{GM v} \Gamma_2^t(K,\sigma_u) }_F^2} + \abs{2 \langle \sigma_{GM v} \Gamma_1^t(K,\sigma_u)  x_1, \sigma_{GM v} \Gamma_2^t(K,\sigma_u) \eta_t \rangle} > \rho } \\
        &\leq 4 \exp\paren{-C\frac{\min\curly{\rho, \rho^2}}{\sigma^4 \norm{\sigma_{GM v} \Gamma_2^t(K,\sigma_u)}^2}}.
    \end{align*}
    We may now invoke the covering argument in \Cref{lem: symmetric covering} with $\varepsilon=\frac{1}{4}$ to show that there exists a universal positive constant $C$ such that
    \begin{align} 
        \label{eq: prob bound on QI}
        &\Pr\brac{\norm{M G^\top X^\top X G M - I} > \rho} \leq 4 \times 9^{d_G} \exp\paren{-C  \frac{\min\curly{\rho,\rho^2}}{\sigma^4 \norm{\sigma_{GM v} \Gamma_2^t(K,\sigma_u)}^2}}. 
    \end{align}

    To conclude, we must bound the term $\norm{\sigma_{GM v} \Gamma_2^t(K,\sigma_u)}$. To do so, we write
    \begin{align*}
        \norm{\sigma_{GM v} \Gamma_2^t(K,\sigma_u)} &= \frac{1}{t}\norm{\sigma_{GM v} H \Xi_2^t(K,\sigma_u)} \\
        &= \frac{1}{\sqrt{t}}\norm{\sigma_{GM v} H \paren{\paren{I_t \otimes \bmat{I \\ K}} F_2^t(K,\sigma_u)   + \paren{I_t \otimes \bmat{0 & 0 \\ 0 & \sigma_u I_{\du}}}} } \\
        &\leq \frac{1}{\sqrt{t}}\paren{\underbrace{\norm{\sigma_{GM v} H \paren{I_t \otimes \bmat{I \\  K}} F_2^t(K,\sigma_u)}}_{\mbox{state term}} + \underbrace{\norm{ \sigma_{GM v} H \paren{I_t \otimes \bmat{0 & 0 \\ 0 & \sigma_u I_{\du}}}}}_{\mbox{input noise term}} }.
    \end{align*}
    For the state term, we have 
    \begin{align*}
        \norm{\sigma_{GM v} H \paren{I_t \otimes \bmat{I \\ K}} \Psi_1} \leq \norm{\sigma_{GM v} H \paren{I_t \otimes \bmat{I \\ K}}} \norm{F_2^t(K,\sigma_u)}.
    \end{align*}
    The second term quantifies the system's amplification of an exogenous input, and is bounded by the $\calH_\infty$ norm of the system $\norm{\left[\begin{array}{c|c} A_K & \bmat{  I & \sigma_u B^\star} \end{array}\right]}_{\calH_\infty}$ \citep[Corollary 4.2]{tilli1998singular}. This, in turn,  is bounded by the $\calH_\infty$ norm of the matrix $A_K$ multiplied by the spectral norm of the input matrix, yielding
    \begin{align*}
        \norm{F_2^t(K,\sigma_u)} \leq \norm{A_K}_{\calH_\infty} \paren{ \norm{B^\star} + 1} % \leq 4 \norm{P^\star}^{3/2} \paren{1 + \norm{B^\star}} \leq 4 \norm{P^\star}^{3/2} \Psi_{B^\star}.
    \end{align*}
    We may bound the first quantity as 
    \begin{align*}
        &\norm{\sigma_{GM v} H \paren{I_t\otimes \bmat{I \\ K}}} \\
        &=\norm{(I_{t \times \dc} \otimes (G M v)^\top ) \paren{I_t \otimes \bmat{I_{\dx+\du} \otimes e_1 \\ \vdots \\I_{\dx + \du} \otimes e_{\dc}}} \paren{I_t \otimes \bmat{I \\ K}}} \\
        &\overset{(i)}{=} \norm{(I_{\dc} \otimes (G M v)^\top )  \bmat{I_{\dx+\du} \otimes e_1 \\ \vdots \\I_{\dx + \du} \otimes e_{\dc}}  \bmat{I \\ K}} \\
        &\overset{(ii)}{=} \sup_{w \in \calS^{\dx}} \norm{(I_{\dc} \otimes (G M v)^\top )  \bmat{I_{\dx+\du} \otimes e_1 \\ \vdots \\I_{\dx + \du} \otimes e_{\dc}}  \bmat{I \\ K}w} \\
        &\overset{(iii)}{=} \sup_{w \in \calS^{\dx}} \norm{\VEC v^\top M G^\top \VEC^{-1} \paren{\bmat{I_{\dx+\du} \otimes e_1 \\ \vdots \\I_{\dx + \du} \otimes e_{\dc}}  \bmat{I \\ K}w}} \\
        &\overset{(iv)}{=}\sup_{w \in \calS^{\dx}} \norm{\VEC v^\top M G^\top\paren{\paren{\bmat{I \\ K}w} \otimes I_{\dc}}} 
    \end{align*}
    where $(i)$ follows by the following kronecker product identities: $I_{a + b} \otimes W = I_a \otimes \paren{I_b \otimes W}$, $(A \otimes B)(C \otimes D) = AC \otimes BD$, and $\norm{I \otimes W} = \norm{W}$. Equality $(ii)$ follows by definition of the operator norm, and equality $(iii)$ follows by the vectorization identity $\VEC(WYZ) = (Z^\top \otimes W) \VEC(Y) $. Equality $(iv)$ follows by the definition of the inverse vectorization operation.  The quantity in the norm above may now be written as the inner product of a vector with itself. In particular, we have
    \begin{align*}
   &\sup_{w \in \calS^{\dx}} \norm{\VEC v^\top M G^\top\paren{\paren{\bmat{I \\K}w} \otimes I_{\dc}}} \\  &=\sup_{w \in \calS^{\dx}} \sqrt{ v^\top M G^\top\paren{\paren{\bmat{I \\ K}ww^\top \bmat{I \\K}^\top} \otimes I_{\dc}} GM v }\\
         &\overset{(i)}{\leq}\sup_{w \in \calS^{\dx}} \sqrt{ v^\top M G^\top\paren{\paren{\bmat{I \\ K} \bmat{I \\ K}^\top} \otimes I_{\dc}} GM v } \\
         &\overset{(ii)}{\leq} 2 \sup_{w \in \calS^{\dx}} \sqrt{ v^\top M G^\top\paren{\Sigma^t(K,\sigma_u,x_1) \otimes I_{\dc}} GM v }\leq 2,
    \end{align*}
    where inequality $(i)$ follows by the fact that $ww^\top \preceq I$, and inequailty $(ii)$ follows from the fact that $\bmat{I \\K}\bmat{I \\ K}^\top \preceq 2 \Sigma^t(K,\sigma_u,x_1)$, as may be verified by examining the expression for $\bar \Sigma^t(K,\sigma_u,x_1)$ in \Cref{lem: covariance facts}.  
     We may similarly show that the input term is bounded by one:
    \begin{align*}
         \norm{ \sigma_{GM v} H \paren{I_t \otimes \bmat{0 & 0 \\ 0 & \sigma_u I_{\du}}}} \leq 1.
    \end{align*}
    \sloppy Combining results we have the bound $\norm{\sigma_{GM v} \Gamma_2^t(K,\sigma_u)} \leq \frac{1}{\sqrt{t}} \paren{1 + 2 \norm{F_2^t(K,\sigma_u)}} \leq \frac{5}{  \sqrt{t}} \norm{P_K}^{3/2} \Psi_{B^\star,\sigma_w}$, where we used the result of \Cref{lem: powers of closed loop by lyap} to bound $\norm{A_K}_{\calH_\infty}$ in terms of $\norm{P_K}$. Substituting this result into the bound in \eqref{eq: prob bound on QI}, we find that for some universal constant $C$, 
    \begin{align*}
        \Pr\brac{\norm{M G^\top X^\top X G M - I} > \rho} \leq 4 \times 9^{d_G} \exp\paren{-C \frac{d \min\curly{\rho,\rho^2} t}{\sigma^4 \norm{P_K}^3  \Psi_{B^\star }^2}},
    \end{align*}
    which concludes our proof. 
\end{proof}

\begin{lemma}[Covariance Concentration]
    \label{lem: covariance concentration}
    For a sufficiently large universal positive constant $c$, we have that as long as $t \geq c \sigma^4 \norm{P_K}^3 \Psi_{B^\star}^2 \paren{ \dtheta + \log\frac{1}{\delta}}$, the event
    \begin{align*}
        \calE_{\mathsf{conc}} = \curly{\frac{1}{2} \hat \Phi^\top \paren{\Sigma^t(K,\sigma_u,x_1) \otimes I_{\dx}} \hat \Phi \preceq \frac{1}{t} \sum_{s=1}^{t} \hat \Phi^\top \paren{\bmat{x_s \\ u_s}\bmat{x_s \\ u_s}^\top \otimes I_{\dx}} \hat \Phi \preceq 2 \hat \Phi^\top \paren{\Sigma^t(K,\sigma_u,x_1) \otimes I_{\dx}} \hat \Phi}
    \end{align*}
    holds with probability at least $1-\frac{\delta}{2}$. 
\end{lemma}
\begin{proof}
    The result follows by invoking \Cref{lem: approximate isometries} with $G\gets \hat \Phi$ and $\rho \gets \frac{1}{2}$. We invert the probability bound by setting $
        \frac{\delta}{2} = 4 \times 9^{\dtheta} \exp\paren{- C \frac{t}{\sigma^4 \norm{P_K}^3 \Psi_{B^\star}^2}}.$
   As long as $t\geq c \sigma^4 \norm{P_K}^3 \Psi_{B^\star}^2 \paren{ \dtheta + \log\frac{1}{\delta}}$
    for a sufficiently large universal  constant $c>0$, the event holds with probability at least $1-\frac{\delta}{2}$. 
\end{proof}

\subsection{Self-Normalized Martingales}

To present the self-normalized martingale bound we first recall the definitions of a filtration and an adapted process.  

\begin{definition}(Filtration and Adapted Process)
    A sequence of sub-$\sigma$-algebras $\curly{\calF_s}_{s=1}^t$ is said to be a \emph{filtration} if $\calF_s \subseteq \calF_k$ for $s\leq k$. A stochastic process $\curly{w_s}_{s=1}^t$ is said to be \emph{adapted} to the filtration $\curly{\calF_s}_{s=1}^t$ if for all $s\geq 1$, $w_s$ is $\calF_s$-measurable.% i.e. $\E\paren{W_t \vert \calF_t} = W_t$. 
\end{definition}

The following theorem is a modification of Theorem 3.4 in \cite{abbasi2011regret} and is a consequence of Lemma 14.7 in \cite{pena2009self}. 

\begin{theorem}
    \label{thm: self normalized martingale}
    Let $\curly{\calF_s}_{s=0}^{t}$ be a filtration and let $\curly{z_s}_{s=1}^t$ be a matrix valued stochastic process assuming values in $\R^{\dx \times \dtheta}$ which is adapted to $\curly{\calF_{s-1}}_{s=1}^t$ and $\curly{w_s}_{s=1}^t$ a vector valued stochastic process assuming values in $\R^{\dx}$ which is adapted to $\curly{\calF_s}_{s=1}^t$. Additionally suppose that for all $ 1\leq s \leq  t$, $w_t$ is $\sigma^2$-conditionally sub-Gaussian with respect to $\calF_s$. Let $\Sigma$ be a positive definite matrix in $\R^{\dtheta\times\dtheta}$. Then for $\delta \in (0,1)$, the following inequality holds with probability at least $1-\delta$:
    \begin{align*}
        \norm{\paren{ \Sigma + \sum_{s=1}^t z_s^\top z_s}^{-1/2} \sum_{s=1}^t z_s^\top w_s}^2 \leq \sigma^2 \log\paren{\frac{\det\paren{\Sigma + \sum_{s=1}^t z_s^\top z_s}}{\det(\Sigma)}} + 2 \sigma^2 \log\frac{1}{\delta}. 
    \end{align*}
\end{theorem}
\begin{proof}
    \sloppy The proof follows by applying Lemma 4.1 from \cite{ziemann2023tutorial} with $P \gets \sum_{s=1}^t \frac{z_s^\top w_s}{\sigma}$ and $Q \gets \sum_{s=1}^t z_s^\top z_s$. To verify the conditions for this lemma to hold, we must verify that
    % \begin{align*}
    $
        \max_{\lambda \in \R^{\dtheta}} \E \exp \paren{\lambda^\top P - \frac{1}{2}\lambda^\top Q \lambda} \leq 1.
    $
    % \end{align*}
    To verify this condition, we define the quantity 
    \begin{align*}
        M_s(\lambda) \triangleq \exp\paren{\sum_{k=1}^{s} \brac{\frac{\lambda^\top z_k^\top w_k}{\sigma} - \frac{1}{2} \lambda^\top z_k^\top z_k \lambda}}. 
    \end{align*}
    Observe that for any $\lambda\in\R^{\dtheta}$, $M_0(\lambda) = 1$. For any $0 < s \leq t$, 
    \begin{align*}
        \E M_s(\lambda) &= \E \E \paren{M_s(\lambda \vert \calF_s} = \E \paren{M_{s-1}(\lambda) \E\paren{\exp\paren{\frac{\lambda\top z_s^\top w_s}{\sigma} - \frac{1}{2} \lambda^\top z_s^\top z_s \lambda} \vert \calF_s}}.
    \end{align*}
    Note that by the fact that $w_s$ is $\sigma^2$-conditionally sub-Gaussian with respect to $\calF_s$,
    \begin{align*}
        \E\paren{\exp\paren{\frac{\lambda^\top z_s^\top w_s}{\sigma} - \frac{1}{2} \lambda^\top z_s^\top z_s \lambda} \vert \calF_t} &= \E\paren{\exp\paren{\frac{\lambda^\top z_s^\top w_s}{\sigma}\vert \calF_s}} \exp\paren{ - \frac{1}{2} \lambda^\top z_s^\top z_s \lambda} \\
        &\leq \exp\paren{\lambda^\top z_s^\top z_s \lambda}\exp\paren{ - \frac{1}{2} \lambda^\top z_s^\top z_s \lambda} = 1.
    \end{align*}
    Therefore $\E M_t \leq \E M_{t-1}(\lambda) \leq \dots \leq  \E M_{1}(\lambda)\leq \E M_0(\lambda) =  1$, which verifies the requisite condition.    
    
\end{proof}

\subsection{Variance Bound}
\begin{lemma}
    \label{lem: variance by empirical cov}
    Let $\Sigma \in \R^{\dtheta\times\dtheta}$ be positive definite. Suppose data $x_{1:t+1}, u_{1:t}$ is generated by \eqref{eq: rollout under K}. Consider generating the least squares estimate $\hat \theta, \Lambda = \texttt{LS}(\hat \Phi, x_{1:t+1}, u_{1:t})$. Suppose that $\Lambda \succeq \Sigma$. There exists an event $\calE_{\mathsf{var}}$ that holds with probability at least $1-\delta/2$ under which the variance satisfies
    \begin{align*}
        &\norm{\hat \Phi \hat \theta - \hat \Phi \bar \theta}^2 \lesssim \frac{\sigma^2}{\lambda_{\min}(\Sigma)} \log\paren{ \frac{\det(\Sigma+\Lambda)}{\det(\Sigma) \delta^2} }+  \frac{\mathrm{Conc}_{\mathsf{cov}}(\delta)^2}{\lambda_{\min}(\Sigma)^2} \norm{\Phi^\star \theta^\star - \hat \Phi \bar \theta}^2,
    \end{align*}
    where 
    \begin{align*}
        \mathrm{Conc}_{\mathsf{cov}}(\delta) = \norm{\Sigma^t(K,\sigma_u,x_1)} \max \bigg\{&\sqrt{t}  \sigma^2  \norm{P_K}^{3/2} \Psi_{B^\star} \sqrt{\dx+\du + \log \frac{1}{\delta}}, \\
        &\sigma^4  \norm{P_K}^3 \Psi_{B^\star}^2 \paren{\dx+\du + \log \frac{1}{\delta}} \bigg\}.
    \end{align*}    .
\end{lemma}
\begin{proof} By \Cref{lem: ls solution}, $\hat \theta$ solves \eqref{eq: ls problem}. Therefore, 
    \begin{align*}
        \sum_{s=1}^t \norm{x_{s+1} - \VEC^{-1}(\hat \Phi \hat \theta) \bmat{x_s \\ u_s}}^2 \leq \sum_{s=1}^{t} \norm{x_{s+1} - \VEC^{-1}(\hat \Phi \bar \theta) \bmat{x_s \\ u_s}}^2. 
    \end{align*}
    Inserting $x_{s+1} = \VEC^{-1} (\Phi^\star \theta^\star) \bmat{x_s \\ u_s} + w_s$, we may rearrange the above inequality to read
    \begin{align*}
        &\sum_{s=1}^t \norm{\VEC^{-1} \paren{\hat \Phi \bar \theta - \hat \Phi \hat \theta}\bmat{x_s \\ u_s}}^2 \leq  \\
        &2 \sum_{s=1}^t\bigg \langle \VEC^{-1}\paren{ \hat \Phi \paren{\hat \theta - \bar \theta}}\bmat{x_s \\ u_s}, \VEC^{-1} \paren{\Phi^\star \theta^\star - \hat \Phi \bar \theta}\bmat{x_s \\ u_s} \bigg\rangle \\
        &+ 2 \sum_{s=1}^t \bigg\langle \VEC^{-1}\paren{ \hat \Phi \paren{\hat \theta - \bar \theta}}\bmat{x_s \\ u_s}, w_s \bigg \rangle. 
    \end{align*}
    Multiplying the above inequality by two, and bringing one of the sums from the right side to the left leaves us with 
    \begin{equation}
    \label{eq: variance decomposition}
   \begin{aligned}
        &\sum_{s=1}^t \norm{\VEC^{-1} \paren{\hat \Phi \bar \theta - \hat \Phi \hat \theta}\bmat{x_s\\ u_s}}^2 \leq \\
        &\underbrace{4 \sum_{s=1}^t\bigg \langle \VEC^{-1}\paren{ \hat \Phi \paren{\hat \theta - \bar \theta}}\bmat{x_s \\ u_s}, \VEC^{-1} \paren{\Phi^\star \theta^\star - \hat \Phi \bar \theta}\bmat{x_s \\ u_s} \bigg\rangle}_{\mbox{variance from misspecification}} \\
        &+ \underbrace{4 \sum_{s=1}^t \bigg\langle \VEC^{-1}\paren{ \hat \Phi \paren{\hat \theta - \bar \theta}}\bmat{x_s \\ u_s}, w_s \bigg \rangle - \sum_{s=1}^t \norm{\VEC^{-1} \paren{\hat \Phi (\bar \theta  -\hat \theta)}\bmat{x_s \\ u_s}}^2 }_{\mbox{self-normalized process}}.
    \end{aligned}
    \end{equation}
    \paragraph{Bound on Self-normalized Process: } We first bound the self-normalized process above. To simplify notation, we denote $z_s = \paren{\bmat{x_s \\ u_s}^\top \otimes I_{\dx}} \hat \Phi$. By using the identity $\VEC(XYZ) = (X^\top \otimes Z) \VEC(Y)$, we may therefore write the self-normalized process term above as  
    \begin{align*}
       4  \sum_{s=1}^t \bigg\langle z_t (\hat \theta - \bar \theta ), w_s \bigg \rangle - \sum_{s=1}^t\norm{z_s (\bar \theta   -\hat \theta )}^2 &\leq \sup_{\theta}  \curly{4 \sum_{s=1}^t \bigg\langle z_s \theta, w_s \bigg \rangle - \sum_{s=1} ^t \norm{z_s \theta}^2} = 4\norm{\Lambda ^{-1/2} \sum_{s=1}^t z_s^\top w_s}^2.
    \end{align*}
    Using the lower bound $\Lambda  \succeq \Sigma$, we have
    \begin{align*}
       4 \norm{\Lambda ^{-1/2} \sum_{s=1}^t z_t^\top w_s}^2 \leq 8 
       \norm{\paren{\Lambda +\Sigma}^{-1/2} \sum_{s=1}^t z_s^\top w_s}^2.
    \end{align*}
    Invoking the self-normalized martingale bound in \Cref{thm: self normalized martingale} we find that there exists an event $\calE_{\mathsf{noise}}$ that holds with probability at least $1-\delta/4$ such that under the event $\calE_{\mathsf{noise}}$, 
    \begin{align*}
        \norm{\paren{\Lambda  + \Sigma}^{-1/2} \sum_{s=1}^t z_s^\top w_s}^2 \leq \sigma^2 \log\paren{ \frac{ 16 \det(\Sigma+\Lambda) }{\det(\Sigma) \delta^2} }.
    \end{align*}
    
    \paragraph{Bound on Variance from misspecification: }
    We now proceed to bound the term in \eqref{eq: variance decomposition} arising due to variance from the model mispecification. To do so, note that
    \begin{align*}
        &\sum_{s=1}^t\bigg \langle \VEC^{-1}\paren{ \hat \Phi \paren{\hat \theta   - \bar \theta }}\bmat{x_s \\ u_s}, \VEC^{-1} \paren{\Phi^\star \theta^\star - \hat \Phi \bar \theta  }\bmat{x_s \\ u_s} \bigg\rangle \\
        &= \paren{\Phi^\star \theta^\star - \hat \Phi \bar \theta }^\top \paren{\sum_{s=1}^t \paren{\bmat{x_s\\ u_s} \bmat{x_s \\ u_s}^\top \otimes I_{\dx}}} \paren{\hat \Phi (\hat \theta  - \bar \theta )} \\
        &= t \paren{\Phi^\star \theta^\star - \hat \Phi \bar \theta }^\top \paren{\Sigma^t(K,\sigma_u,x_1) \otimes I_{\dx}} \paren{\hat \Phi (\hat \theta - \bar \theta)} \\
        &+ \paren{\Phi^\star \theta^\star - \hat \Phi \bar \theta }^\top \paren{\paren{\sum_{s=1}^t\bmat{x_s\\ u_s} \bmat{x_s \\ u_s}^\top - t \Sigma^t(K,\sigma_u,x_1)} \otimes I_{\dx}}\paren{\hat \Phi (\hat \theta - \bar \theta)}. 
    \end{align*}
    By the optimality of $\bar \theta$ in \eqref{eq: optimal projection}, 
    \[ 
        \paren{\Phi^\star \theta^\star - \hat \Phi \bar \theta }^\top \paren{\sum_{s=1}^t \paren{\Sigma^t(K,\sigma_u) \otimes I_{\dx}}} \paren{\hat \Phi (\hat \theta  - \bar \theta )} = 0.
    \]  
    Therefore, we are left with 
     \begin{align*}
        & \sum_{s=1}^t\bigg \langle \VEC^{-1}\paren{ \hat \Phi \paren{\hat \theta  - \bar \theta }}\bmat{x_s \\ u_s}, \VEC^{-1} \paren{\Phi^\star \theta^\star - \hat \Phi \bar \theta }\bmat{x_s \\ u_s} \bigg\rangle \\
        &=  \paren{\Phi^\star \theta^\star - \hat \Phi \bar \theta_k}^\top \paren{\paren{\sum_{s=1}^t\bmat{x_s\\ u_s} \bmat{x_s \\ u_s}^\top - t \Sigma^t(K,\sigma_u,x_1)} \otimes I_{\dx}}\paren{\hat \Phi (\hat \theta  - \bar \theta )} \\
        &\leq \norm{\Phi^\star \theta^\star - \hat \Phi \bar \theta  }\norm{\hat \Phi (\hat \theta   - \bar \theta )} \norm{\sum_{s=1}^t\bmat{x_s\\ u_s} \bmat{x_s \\ u_s}^\top - t \Sigma^t(K,\sigma_u,x_1)}.
    \end{align*}
    Invoking \Cref{lem: approximate isometries} with $G = I_{\dx+\du}$, we find that there exists an event $\calE_{\mathsf{cov}}$ that holds with probability at least $1-\delta/4$ such that under $\calE_{\mathsf{cov}}$, 
    \begin{align*}
        \norm{\sum_{s=1}^t\bmat{x_s\\ u_s} \bmat{x_s \\ u_s}^\top - t \Sigma^t(K,\sigma_u,x_1)} \lesssim \mathrm{Conc}_{\mathsf{\mathsf{cov}}}(\delta).
    \end{align*}  

    \paragraph{Combining bounds: }
    In light of the bounds on the self-normalized process and the variance from misspecification above, we define the event $\calE_{\mathsf{var}} = \calE_{\mathsf{noise}} \cap \calE_{\mathsf{cov}}$. This event holds with probability at least $1-\delta/2$ by a union bound. We may substitute the bounds arising under these events into into \eqref{eq: variance decomposition}, we find
    \begin{align*}
         \sum_{s=1}^t \norm{\VEC^{-1} \paren{\hat \Phi \bar \theta  - \hat \Phi \hat \theta }\bmat{x_s \\ u_s}}^2 &\lesssim \sigma^2 \log\paren{ \frac{\det(\Sigma+\Lambda) }{\det(\Sigma) \delta^2} } +\mathrm{Conc}_{\mathsf{cov}}(\delta)\norm{\Phi^\star \theta^\star - \hat \Phi \bar \theta }\norm{\hat \Phi (\hat \theta - \bar \theta )},
    \end{align*}
    where we absorbed the $16$ from the $\log$ in the first term into a universal constant. 
    The left side of the above inequality may be lower bounded by $\lambda_{\min}(\Sigma) \norm{\bar \theta  - \hat \theta }^2$. which may in turn be lower bounded by $\lambda_{\min}(\Sigma ) \norm{\hat \Phi \bar \theta  - \hat \Phi \hat \theta }^2$, by the fact that $\hat \Phi$ has orthonormal columns. We therefore obtain the following bound:
    \begin{align*}
        \norm{\hat \Phi \bar \theta  - \hat \Phi \hat \theta }^2 \lesssim \frac{\sigma^2}{\lambda_{\min}(\Sigma )} \log\paren{ \frac{ \det(\Sigma+\Lambda )}{\det(\Sigma) \delta^2} } + \frac{\mathrm{Conc}_{\mathsf{cov}}(\delta)}{\lambda_{\min}(\Sigma )} \norm{\Phi^\star \theta^\star - \hat \Phi \bar \theta }\norm{\hat \Phi (\hat \theta  - \bar \theta )}.
    \end{align*}
     Now, if we define 
    \begin{align*}
        a &\triangleq \frac{\mathrm{Conc}_{\mathsf{cov}}(\delta)}{\lambda_{\min}(\Sigma)} \norm{\Phi^\star \theta^\star - \hat \Phi \bar \theta}, \quad
        b \triangleq \frac{\sigma^2}{\lambda_{\min}(\Sigma )} \log\paren{ \frac{ \det(\Sigma+\Lambda )}{\det(\Sigma) \delta^2} }, \quad
        x \triangleq \norm{\hat \Phi \hat \theta  - \hat \Phi \bar \theta },
    \end{align*}
    then the above inequalities may be combined to read $x^2 \leq ax + b$, from which we can conclude $ x \leq \frac{1}{2} \paren{ a + \sqrt{a^2 + 4b}}$, 
    or equivalently, 
    \begin{align*}
        x^2 &\leq \frac{1}{4}\paren{2 a^2 + 4b + 2 \sqrt{a^4 + 4a^2 b}} \leq \frac{1}{4}\paren{2 a^2 + 4b + 2 a^2 + 4 a^2 + 4 b} = 2a^2 + 2 b.
    \end{align*}
    Subtituting in the original quantities of interest and absorbing the factor of two into the universal constant concludes the proof.     
     
\end{proof}

\begin{lemma}[Variance Bound]
    \label{lem: variance bound}
    Suppose $t \geq c \tau_{\ls}$ for $c$ and $\tau_{\ls}$ defined in \Cref{thm: ls estimation error}. 
    % $$\tau_{\mathsf{var}} \triangleq \max\curly{\sigma^4  (1+\norm{A_K}_{\calH_{\infty}})^2 \Psi_{B^\star}^2 \paren{\dx(\dx+\du) + \log \frac{1}{\delta}},  \norm{x_1}^2 \norm{P_K} + 1}$$  and a sufficiently large universal constant $c>0$.   
    Then there exists an event $\calE_{\ls}$ that holds with probability at least $1 - \delta$ under which
       \begin{align*}
        &\norm{\hat \Phi \hat \theta  - \hat \Phi \bar \theta }^2 \lesssim \frac{ \dtheta \sigma^2  }{t \lambda_{\min} \paren{\hat \Phi^\top \paren{\bar \Sigma^t(K,\sigma_u,x_1) \otimes I_{\dx}} \hat \Phi}} \log\paren{ \frac{1}{ \delta} } \\
        &+  \frac{\sigma^4 \norm{\Sigma^t(K,\sigma_u, x_1)}^2 \norm{P_K}^3 \Psi_{B^\star}^2  }{ t\lambda_{\min} \paren{\hat \Phi^\top \paren{\bar \Sigma^t(K,\sigma_u,x_1) \otimes I_{\dx}} \hat \Phi}^2 } \paren{\dx+\du + \log \frac{1}{\delta}}\norm{\Phi^\star \theta^\star - \hat \Phi \bar \theta }^2.
    \end{align*}
\end{lemma}

\begin{proof}   
    \sloppy By the lower bound on $t$, the event $\calE_{\mathsf{conc}}$ defined in \Cref{lem: covariance concentration} holds with probability at least $1-\delta/2$. Denoting $\Sigma \gets t \frac{1}{2} \hat \Phi^\top \paren{\bar \Sigma^t(K,\sigma_u,x_1) \otimes I_{\dx}} \hat \Phi$, we may invoke \Cref{lem: variance by empirical cov} to show that there exists an event $\calE_{\mathsf{var}}$ that holds with probability at least $1-\delta/2$ such that conditioned on $\calE_{ls} = \calE_{\mathsf{conc}}\cap\calE_{\mathsf{var}}$, 
     \begin{align*}
        \norm{\hat \Phi \hat \theta  - \hat \Phi \bar \theta }^2 &\lesssim \frac{\sigma^2}{\lambda_{\min}(\Sigma)} \log\paren{ \frac{\det(\Sigma+\Lambda)}{\det(\Sigma) \delta^2} }+  \frac{\mathrm{Conc}_{\mathsf{cov}}(\delta)^2}{\lambda_{\min}(\Sigma)^2} \norm{\Phi^\star \theta^\star - \hat \Phi \bar \theta}^2. 
    \end{align*}
    By a union bound, $\calE_{\ls}$ holds with probability at least $1-\delta$.
    Subsituting the definition of $\mathrm{Conc}_{\mathsf{cov}}$ from \Cref{lem: variance by empirical cov}, we may simplify the above bound to find
    \begin{align*}
        \norm{\hat \Phi \hat \theta  - \hat \Phi \bar \theta }^2 &\lesssim 
         \frac{\sigma^2}{\lambda_{\min}(\Sigma )} \log\paren{\frac{ \det(\Sigma+\Lambda )}{\det(\Sigma) \delta^2} } \\
        &+  \frac{ \sigma^4 \norm{\Sigma^t(K,\sigma_u,x_1)}^2 \norm{P_K}^3 \Psi_{B^\star}^2 t}{\lambda_{\min}(\Sigma   )^2} \paren{\dx+\du + \log \frac{4}{\delta}}\norm{\Phi^\star \theta^\star - \hat \Phi \bar \theta }^2,
    \end{align*}
    where we have used the fact that $$t \geq c\tau_{\ls} \geq  \sigma^4  \norm{P_K}^3 \Psi_{B^\star}^2 \paren{\dx+\du + \log \frac{1}{\delta}},$$ to show that $\mathrm{Conc}_{\mathsf{cov}}$  assumes the quantity that grows with $t$.
    
    The upper bound from the event $\calE_{\conc}$ provides 
    \[  
        \Lambda  \preceq 2 t \hat \Phi^\top \paren{ \Sigma^t(K,\sigma_u,x_1) \otimes I_{\dx}} \hat \Phi \preceq 2 t (1 + \norm{P_K}\frac{\norm{x_1}^2}{t-1}) \hat \Phi^\top \paren{\bar \Sigma^t(K,\sigma_u,x_1) \otimes I_{\dx}} \hat \Phi,
    \]
    where the final inequality follows from \Cref{lem: covariance facts} point 4.
    The $\log$ term may then be bounded as
    \begin{align*}
        \log\paren{\frac{ \det(\Sigma+\Lambda)}{\det(\Sigma) \delta^2}} \leq \dtheta \log\paren{4 \frac{1 + \norm{P_K} \frac{\norm{x_1}^2}{t-1}}{\delta^2}}
    \end{align*}
    %where the final inequality follows from \Cref{lem: transition bounds} under the event $\calE_{\mathsf{bound}}$. 
    By the fact that $t \geq c \tau_{\ls} \geq \norm{x_1}^2 \norm{P_K} + 1$, we have $\frac{\norm{P_K} \norm{x_1}^2}{t-1} \leq 1$. Therefore, the above quantity is bounded by $\dtheta \log\frac{8}{\delta^2} \lesssim \dtheta \log \frac{1}{\delta}$. 
    
    Substituting the expression for $\Sigma$ tells us that %and \Cref{lem: covariance facts} to find
    \begin{align*}
        \lambda_{\min}(\Sigma ) =  \frac{t}{2}\lambda_{\min} \paren{\hat \Phi^\top \paren{\bar \Sigma^t(K,\sigma_u,x_1) \otimes I_{\dx}} \hat \Phi}. %\frac{t\sigma_u^2 }{2(1+\norm{K}^2 + \sigma_u^2).}
    \end{align*}
    Substituing this into the bound on $\norm{\hat \Phi \hat \theta - \hat \Phi \bar \theta}^2$ along with the upper bound on $\log\paren{\frac{\det(\Sigma+\Lambda)}{\det (\Sigma) \delta^2}}$ completes the proof. 
    % Therefore under the event $\calE_{\ls}$,
    % \begin{align*}
    %     &\norm{\hat \Phi \hat \theta  - \hat \Phi \bar \theta }^2 \leq \frac{ 16\dtheta \sigma^2  }{t \lambda_{\min} \paren{\hat \Phi^\top \paren{\bar \Sigma^t(K,\sigma_u,x_1) \otimes I_{\dx}} \hat \Phi}} \log\paren{ \frac{128}{ \delta} } \\
    %     &+ 2000000  \frac{\sigma^4 \norm{\Sigma^t(K,\sigma_u, x_1)}^2 \norm{P_K}^3 \Psi_{B^\star}^2  }{ t\lambda_{\min} \paren{\hat \Phi^\top \paren{\bar \Sigma^t(K,\sigma_u,x_1) \otimes I_{\dx}} \hat \Phi}^2 } \paren{\dx(\dx+\du) + \log \frac{4}{\delta}}\norm{\Phi^\star \theta^\star - \hat \Phi \bar \theta }^2.
    % \end{align*}

    \end{proof}

    \subsection{Bias Bound}
    % \Bruce{
    % Ways to avoid bias growing with $T$:
    % \begin{itemize}
    %     \item Incorporate regularization into the least squares problem. Then $\bar \theta$ becomes the solution to 
    %     \begin{align*}
    %         (\hat \Phi\theta - \Phi^\star \theta^\star)^\top\paren{\Sigma^t(K,\sigma_u,x_1)} (\hat\Phi \theta - \Phi^\star \theta^\star) + \lambda \norm{\hat \Phi \theta}^2
    %     \end{align*}
    %     Here, the estimate gets an $\lambda \times$ identity added to the covariance, which reduces the bias to something proportional to $\lambda_{\min}(\Sigma)+\lambda$, but it also leads to an additional bias term proportional to $\lambda$. Balancing these results in a bias term proportional to $\varepsilon$ rather than $\varepsilon^2$, just like adding noise would do. 
    %     \item 
    %     \item Importance sampling
    % \end{itemize}
    % }
    
    \begin{lemma}
        \label{lem: bias bound}
        The bias of the estimate arising in \Cref{lem: bias variance decomposition} is bounded as
        \begin{align*}
            \norm{\hat \Phi \bar \theta - \Phi^\star \theta^\star}^2 \leq 4 \frac{\norm{\Sigma^t(K,\sigma_u,x_1)}}{\lambda_{\min}(\hat \Phi^\top (\bar \Sigma^t(K,\sigma_u,x_1) \otimes I_{\dx}) \hat \Phi)} d(\hat \Phi, \Phi^\star)^2 \norm{\theta^\star}^2.
        \end{align*}
    \end{lemma}

    \begin{proof}
        
    We first left multiply the quantity in the norm by the identity defined in terms of the matrix $\hat \Phi \hat \Phi^\top \paren{\Sigma^t(K,\sigma_u,x_1) \otimes I_{\dx}} \hat \Phi \hat \Phi^\top  + \hat \Phi_\perp \hat \Phi_\perp^\top$ left multiplied by its inverse:
    \begin{align*}
        &\norm{\Phi^\star \theta^\star - \hat \Phi \bar \theta}^2 =\bigg\|(\hat \Phi \hat \Phi^\top \paren{\Sigma^t(K,\sigma_u,x_1) \otimes I_{\dx}} \hat \Phi \hat \Phi^\top  + \hat \Phi_\perp \hat \Phi_\perp^\top)^{-1} \\
        &\quad \times\paren{\hat \Phi \hat \Phi^\top \paren{\Sigma^t(K,\sigma_u,x_1) \otimes I_{\dx}} \hat \Phi \hat \Phi^\top  + \hat \Phi_\perp \hat \Phi_\perp^\top} \paren{\Phi^\star \theta^\star- \hat \Phi \bar \theta }\bigg\|^2.
    \end{align*}
    To simplify the above quantity, firstly observe that the term 
    \begin{align}
        \label{eq: bias term orthogonal reduction}
        \hat \Phi_\perp \hat \Phi_\perp^\top \paren{\Phi^\star \theta^\star - \hat \Phi \bar \theta } = \hat \Phi_\perp \hat \Phi_\perp^\top \Phi^\star \theta^\star
    \end{align}
    by the fact that $\hat \Phi_\perp^\top \hat \Phi = 0$. Second, note that the term
    \begin{align*}
        &\hat \Phi \hat \Phi^\top \paren{\Sigma^t(K,\sigma_u,x_1) \otimes I_{\dx}} \hat \Phi \hat \Phi^\top \paren{\Phi^\star \theta^\star - \hat \Phi \bar \theta} \\&= \hat \Phi \hat \Phi^\top \paren{\Sigma^t(K,\sigma_u,x_1) \otimes I_{\dx}} \hat \Phi \hat \Phi^\top \Phi^\star \theta^\star - \hat \Phi \hat \Phi^\top \paren{\Sigma^t(K,\sigma_u,x_1) \otimes I_{\dx}} \hat \Phi \bar \theta
    \end{align*}
    by the fact that the columns of $\hat \Phi$ are orthonormal. We now substitute the optimal projection
    \[
        \bar \theta  = \paren{\hat \Phi^\top \paren{\Sigma^t(K,\sigma_u,x_1) \otimes I_{\dx}} \hat \Phi}^{-1}\hat \Phi^\top \paren{\Sigma^t(K,\sigma_u,x_1) \otimes I_{\dx}} \Phi^\star \theta^\star
    \]
    into the above to find that 
    \begin{align*}
        \hat \Phi \hat \Phi^\top &\paren{\Sigma^t(K,\sigma_u,x_1) \otimes I_{\dx}} \hat \Phi \hat \Phi^\top \paren{\Phi^\star \theta^\star - \hat \Phi \bar \theta } \\
        &=\hat \Phi \hat \Phi^\top \paren{\Sigma^t(K,\sigma_u,x_1) \otimes I_{\dx}} \hat \Phi \hat \Phi^\top \Phi^\star \theta^\star \\&\quad- \hat \Phi \hat \Phi^\top \paren{\Sigma_k \otimes I_{\dx}} \hat \Phi  \paren{\hat \Phi^\top \paren{\Sigma^t(K,\sigma_u,x_1) \otimes I_{\dx}} \hat \Phi}^{-1}\hat \Phi^\top \paren{\Sigma^t(K,\sigma_u,x_1) \otimes I_{\dx}} \Phi^\star \theta^\star \\
        &=  \hat \Phi \hat \Phi^\top \paren{\Sigma^t(K,\sigma_u,x_1) \otimes I_{\dx}} \hat \Phi \hat \Phi^\top \Phi^\star \theta^\star - \hat \Phi \hat \Phi^\top \paren{\Sigma^t(K,\sigma_u,x_1) \otimes I_{\dx}} \Phi^\star \theta^\star \\
        &=  \hat \Phi \hat \Phi^\top \paren{\Sigma^t(K,\sigma_u,x_1) \otimes I_{\dx }} \paren{ \hat \Phi \hat \Phi^\top  - I_{\dx}} \Phi^\star \theta^\star  \\
        &=  -\hat \Phi \hat \Phi^\top \paren{\Sigma^t(K,\sigma_u,x_1) \otimes I_{\dx}} \hat \Phi_\perp \hat \Phi_\perp^\top  \Phi^\star \theta^\star.
    \end{align*}
    Combining the above result with \eqref{eq: bias term orthogonal reduction}, we find that
    \begin{equation}
    \label{eq: bias bound main ineq}
    \begin{aligned}
        &\norm{\Phi^\star \theta^\star - \hat \Phi \bar \theta }^2 \\
        &= \norm{(\hat \Phi \hat \Phi^\top \paren{\Sigma^t(K,\sigma_u,x_1) \otimes I_{\dx}} \hat \Phi \hat \Phi^\top  + \hat \Phi_\perp \hat \Phi_\perp^\top)^{-1} \paren{I - \hat \Phi \hat \Phi^\top \paren{\Sigma^t(K,\sigma_u,x_1) \otimes I_{\dx}}} \hat \Phi_\perp \hat \Phi_\perp^\top \Phi^\star \theta^\star}^2 \\
        & \leq \norm{(\hat \Phi \hat \Phi^\top \paren{\Sigma^t(K,\sigma_u,x_1) \otimes I_{\dx}} \hat \Phi \hat \Phi^\top  + \hat \Phi_\perp \hat \Phi_\perp^\top)^{-1} \paren{I - \hat \Phi \hat \Phi^\top \paren{\Sigma^t(K,\sigma_u,x_1) \otimes I_{\dx}}} \hat \Phi_\perp}^2 \norm{\hat \Phi_\perp^\top \Phi^\star}^2 \norm{\theta^\star}^2.
    \end{aligned}
    \end{equation}
    Application of the triangle inequality and Cauchy-Schwarz yields
    \begin{equation}
        \label{eq: decompose orthogonal and nonorthogonal}
    \begin{aligned}
        &\norm{(\hat \Phi \hat \Phi^\top \paren{\Sigma^t(K,\sigma_u,x_1) \otimes I_{\dx}} \hat \Phi \hat \Phi^\top  + \hat \Phi_\perp \hat \Phi_\perp^\top)^{-1} \paren{I - \hat \Phi \hat \Phi^\top \paren{\Sigma^t(K,\sigma_u,x_1) \otimes I_{\dx}}} \hat \Phi_\perp}^2 \\
        &\qquad \leq 2 \norm{(\hat \Phi \hat \Phi^\top \paren{\Sigma^t(K,\sigma_u,x_1) \otimes I_{\dx}} \hat \Phi \hat \Phi^\top  + \hat \Phi_\perp \hat \Phi_\perp^\top)^{-1} \hat \Phi_\perp}^2 \\
        &\qquad + 2 \norm{(\hat \Phi \hat \Phi^\top \paren{\Sigma^t(K,\sigma_u,x_1) \otimes I_{\dx}} \hat \Phi \hat \Phi^\top  + \hat \Phi_\perp \hat \Phi_\perp^\top)^{-1} \hat \Phi \hat \Phi^\top \paren{\Sigma^t(K,\sigma_u,x_1) \otimes I_{\dx}} \hat \Phi_\perp}^2.
    \end{aligned}
    \end{equation}
    The first term may be simplified as follows:
    \begin{equation}
        \label{eq: orthogonal component bias bound}
    \begin{aligned}
         &\norm{(\hat \Phi \hat \Phi^\top \paren{\Sigma^t(K,\sigma_u,x_1) \otimes I_{\dx}} \hat \Phi \hat \Phi^\top  + \hat \Phi_\perp \hat \Phi_\perp^\top)^{-1} \hat \Phi_\perp}^2 \\
         &= \norm{\paren{\bmat{\hat \Phi & \hat \Phi_\perp} \bmat{\hat \Phi^\top \paren{\Sigma^t(K,\sigma_u,x_1) \otimes I_{\dx}} \hat \Phi \\ &  I} \bmat{ \hat \Phi & \hat \Phi_\perp}^\top }^{-1} \hat \Phi_\perp}^2 \\
         &= \norm{\bmat{\hat \Phi & \hat \Phi_\perp} \bmat{\hat \Phi^\top \paren{\Sigma^t(K,\sigma_u,x_1) \otimes I_{\dx}} \hat \Phi \\ &  I}^{-1} \bmat{ \hat \Phi & \hat \Phi_\perp}^\top  \hat \Phi_\perp}^2 \\
         &= 1.
    \end{aligned}
    \end{equation}
    For the second term, we may similarly show that 
    \begin{align*}
        &\norm{(\hat \Phi \hat \Phi^\top \paren{\Sigma^t(K,\sigma_u,x_1) \otimes I_{\dx}} \hat \Phi \hat \Phi^\top  + \hat \Phi_\perp \hat \Phi_\perp^\top)^{-1} \hat \Phi \hat \Phi^\top \paren{\Sigma^t(K,\sigma_u,x_1) \otimes I} \hat \Phi_\perp}^2 \\
        &= \norm{(\hat \Phi \hat \Phi^\top \paren{\Sigma^t(K,\sigma_u,x_1) \otimes I_{\dx} } \hat \Phi \hat \Phi^\top)^{\dagger} \hat \Phi \hat \Phi^\top \paren{\Sigma^t(K,\sigma_u,x_1) \otimes I_{\dx}} \hat \Phi_\perp}^2 \\
        &\leq \frac{\norm{\Sigma^t(K,\sigma_u,x_1)}}{\lambda_{\min}(\hat \Phi^\top (\bar \Sigma^t(K,\sigma_u,x_1) \otimes I_{\dx}) \hat \Phi)},
    \end{align*}
    where the final inequality follows from submultiplicativity, the fact that $\hat \Phi$ has orthonormal columns and the fact that $\Sigma^t(K,\sigma_u,x_1) \preceq \bar \Sigma^t(K,\sigma_u,x_1)$.

The lemma follows by substituting the above inequality along with \eqref{eq: orthogonal component bias bound} into \eqref{eq: decompose orthogonal and nonorthogonal}. The result is then substituted into \eqref{eq: bias bound main ineq}, where we recall the definition \Cref{def: representation error}. 
\end{proof}

\subsection{Proof of \Cref{thm: ls estimation error}}

\begin{proof}
    Combining the results of \Cref{lem: bias variance decomposition} with \Cref{lem: variance bound}, we find that under the event  $\calE_{\ls}$ of \Cref{lem: variance bound},
    \begin{align*}
        &\norm{\hat \Phi \hat \theta  - \hat \Phi \bar \theta }^2 \lesssim \frac{ \dtheta \sigma^2  }{t \lambda_{\min} \paren{\hat \Phi^\top \paren{\bar \Sigma^t(K,\sigma_u,x_1) \otimes I_{\dx}} \hat \Phi}} \log\paren{ \frac{1}{ \delta} } \\
        &+ \paren{1+ \frac{\sigma^4 \norm{\Sigma^t(K,\sigma_u, x_1)}^2 \norm{P_K}^3 \Psi_{B^\star}^2  }{ t\lambda_{\min} \paren{\hat \Phi^\top \paren{\bar \Sigma^t(K,\sigma_u,x_1) \otimes I_{\dx}} \hat \Phi}^2 } \paren{\dx+\du + \log \frac{1}{\delta}}}\norm{\Phi^\star \theta^\star - \hat \Phi \bar \theta }^2.
    \end{align*}
    Using \Cref{lem: bias bound}, the above may be bounded as 
      \begin{align*}
        &\norm{\hat \Phi \hat \theta  - \hat \Phi \bar \theta }^2 \lesssim \frac{ \dtheta \sigma^2  }{t \lambda_{\min} \paren{\hat \Phi^\top \paren{\bar \Sigma^t(K,\sigma_u,x_1) \otimes I_{\dx}} \hat \Phi}} \log\paren{ \frac{1}{ \delta} } \\
        &+ \paren{1+ \frac{\sigma^4 \norm{\Sigma^t(K,\sigma_u, x_1)}^2 \norm{P_K}^3 \Psi_{B^\star}^2  }{ t\lambda_{\min} \paren{\hat \Phi^\top \paren{\bar \Sigma^t(K,\sigma_u,x_1) \otimes I_{\dx}} \hat \Phi}^2 } \paren{\dx+\du + \log \frac{1}{\delta}}} \\
        &\times \frac{\norm{\Sigma^t(K,\sigma_u,x_1)}}{\lambda_{\min}(\hat \Phi^\top \paren{\bar \Sigma^t(K,\sigma_u,x_1) \otimes I_{\dx}} \hat \Phi)} d(\hat \Phi, \Phi^\star)^2 \norm{\theta^\star}^2.
    \end{align*}
  
    % Using the lower bound on $\Sigma_k$ from \Cref{lem: covariance facts}, we have $\lambda_{\min}(\Sigma_k) \geq \frac{\sigma_k^2}{8.2 \norm{P^\star}} \geq \frac{\sigma_{\mathsf{in}, 2}^2 \varepsilon}{8.2 \norm{P^\star}} \geq \frac{\sigma_{\mathsf{in}, 2}^2 d(\hat \Phi, \Phi^\star)}{8.2 \norm{P^\star}}$. This leads to the bound
    %  \begin{align*}
    %     &\norm{\hat \Phi \hat \theta_k - \Phi^\star \theta^\star}^2 \lesssim \frac{\dtheta \sigma^2 \norm{P^\star}}{\sigma_{\mathsf{in},1}^2 \sqrt{\tau_{k-1}}} \log\paren{ \frac{1 }{ \delta} } \\
    %     &+ \paren{1 + \frac{\sigma^4 \norm{\Sigma_k}^2 \norm{P^\star}^5 \Psi_{B^\star}^2}{\sigma_{\mathsf{in},1}^4} \paren{\dx(\dx+\du) + \log \frac{1}{\delta}}} \norm{P^\star} \norm{\Sigma_k} \frac{d(\hat \Phi, \Phi^\star)}{\sigma_{\mathsf{in},2}^2} \norm{\theta^\star}^2. 
    % \end{align*}
    % Using the fact that $\sigma_\mathsf{in}^2 \baroversim $, the above bound simplifies to \Bruce{Can't actually do this yet... leave the result in terms of $\sigma_{\mathsf{in}}^2$}
    % \begin{align*}
    %      &\norm{\hat \Phi \hat \theta_k - \Phi^\star \theta^\star}^2 \lesssim \frac{\dtheta \sigma^2 \norm{P^\star}}{\sigma_{\mathsf{in}}^2 \sqrt{\tau_{k-1}}} \log\paren{ \frac{\norm{P^\star}^4 \Psi_{B^\star} \calJ_0 \log\frac{1}{\delta}}{ \delta} } \\
    %     &+ \paren{1 + \frac{\sigma^4 \norm{\Sigma_k}^2 \norm{P^\star}^5 \Psi_{B^\star}^2}{\sigma_{\mathsf{in}}^4} \paren{\dx(\dx+\du) + \log \frac{1}{\delta}}} \norm{P^\star} \norm{\Sigma_k} d(\hat \Phi, \Phi^\star) \norm{\theta^\star}^2.
    % \end{align*}
    To conclude, we use the upper bounds on $\norm{\Sigma^t(K,\sigma_u,x_1)}$ from \Cref{lem: covariance facts}:
    \begin{align*}
        \norm{\Sigma^t(K,\sigma_u,x_1)} &\leq (1+  \norm{P_K}\frac{\norm{ x_1^2}}{t-1})  \norm{\bar\Sigma^t(K,\sigma_u,x_1)}\leq 5 (1+  \norm{P_K}\frac{\norm{ x_1^2}}{t-1})   \norm{P_K}^2 \Psi_{B^\star}^2.
    \end{align*}
    % where the final inequality holds under the event $\calE_{\mathsf{bound}}$ from \Cref{lem: transition bounds}. 
    By the fact that $t \geq c \tau_{\ls}$, the above quantity is bounded as $\norm{\Sigma^t(K,\sigma_u,x_1)} \leq 10 \norm{P_K}^2 \Psi_{B^\star}^2$. 
    % Then using the fact that 
    % $$t \geq c \tau_{\ls} \geq \frac{\sigma^4 \norm{P_K} (1+\norm{K}^2) \Psi_{B^\star}^3(1+\norm{A_K}_{\calH_\infty})^2}{\lambda_{\min}\paren{\hat \Phi^\top \paren{\bar \Sigma^t(K,\sigma_u,x_1) \otimes I_{\dx}}\hat \Phi}^2}\paren{\dx(\dx+\du) + \log \frac{1}{\delta}},$$
    % we find
    Therefore we find that 
    \begin{align*}
        &\norm{\hat \Phi \hat \theta  - \Phi^\star \theta^\star}^2 \lesssim \frac{\dtheta \sigma^2  }{t \lambda_{\min} \paren{\hat \Phi^\top \paren{\bar \Sigma^t(K,\sigma_u,x_1) \otimes I_{\dx}} \hat \Phi}} \log\paren{ \frac{1}{ \delta} } \\
        &\quad+ \paren{1+ \frac{\sigma^4 \norm{P_K}^7 \Psi_{B^\star}^6 \paren{\dx+\du + \log \frac{1}{\delta}}}{t\lambda_{\min}(\hat \Phi^\top \paren{\bar \Sigma^t(K,\sigma_u,x_1) \otimes I_{\dx}} \hat \Phi)^2 } } \frac{\norm{P_K}^2 \Psi_{B^\star}^2 d(\hat \Phi, \Phi^\star)^2 \norm{\theta^\star}^2}{\lambda_{\min}(\hat \Phi^\top \paren{\bar \Sigma^t(K,\sigma_u,x_1) \otimes I_{\dx}} \hat \Phi)} .
    \end{align*}
    %wher the final inequality follows by observing that $\sigma_{\mathsf{in},1}^4 ...$ \Bruce{Or is it the other way?}
\end{proof}

% \subsection{Estimation Error Bounds From Adaptive Control}

% We now consider the error in the estimates $\bmat{\hat A_1 & \hat B_1}, \bmat{\hat A_2 & \hat B_2}, \bmat{\hat A_3 & \hat B_3}, \dots \bmat{\hat A_{k_\fin} & \hat B_{k_{\fin}}}$ generated online during deployment of the adaptive controller  in \Cref{alg: ce with exploration} for $T$ timesteps. We first present a result that holds when we have persistent, but decaying exploration. We then present a bound which holds only when the dynamics estimate $\hat \Phi$ satisfies a condition that eliminates the need for persistent exploration.
 
% \begin{lemma}
%     Suppose 
%     \begin{align*}
%         \epoch &\geq \\
%         \sigma_k^2 &= \min\curly{1, \frac{\sigma_{\mathsf{in}, 1}^2}{\sqrt{\epoch 2^k}} + \sigma_{\mathsf{in}, 2}^2 \varepsilon} \quad \forall k \geq 1
%     \end{align*}
%     Consider running \Cref{alg: ce with exploration} on \eqref{eq: dynamics} with inputs $K_0$, $\varepsilon$, $\epoch$, $\sigma_1^2, \sigma_2^2, \dots$ for $T$ timesteps. Let $\delta \in (0,1)$. There exists an event $\calE_{\ls}$ that holds with probability at least $1-\frac{\delta}{2}$ such that under the event $\calE_{\ls}$ 
%     \begin{align*}
%         \norm{\bmat{\hat A_k & \hat B_k} - \bmat{A^\star & B^\star}}_F^2 \lesssim \frac{\sigma^2 \dtheta}{\sqrt{\tau_k}} \log\frac{1}{\delta} + \quad \forall k=1, \dots k_{\fin}
%     \end{align*}
% \end{lemma}
% \begin{proof}
    
% \end{proof}
\section{High Probability Bounds on the Success Events}
\label{s: success event bounds}

We first state a more complete counterpart of \Cref{lem: CE closeness main body} that also
characterize the error in the learned controller gain and the Lyapunov equation solution $P_{\hat K}$ under this controller. 
\begin{lemma}[Theorem 3 of \cite{simchowitz2020naive}]
    \label{lem: CE closeness}
    Define $\varepsilon \triangleq \frac{1}{2916 \norm{P^\star(A^\star, B^\star)}^{10}}$. As long as 
    % \begin{align}
    % \label{eq: control safety condition}
    $
        \norm{\bmat{\hat A & \hat B}-\bmat{A^\star & B^\star}}_F^2 \leq \varepsilon, 
    $
    we have that $P_{\hat K} \preceq \frac{21}{20} P^\star$, $\norm{\hat K - K^\star}\leq \frac{1}{6 \norm{P^\star}^{3/2}}$, and 
    \begin{align*}
        \calJ(\hat K)  - \calJ(K^\star) \leq 142 \norm{P^\star}^8 \norm{\bmat{\hat A & \hat B}-\bmat{A^\star & B^\star}}_F^2.
    \end{align*}
\end{lemma}

\begin{lemma}
    \label{lem: noise bound}
    Let $\delta \in (0,1)$. With probability at least $1-\delta$
    \begin{align*}
        \max_{1\leq t \leq T} \norm{\bmat{w_t \\ g_t}} \leq 4\sigma \sqrt{ (\dx+\du) \log\frac{T}{\delta}}.
    \end{align*}
\end{lemma}
\begin{proof}
    The vectors $w_t$ and $g_t$ have $\sigma^2$-sub-Gaussian components. Therefore, a sub-Gaussian tail bound tells us that for any unit vector $v \in \mathbb{R}^{\dx+\du)}$
    \begin{align*}
        \P\brac{ v^\top \bmat{w_t \\ g_t} \geq \rho} \leq \exp\paren{-\frac{\rho^2}{2\sigma^2}}.
    \end{align*}
    By a covering argument, we therfore conclude 
    \begin{align*}
        \P\brac{ \norm{\bmat{w_t \\ g_t}} \geq \rho} \leq 5^{\dx+\du} \exp\paren{-\frac{\rho^2}{8\sigma^2}}.
    \end{align*}
    Union bounding over the $T$ timesteps yields
    \begin{align*}
        \P\brac{\max_{1\leq t \leq T} \norm{\bmat{w_t \\ g_t}} \geq \rho} \leq T 5^{\dx+\du} \exp\paren{-\frac{\rho^2}{8\sigma^2}}.
    \end{align*}
    Setting $\rho = 4\sigma \sqrt{ (\dx+\du) \log\frac{T}{\delta}}$, the above probability becomes 
    \begin{align*}
        T 5^{\dx+\du} \exp\paren{-2 (\dx+\du) \log \frac{T}{\delta}} &\leq T 5^{\dx+\du} \exp\paren{-2 (\dx+\du) -\log \frac{T}{\delta}} \\
        &\leq \delta 5^{\dx+\du} \exp(-2)^{\dx+\du} \leq \delta.
    \end{align*}
\end{proof}

\begin{lemma}
    \label{lem: state rollout bound}
    Consider rolling out the system $x_{s+1} = A^\star x_s + B^\star u_s +w_s$ from initial state $x_1$ for $t$ timesteps under the control action $u_s = K x_s + \sigma_u g_s$ where $K$ is stabilizing and $\sigma_u \leq 1$. Then the following bound holds:
    \begin{align*}
        \norm{x_t} &\leq \sqrt{\paren{1 - \frac{1}{\norm{P_K}}}^t \norm{P_K}} \norm{x_1} + 2\norm{P_{K}}^{3/2} \Psi_{B^\star} \max_{1\leq t \leq T } \norm{\bmat{w_t \\ g_t}}
    \end{align*}
\end{lemma}
\begin{proof}
    The states may be expressed as $x_{t} = A_K^{t-1} x_1 +  \sum_{k=0}^{t-1} A_{K}^{t-1-k} \bmat{I & \sigma_u B^\star} \bmat{w_k \\ g_k}$. Therefore, by \Cref{lem: powers of closed loop by lyap},
    \begin{align*}
        \norm{x_t} &\leq \norm{A_K^{t-1}} \norm{x_1} + \sum_{k=0}^{t-1} \norm{A_{K}^{t-1-k}} \norm{\bmat{I & \sigma_u B^\star}} \max_{1\leq t \leq T } \norm{\bmat{w_t \\ g_t}} \\
        &\leq \sqrt{\paren{1 - \frac{1}{\norm{P_K}}}^{t-1} \norm{P_K}} \norm{x_1} + 2\norm{P_{K}}^{3/2} \Psi_{B^\star} \max_{1\leq t \leq T } \norm{\bmat{w_t \\ g_t}}.
    \end{align*}
\end{proof}

\begin{lemma}
    \label{lem: state rollout bounds}
    Consider rolling out the system $x_{s+1} = A^\star x_s + B^\star u_s +w_s$ from initial state $x_1$ for $t$ timesteps under the control action $u_s = K x_s + \sigma_u g_s$ where $K$ is stabilizing and $\sigma_u \leq 1$. Suppose 
    \begin{itemize}
        \item $\norm{x_1} \leq 16 \norm{P_{K_0}}^{3/2} \Psi_{B^\star} \max_{1\leq t \leq T} \norm{\bmat{w_t \\ g_t}}$
        \item $\norm{P_K} \leq 2 \norm{P_{K_0}}$
        \item $t \geq \log_{\paren{1 - \frac{1}{\norm{P_K}}}}\paren{\frac{1}{4\norm{P_K}}}+1$.
    \end{itemize}  Then for $s=1, \dots, t$
    \begin{align*}
        \norm{x_s} \leq 40 \norm{P_{K_0}}^2  \Psi_{B^\star} \max_{1 \leq t \leq T} \norm{\bmat{w_t \\ g_t}}.
    \end{align*}
    Furthermore, 
    \begin{align*}
        \norm{x_t} \leq 16 \norm{P_{K_0}}^{3/2} \Psi_{B^\star} \max_{1\leq t \leq T} \norm{\bmat{w_t \\ g_t}}.
    \end{align*}
\end{lemma}
\begin{proof}
    By \Cref{lem: state rollout bound}, 
    \begin{align*}
        \norm{x_s} &\leq \norm{P_K}^{1/2} \norm{x_1} + 2 \norm{P_K}^{3/2} \Psi_{B^\star} \max_{1\leq t \leq T} \norm{\bmat{w_t \\ g_t}} \\
        &\leq 2 \norm{P_{K_0}}^{1/2} \norm{x_1} + 8 \norm{P_{K_0}}^{3/2} \Psi_{B^\star} \max_{1 \leq t \leq T} \norm{\bmat{w_t \\ g_t}}
    \end{align*}
    The bound in the first point then follows by substituting in the bound for $\norm{x_1}$. For the second point, note that by \Cref{lem: state rollout bound}, and the fact that $t \geq \log_{\paren{1 - \frac{1}{\norm{P_K}}}}\paren{\frac{1}{4 \norm{P_K}}} + 1$,
    \begin{align*}
        \norm{x_t} &\leq \frac{1}{2} \norm{x_1} + 2 \norm{P_K}^{3/2} \Psi_{B^\star} \max_{1\leq t \leq T} \norm{\bmat{w_t \\ g_t}} \\
        &\leq  \frac{1}{2} \norm{x_1} +  8 \norm{P_{K_0}}^{3/2} \Psi_{B^\star} \max_{1\leq t \leq T} \norm{\bmat{w_t \\ g_t}} \\ 
        & \leq 16 \norm{P_{K_0}}^{3/2} \Psi_{B^\star} \max_{1\leq t \leq T} \norm{\bmat{w_t \\ g_t}}.
    \end{align*}
\end{proof}

\begin{lemma}
    Running \Cref{alg: ce with exploration} with the arguments defined in \Cref{thm: regret bound naive exploration}, the event $\calE_{\mathsf{success},1} $ holds with probability at least $1-T^{-2}$.
\end{lemma}
\begin{proof}
    By \Cref{lem: noise bound}, we have that with probability at least $1-\frac{1}{2}T^{-2}$, 
    \begin{align*}
        \max_{1\leq t \leq T} \norm{\bmat{w_t \\ g_t}} \leq 4\sigma\sqrt{3 (\dx + \du) \log 2 T}.  
    \end{align*}
    Denote the event that the above inequality holds by $\calE_{\mathsf{w\,bound}}$. 
    We show by induction that the requisite conditions of $\calE_{\mathsf{success},1}$ hold for all epochs. 
    A key step in this argument is instantiating the result of \Cref{thm: ls estimation error} for every epoch. The rate of decay in this bound is the determined by minimum eigenvalue of the state and input covariance, 
    $\lambda_{\min}(\hat \Phi^\top \paren{\bar \Sigma^t(K,\sigma_u,x_1) \otimes I_{\dx}} \hat \Phi)$. 
    We control this quantity by invoking point 3 of \Cref{lem: covariance facts} along with the facts that 1) $\hat \Phi$ has orthonormal columns and 2)  $\sigma_u = \sigma_k \leq 1$ for all epochs due to the lower bound on $\tau_1$. In particular, we may show that
    \begin{align}
        \label{eq: covariance lower bound continual exploration}
        \lambda_{\min}(\hat \Phi^\top \paren{\bar \Sigma^t(K,\sigma_u,x_1) \otimes I_{\dx}} \hat \Phi) \geq \frac{\sigma_u^2}{2(2 + 2 \norm{K}^2)} \geq \frac{\sigma_u^2}{8 \norm{P_K}},
    \end{align}
    where the final inequality follows by noting that $2 + 2\norm{K}^2 \leq 2 + 2\norm{P_K} \leq 4 \norm{P_K}$ by the fact that $\norm{P_K} \geq 1$.

    \paragraph{Base Case:} For the first epoch, observe that our lower bound on $\tau_1$ ensures that $\tau_1 \geq c \tau_{\ls}(K_0, 0, \frac{1}{2}T^{-3})$ for a sufficiently large constant $c$. Therefore applying \Cref{thm: ls estimation error} along with \eqref{eq: covariance lower bound continual exploration} implies that there exists an event $\calE_{\ls,1}$ which holds with probability at least $\frac{1}{2}{T^{-3}}$ such that under $\calE_{\ls,1}$, 
    \begin{align*}
        &\norm{\bmat{\hat A_1 & \hat B_1} - \bmat{A^\star & B^\star}}_F^2 \lesssim \frac{\dtheta \sigma^2 \norm{P_{K_0}} }{\tau_1 \sigma_1^2} \log\paren{ T} \\
        &\quad+ \paren{1+ \frac{\sigma^4 \norm{P_{K_0}}^9 \Psi_{B^\star}^6 \paren{\dx+\du + \log\paren{T}}}{\tau_1 \sigma_1^4 } } \frac{\norm{P_{K_0}}^3 \Psi_{B^\star}^2 d(\hat \Phi, \Phi^\star)^2 \norm{\theta^\star}^2}{\sigma_1^2}.
    \end{align*}
    Using the fact that $\sigma_1^2 \geq \frac{\sqrt{\dtheta/\du}}{\sqrt{\tau_1}}$ to simplify the first appearance of $\sigma_1$, and $\sigma_1^2 \geq d(\hat \Phi, \Phi^\star)^{1/2}$ for the last two appearances, the above bound simplifies to 
    \begin{align*}
        &\norm{\bmat{\hat A_1 & \hat B_1} - \bmat{A^\star & B^\star}}_F^2 \lesssim \frac{\sqrt{\dtheta \du} \sigma^2 \norm{P_{K_0}} }{\sqrt{\tau_1}} \log\paren{ T} \\
        &\quad+ \sigma^4 \frac{\dtheta}{\du} \norm{P_{K_0}}^{12} \Psi_{B^\star}^8 \paren{\dx+\du}  d(\hat \Phi, \Phi^\star)^{1/2} \norm{\theta^\star}^2,
    \end{align*}
    where we have used the fact that $\dx + \du + \log T \leq (\dx +\du) \log T$ and that $\tau_1 \geq \log T$ to remove the $\log T$ term from the numerator. 
    By recalling the definition of $\beta_1$ from \Cref{asmp: upper bound on representation error exp}, the above bound can be simplified to 
    \begin{align*}
        \norm{\bmat{\hat A_1 & \hat B_1} - \bmat{A^\star & B^\star}}_F^2 \leq C_{\mathsf{est},1} \sigma^2 \frac{\sqrt{\dtheta \du} \norm{P_{K_0}}  }{\sqrt{\tau_1}} \log\paren{T} + \beta_1 d(\hat \Phi, \Phi^\star)^{1/2}.
    \end{align*}
    \sloppy Meanwhile, \Cref{lem: state rollout bounds} tells us that under the event $\calE_{\mathsf{w\,bound}},$ $\norm{x_t}^2 \leq x_b^2\log T$ for $t=1,\dots, \tau_1$, and that $\norm{x_{\tau_1}} \leq 16 \norm{P_{K_0}}^{3/2} \Psi_{B^\star} \max_{1\leq t \leq T} \norm{\bmat{w_t \\ g_t}}$.  Finally, we have $\norm{K_0} \leq K_b$.

    \paragraph{Induction step: } Suppose $$\norm{x_{\tau_k}} \leq 16 \norm{P_{K_0}}^{3/2} \Psi_{B^\star} \max_{1\leq t \leq T} \norm{\bmat{w_t \\ g_t}}$$ and  
    \begin{align*}
        \norm{\bmat{\hat A_k & \hat B_k} - \bmat{A^\star & B^\star}}_F^2 &\leq C_{\mathsf{est},1}  \sigma^2 \frac{\sqrt{\dtheta \du} \norm{P_{K_0}}  }{\sqrt{\tau_1 2^{k-1}}} \log\paren{T} + \beta_1 \sqrt{d(\hat \Phi, \Phi^\star)}.
    \end{align*}
    By the assumptions that $\tau_1 \geq c\paren{ \frac{ \sigma^2 \sqrt{\dtheta \du} \norm{P_{K_0}} }{2\varepsilon} \log T }^2 $ for a sufficiently large constant $c$ and that  $d(\hat \Phi, \Phi^\star) \leq \frac{\varepsilon^2}{4 \beta_1^2}$, we have that $\norm{\bmat{\hat A_k & \hat B_k} - \bmat{A^\star & B^\star}}_F^2 \leq \varepsilon$. Then the conditions of \Cref{lem: CE closeness} are satisfied, and consequently, $\norm{P_{\hat K_{k+1}}} \leq 1.05 \norm{P^\star}$. Therefore
    \begin{align*}
        \tau_{\ls}(\hat K_{k+1}, x_b^2\log T, \frac{1}{2}T^{-3}) &\leq  2\tau_{\ls}(K^\star, x_b^2\log T, \frac{1}{2} T^{-3}) \quad \mbox{ and } \quad \norm{P_{\hat K_{k+1}}}  \leq 1.05 \norm{P_{K_0}} \leq 2 \norm{P_{K_0}}.
    \end{align*}
    We may use our lower bound on $\tau_1$ to conclude that $\tau_1 \geq c \tau_{\ls}(K^\star, x_b^2 \log T, \frac{1}{2}T^{-3})$ for a sufficiently large universal constant $c$. This guarantees that the conditions of \Cref{thm: ls estimation error} are satsified. As a result, there exists an event $\calE_{\ls,k+1}$ that holds with probability at least $1-\frac{1}{2}T^{-3}$ under which
     \begin{align*}
       &\norm{\bmat{\hat A_{k+1} & \hat B_{k+1}} - \bmat{A^\star & B^\star}}_F^2 \lesssim \frac{\dtheta \sigma^2 \norm{P_{\hat K_{k+1}}} }{(\tau_{k+1}-\tau_k) \sigma_{k+1}^2} \log\paren{T} \\
        &\quad+ \paren{1+ \frac{\sigma^4 \norm{P_{\hat K_{k+1}}}^9 \Psi_{B^\star}^6 \paren{\dx+\du + \log\paren{T}}}{(\tau_{k+1}-\tau_k) \sigma_{k+1}^4 } }  \frac{\norm{P_{\hat K_{k+1}}}^3 \Psi_{B^\star}^2 d(\hat \Phi, \Phi^\star)^2 \norm{\theta^\star}^2}{\sigma_{k+1}^2},
    \end{align*}
    where we have again used \eqref{eq: covariance lower bound continual exploration} to simplify the bound. Observing that $\tau_{k+1} - \tau_k = \frac{1}{2}\tau_{k+1}$, and using the lower bounds $\sigma_{k+1}^2 \geq \frac{\sqrt{\dtheta/\du}}{\sqrt{\tau_1 2^{k}}} = \frac{\sqrt{\dtheta/\du}}{\sqrt{\tau_{k+1}}}$ and $\sigma_{k+1}^2 \geq \exploration d(\hat \Phi, \Phi^\star)^{1/2} \geq d(\hat \Phi, \Phi^\star)^{1/2}$, the above bound simplifies to 
    \begin{align*}
       &\norm{\bmat{\hat A_{k+1} & \hat B_{k+1}} - \bmat{A^\star & B^\star}}_F^2 \lesssim  \frac{ \sqrt{\dtheta \du}\sigma^2 \norm{P_{\hat K_{k+1}}} }{\sqrt{\tau_{k+1}}} \log\paren{ T} \\
        &\quad+  \sqrt{\dtheta/\du}\sigma^4 \norm{P_{\hat K_{k+1}}}^{12} \Psi_{B^\star}^8 \paren{\dx+\du} d(\hat \Phi, \Phi^\star)^{1/2} \norm{\theta^\star}^2.
    \end{align*}
    Using the fact that $\norm{P_{\hat K_{k+1}}} \leq 1.05 \norm{P^\star} \leq 1.05 \norm{P_{K_0}}$, the above bound may be simplified to 
    \begin{align*}
       &\norm{\bmat{\hat A_{k+1} & \hat B_{k+1}} - \bmat{A^\star & B^\star}}_F^2 \leq C_{\mathsf{est},1}
       \sigma^2 \frac{\sqrt{\dtheta \du}  \norm{P_{K_0}}  }{\sqrt{\tau_{k+1}}} \log\paren{T} + \beta_1   d(\hat \Phi, \Phi^\star)^{1/2}, 
    \end{align*}
    where the powers of $1.05$  are absorbed into the universal constants, and $\beta_1$ is as defined in \Cref{asmp: upper bound on representation error exp}.
    
    Furthermore, \Cref{lem: state rollout bounds} (which applies due to the inductive hypothesis, the lower bound on $\tau_1$, and the fact that $\norm{P_{\hat K_{k+1}}} \leq 2 \norm{P_{K_0}}$) informs us that under the event $\calE_{\mathsf{w\,bound}}$,  
    \begin{align*}
        \norm{x_t}^2 \leq x_b^2 \log T \quad \forall t =\tau_{k}+1, \dots, \tau_{k+1}, \mbox{ and } \norm{x_{\tau_{k+1}}}  \leq 16 \norm{P_{K_0}}^{3/2} \Psi_{B^\star} \max_{1\leq t \leq T} \norm{\bmat{w_t \\ g_t}}.
    \end{align*}
    Also observe that $\norm{\hat K_{k+1}}^2 \leq \norm{P_{\hat K_{k+1}}} \leq 2 \norm{P_{K_0}}$, which ensures that $\norm{\hat K_{k+1}} \leq 2 \norm{K^\star} \leq K_b$.
    In particular, we have verified that the algorithm does not abort on this epoch, and we have verified that the inductive hypothesis holds to start the next epoch.

    \paragraph{Union bounding over the events} The bounds on the norm of the state and the controller gain ensure that the algorithm does not abort. The estimation error bounds ensure that 
    \begin{align*}
        \norm{\bmat{\hat A_{k} & \hat B_{k}} - \bmat{A^\star & B^\star}}_F^2 &\leq C_{\mathsf{est},1} \sigma^2 \frac{\sqrt{\dtheta \du} \norm{P_{K_0}}  }{\sqrt{\tau_{1}}} \log\paren{T} + \beta_1  d(\hat \Phi, \Phi^\star)^{1/2} \quad \forall k \geq 1.
        %\norm{\bmat{\hat A_{k} & \hat B_{k}} - \bmat{A^\star & B^\star}}_F^2 &\leq C \sigma^2 \frac{\sqrt{\dtheta \du} \norm{P_{K^\star}}  }{\sqrt{\tau_{k}}} \log\paren{T} + \beta_1  \log(T) d(\hat \Phi, \Phi^\star) \quad \forall k \geq 2
    \end{align*}
Therefore, $\calE_{\mathsf{success},1} \subseteq\calE_{\mathsf{w\, bound}} \cap \calE_{\ls,1} \cap \dots \cap \calE_{\ls, k_{\fin}}$. Union bounding over the good events $\calE_{\mathsf{w\, bound}}, \calE_{\ls,1}, \dots, \calE_{\ls, k_{\fin}}$, we find that $\calE_{\mathsf{success},1}$ holds with probability at least $1-T^{-2}$.
\end{proof}

\begin{lemma}
    Running \Cref{alg: ce with exploration} with the arguments defined in \Cref{thm: regret bound no exploration}, the event $\calE_{\mathsf{success},2} $ holds with probability at least $1-T^{-2}$.
\end{lemma}
\begin{proof}
    By \Cref{lem: noise bound}, we have that with probability at least $1-\frac{1}{2}T^{-2}$, 
    \begin{align*}
        \max_{1\leq t \leq T} \norm{\bmat{w_t \\ g_t}} \leq 4\sigma\sqrt{3 (\dx + \du) \log 2 T}.  
    \end{align*}
    Denote the event that the above inequality holds by $\calE_{\mathsf{w\,bound}}$. 
    We show by induction that the requisite conditions of $\calE_{\mathsf{success},2}$ hold for all epochs. 
    A key step in this argument is instantiating the result of \Cref{thm: ls estimation error} for every epoch. The rate of decay in this bound is the determined by minimum eigenvalue of the state and input covariance, 
    $\lambda_{\min}(\hat \Phi^\top \paren{\bar \Sigma^t(K,\sigma_u,x_1) \otimes I_{\dx}} \hat \Phi)$. 
    We control this quantity by observing that for any unit vector  $v$
    \begin{equation}
    \begin{aligned}
        \label{eq: excitation reduction}
        v^\top \hat \Phi^\top \paren{\bar \Sigma^t(K,\sigma_u,x_1) \otimes I_{\dx}} \hat \Phi v &\geq \frac{1}{2}  v \hat \Phi^\top \paren{\bmat{I \\ K} \bmat{I \\ K}^\top \otimes I_{\dx}} \hat \Phi v.
        % \\
        % &=\frac{1}{2} \paren{\sum_{i=1}^{\dtheta} v_i \hat \Phi_i}^\top \paren{\bmat{I \\ K} \bmat{I \\ K}^\top \otimes I_{\dx}}\paren{\sum_{i=1}^{\dtheta} v_i \hat \Phi_i} \\
        % &= \frac{1}{2} \norm{\sum_{i=1}^{\dtheta} v_i \VEC^{-1}\paren{ \hat \Phi_i} \bmat{I \\ K}}_F^2 \\
        % &= \frac{1}{2} \norm{\sum_{i=1}^{\dtheta} v_i \bmat{\hat \Phi^A_i + \hat \Phi^B_i K}}_F^2.
    \end{aligned}
    \end{equation}
    We use this result along with \Cref{asmpt: persisent excitation} to show a lower bound $\lambda_{\min}(\hat \Phi^\top \paren{\bar \Sigma^t(K,\sigma_u,x_1) \otimes I_{\dx}} \hat \Phi) \geq \frac{\alpha^2}{8}$, and in turn show that the event $\calE_{\mathsf{success},2}$ holds with high probability. This result will proceed by induction.
     
     \sloppy \paragraph{Base Case:} For the first epoch, observe that our lower bound on $\tau_1$ guarantees that $\tau_1 \geq c \tau_{\ls}(K_0, 0, \frac{1}{2}T^{-3})$ for  a sufficiently large universal constant $c$. Therefore applying \Cref{thm: ls estimation error} along with \eqref{eq: excitation reduction} and \Cref{asmpt: persisent excitation} implies that there exists an event $\calE_{\ls,1}$ which holds with probability at least $\frac{1}{2}{T^{-3}}$ such that under $\calE_{\ls,1}$, 
    \begin{align*}
        &\norm{\bmat{\hat A_1 & \hat B_1} - \bmat{A^\star & B^\star}}_F^2 \lesssim \frac{\dtheta \sigma^2  }{\tau_1 \alpha^2} \log  T \\
        &\quad+ \paren{1+ \frac{\sigma^4 \norm{P_{K_0}}^7 \Psi_{B^\star}^6 \paren{\dx+\du + \log T}}{\tau_1 \alpha^4 } }  \frac{\norm{P_{K_0}}^2 \Psi_{B^\star}^2 d(\hat \Phi, \Phi^\star)^2 \norm{\theta^\star}^2}{\alpha^2}. 
    \end{align*}
    Using the fact that $\tau_1 = \tau_{\mathsf{warm\,up}}\log T$, we can cancel the $\log T$ in the numerator in the second term. We can further use the fact that $\tau_{\mathsf{warm\,up}} \geq c \frac{\sigma^4 \dtheta}{\epsilon \alpha^2}$ for some universal positive constant $c$ to simplify the second term to $\beta_2 d(\hat \Phi, \Phi^\star)^2$, where $\beta_2$ is as in \Cref{asmp: upper bound on representation error no exp}. Then the above bound may be simplified to 
    \begin{align*}
        \norm{\bmat{\hat A_1 & \hat B_1} - \bmat{A^\star & B^\star}}_F^2 \leq C_{\mathsf{est},2} \frac{\sigma^2 \dtheta \log T}{\tau_1 \alpha^2} + \beta_2 d(\hat \Phi, \Phi^\star)^2.
    \end{align*}
    where $C_{\mathsf{est},2}$ is a sufficiently large universal constant.
    \sloppy Meanwhile, \Cref{lem: state rollout bounds} tells us that under the event $\calE_{\mathsf{w\,bound}},$ $\norm{x_t}^2 \leq x_b^2 \log T$ for $t=1,\dots, \tau_1$, and that $\norm{x_{\tau_1}} \leq 16 \norm{P_{K_0}}^{3/2} \Psi_{B^\star} \max_{1\leq t \leq T} \norm{\bmat{w_t \\ g_t}}$.  Finally, we have $\norm{K_0} \leq K_b$.

    \paragraph{Induction step: } Suppose $$\norm{x_{\tau_k}} \leq 16 \norm{P_{K_0}}^{3/2} \Psi_{B^\star} \max_{1\leq t \leq T} \norm{\bmat{w_t \\ g_t}}$$ and  
    \begin{align*}
        \norm{\bmat{\hat A_k & \hat B_k} - \bmat{A^\star & B^\star}}_F^2 &\leq C_{\mathsf{est},2} \frac{\sigma^2 \dtheta \log T}{\tau_1 \alpha^2} + \beta_2 d(\hat \Phi, \Phi^\star)^2.
    \end{align*}
    By the lower bound on $\tau_1$, we have that $\tau_1 \geq c \frac{ \sigma^2 \dtheta \log T}{2\varepsilon \alpha^2} $ for a sufficiently large universal constant $c$. Combining this with the condition in assumption \Cref{asmp: upper bound on representation error no exp}, which says $d(\hat \Phi, \Phi^\star) \leq \sqrt{\frac{\varepsilon}{2 \beta_2 }}$, we find $\norm{\bmat{\hat A_k & \hat B_k} - \bmat{A^\star & B^\star}}_F^2 \leq \varepsilon$. Then the conditions of \Cref{lem: CE closeness} are satisfied, and consequently, $\norm{P_{\hat K_{k+1}}} \leq 1.05 \norm{P^\star}$. Therefore
    \begin{align*}
        \tau_{\ls}(\hat K_{k+1}, x_b^2 \log T, \frac{1}{2}T^{-3}) &\leq 2 \tau_{\ls}(K^\star, x_b^2 \log T, \frac{1}{2} T^{-3}) \quad \mbox{ and } \quad \norm{P_{\hat K_{k+1}}}  \leq 1.05 \norm{P_{K_0}} \leq 2 \norm{P_{K_0}}.
    \end{align*}
    We may guarantee from our lower bound on $\tau_1$ that $\tau_1 \geq c \tau_{\ls}(K^\star, x_b^2\log T, \frac{1}{2}T^{-3})$ for a sufficiently large universal constant $c$, which ensures that the conditions of \Cref{thm: ls estimation error} are satisfied. As a result, there exists an event $\calE_{\ls,k+1}$ that holds with probability at least $1-\frac{1}{2}T^{-3}$ under which
    \begin{align*}
       &\norm{\bmat{\hat A_{k+1} & \hat B_{k+1}} - \bmat{A^\star & B^\star}}_F^2 \lesssim \frac{\dtheta \sigma^2  }{(\tau_{k+1}-\tau_k) \alpha^2} \log  T \\
        &\quad+ \paren{1+ \frac{\sigma^4 \norm{P_{\hat K_{k+1}}}^7 \Psi_{B^\star}^6 \paren{\dx+\du + \log  T}}{(\tau_{k+1}-\tau_k) \alpha^4 } } \frac{\norm{P_{\hat K_{k+1}}}^2 \Psi_{B^\star}^2 d(\hat \Phi, \Phi^\star)^2 \norm{\theta^\star}^2}{\alpha^2},
    \end{align*}
    where we lower bounded $\lambda_{\min}(\hat \Phi^\top \paren{\bar \Sigma^{\tau_{k+1}-\tau_k}(\hat K_{k+1}, ) \otimes I_{\dx}} \hat \Phi)$ using \eqref{eq: excitation reduction} by noting that for any matrix $K \in \R^{\du \times \dx}$ and unit vector $v \in \R^{d_{\theta}}$, 
    \begin{align*}
     v \hat \Phi^\top \paren{\bmat{I \\ K} \bmat{I \\ K}^\top \otimes I_{\dx}} \hat \Phi v
        &= \paren{\sum_{i=1}^{\dtheta} v_i \hat \Phi_i}^\top \paren{\bmat{I \\ K} \bmat{I \\ K}^\top \otimes I_{\dx}}\paren{\sum_{i=1}^{\dtheta} v_i \hat \Phi_i} \\
        &= \norm{\sum_{i=1}^{\dtheta} v_i \VEC^{-1}\paren{ \hat \Phi_i} \bmat{I \\ K}}_F^2 \\
        &= \norm{\sum_{i=1}^{\dtheta} v_i \bmat{\hat \Phi^A_i + \hat \Phi^B_i K}}_F^2.
    \end{align*}
    Then the reverse triangle inequality  implies that:
    \begin{align*}
        &v^\top \hat \Phi^\top \paren{\bar \Sigma^{\tau_{k+1}-\tau_k}(\hat K_{k+1}, ) \otimes I_{\dx}} \hat \Phi v \geq  \norm{\sum_{i=1}^{\dtheta} v_i \bmat{\hat \Phi^A_i + \hat \Phi^B_i \hat K_{k+1} }}_F^2\\
        &=\norm{\sum_{i=1}^{\dtheta} v_i \paren{\hat \Phi^A_i + \hat \Phi^B_i K^\star } +\sum_{i=1}^{\dtheta} v_i \hat \Phi^B_i (\hat K_{k+1} - K^\star) }_F \geq  \norm{\sum_{i=1}^{\dtheta} v_i \paren{\hat \Phi^A_i + \hat \Phi^B_i K^\star }} - \norm{\hat K_{k+1} - K^\star} \\
        &\geq \sqrt{v \hat \Phi^\top \paren{\bmat{I \\ K} \bmat{I \\ K}^\top \otimes I_{\dx}} \hat \Phi v} - \frac{1}{6\norm{P^\star}^{3/2}} \quad \mbox{(By \Cref{lem: CE closeness})}
         \\
        &\geq \alpha - \frac{\alpha}{2} \quad \mbox{(by \Cref{asmpt: persisent excitation})}.
    \end{align*}
     Observing that $\tau_{k+1} - \tau_k = \frac{1}{2}\tau_{k+1}$, and 
    again using the fact that $\hat K_{k+1} = K^\star(\hat A_k, \hat B_k)$ where $\norm{\bmat{\hat A_k & \hat B_k} - \bmat{A^\star & B^\star}}_F^2 \leq \varepsilon$, we have that $\norm{P_{\hat K_{k+1}}} \leq 1.05 \norm{P^\star}$. Therefore by absorbing powers of $1.05$ into the universal constants, and following the same simplifications taken in the base case, the estimation error bound may be simplified to 
    \begin{align*}
       &\norm{\bmat{\hat A_{k+1} & \hat B_{k+1}} - \bmat{A^\star & B^\star}}_F^2 \leq C_{\mathsf{est},2} \frac{\sigma^2 \dtheta \log T}{\tau_{k+1} \alpha^2} + \beta_2  d(\hat \Phi, \Phi^\star)^2.
    \end{align*}
    
    Furthermore, \Cref{lem: state rollout bounds} (which applies due to the inductive hypothesis, the lower bound on $\tau_1$, and the fact that $\norm{P_{\hat K_{k+1}}} \leq 2 \norm{P_{K_0}}$) informs us that under the event $\calE_{\mathsf{w\,bound}}$,  
    \begin{align*}
        \norm{x_t}^2 \leq x_b^2 \log T \quad \forall t =\tau_{k}+1, \dots, \tau_{k+1}, \mbox{ and } \norm{x_{\tau_{k+1}}}  \leq 16 \norm{P_{K_0}}^{3/2} \Psi_{B^\star} \max_{1\leq t \leq T} \norm{\bmat{w_t \\ g_t}}.
    \end{align*}
    Also observe that $\norm{\hat K_{k+1}}^2 \leq \norm{P_{\hat K_{k+1}}} \leq 2 \norm{P_{K_0}}$, which ensures that $\norm{\hat K_{k+1}} \leq 2 \norm{K^\star} \leq K_b$.
    In particular, we have verified that the algorithm does not abort on this epoch, and we have verified that the inductive hypothesis holds to start the next epoch.

    \paragraph{Union bounding over the events} The bounds on the norm of the state and the controller gain ensure that the algorithm does not abort. The estimation error bounds ensure that 
    \begin{align*}
        \norm{\bmat{\hat A_{k} & \hat B_{k}} - \bmat{A^\star & B^\star}}_F^2 &\leq C_{\mathsf{est},2} \frac{\sigma^2 \dtheta \log T}{\tau_{k+1} \alpha^2} + \beta_2  d(\hat \Phi, \Phi^\star)^2 \quad \forall k \geq 1.
        %\norm{\bmat{\hat A_{k} & \hat B_{k}} - \bmat{A^\star & B^\star}}_F^2 &\leq C \sigma^2 \frac{\sqrt{\dtheta \du} \norm{P_{K^\star}}  }{\sqrt{\tau_{k}}} \log\paren{T} + \beta_1  \log(T) d(\hat \Phi, \Phi^\star) \quad \forall k \geq 2
    \end{align*}
Therefore, $\calE_{\mathsf{success},1} \subseteq\calE_{\mathsf{w\, bound}} \cap \calE_{\ls,1} \cap \dots \cap \calE_{\ls, k_{\fin}}$.  Union bounding over the good events $\calE_{\mathsf{w\, bound}}, \calE_{\ls,1}, \dots, \calE_{\ls, k_{\fin}}$, we find that $\calE_{\mathsf{success},1}$ holds with probability at least $1-T^{-2}$.
\end{proof}

\section{Proof of Main result}
\label{s: r1 r2 r3 bounds}
\begin{lemma}
    \label{lem: expected cost bound}
    Consider rolling out the system $x_{s+1} = A^\star x_s + B^\star u_s +w_s$ from initial state $x_1$ for $t$ timesteps under the control action $u_s = K x_s + \sigma_u g_s$ where $K$ is stabilizing and $\sigma_u \leq 1$. We have that
    \begin{align*}
        \E \brac{\sum_{s=1}^t \norm{x_t}_Q^2 + \norm{u_t}_R^2 \vert x_1} \leq  t J(K) + 2t \sigma_u^2\du \Psi_{B^\star}^2 \norm{P_K})  +  \norm{x_1}^2 \norm{P_K}
    \end{align*}
\end{lemma}
\begin{proof}
    Begin by writing
    \begin{align*}
        \E \brac{\sum_{s=1}^t \norm{x_t}_Q^2 + \norm{u_t}_R^2 \vert x_1} &= \sum_{s=1}^t \E\brac{\norm{x_t}_{Q+K^\top R K}^2 + \norm{\sigma_u g_t}^2 \vert x_1} \\
        &= t \sigma_u^2 + \sum_{s=1}^t \trace\paren{(Q+K^\top R K) \E\brac{ x_t x_t^\top \vert x_1}}. 
    \end{align*}
    We have that
    \begin{align*}
        &\E\brac{x_t x_t^\top \vert x_1} = \E\brac{  (A^\star + B^\star K) x_{t-1} x_{t-1}^\top (A^\star + B^\star K) + \sigma_u^2 B^\star (B^\star)^\top + I \vert x_1} \\
        &= (A^\star + B^\star K)^{t-1} x_1 x_1^\top \paren{(A^\star + B^\star K)^{t-1}}^\top \\
        &\qquad+ \sum_{k=1}^{t-1} (A^\star + B^\star K)^k \paren{\sigma_u^2 B^\star (B^\star)^\top + I} \paren{(A^\star + B^\star K)^k}^\top \\
        &\preceq  (A^\star + B^\star K)^{t-1} x_1 x_1^\top \paren{(A^\star + B^\star K)^{t-1}}^\top + \dlyap(A^\star + B^\star K, \sigma_u^2 B^\star (B^\star)^\top + I).
    \end{align*}
    Substituting this above, we find,
    \begin{align*}
        \E \brac{\sum_{s=1}^t \norm{x_t}_Q^2 + \norm{u_t}_R^2 \vert x_1} &\leq t \sigma_u^2 + t \trace\paren{(Q+K^\top R K) \dlyap(A^\star + B^\star K, \sigma_u^2 B^\star (B^\star)^\top + I)} \\
        &+ \sum_{s=1}^t  \trace\paren{(Q+K^\top R K)(A^\star + B^\star K)^{s-1} x_1 x_1^\top \paren{(A^\star + B^\star K)^{s-1}}^\top} \\
        &\leq  t J(K) + t \sigma_u^2 (1 + \du \Psi_{B^\star}^2 \norm{P_K}) +  \norm{x_1}^2 \norm{P_K}.
    \end{align*}
\end{proof}

\begin{lemma}
    \label{lem: r3 regret bound}
    In the settings of both \Cref{thm: regret bound naive exploration} and \Cref{thm: regret bound no exploration}, 
    \begin{align*}
        R_3 = \E\brac{\sum_{t=1}^{\tau_1} \norm{x_t}_Q^2 + \norm{u_t}_R^2 } \leq 3  \tau_1  \max\curly{\dx, \du} \norm{P_{K_0}} \Psi_{B^\star}^2  .
    \end{align*}
\end{lemma}
\begin{proof}
    This result follows from \Cref{lem: expected cost bound} by noting that $x_1 = 0$, and therefore 
    \begin{align*}
         \E\brac{\sum_{t=1}^{\tau_1} \norm{x_t}_Q^2 + \norm{u_t}_R^2 } \leq \tau_1   J(K_0)   + 2 \tau_1  \du \norm{P_{K_0}} \sigma_1^2 \Psi_{B^\star}^2 \leq 3  \tau_1 \max\curly{\dx, \du} \norm{P_{K_0}} \Psi_{B^\star}^2  
    \end{align*}
\end{proof}

\begin{lemma}
    \label{lem: r2 regret bound}
    Let $\calE_{\mathsf{success}} \gets \calE_{\mathsf{success},1}$ or $\calE_{\mathsf{success}} \gets \calE_{\mathsf{success},2}$ in the settings of \Cref{thm: regret bound naive exploration} and \Cref{thm: regret bound no exploration} respectively. Then in either setting
    \begin{align*}
        R_2 &= \E\brac{\mathbf{1}\paren{\calE_{\mathsf{success}}^c} \sum_{t=\tau_1+1}^{T} \norm{x_t}_Q^2 + \norm{u_t}_R^2 } \\ &\leq T^{-1} \paren{\norm{Q} + 2 K_b^2} x_b^2 \log T + T^{-1}J(K_0) + 24\norm{P_{K_0}} \Psi_{B^\star}^2 (\dx+\du)\sigma^2 T^{-2} \log 3T\\
        &+ 2 T^{-2} \norm{P_{K_0}} \norm{\theta_\star}_F^2 K_b^2 x_b^2 \log T+ \sum_{k=1}^{k_{\fin}}  2 (\tau_k - \tau_{k-1}) \du \sigma_k^2.
    \end{align*}
    In the setting of \Cref{thm: regret bound naive exploration}, this reduces to 
    \begin{align*}
                R_2 &\leq T^{-1} \paren{\norm{Q} + 2 K_b^2} x_b^2 \log T + T^{-1}J(K_0) + 24\norm{P_{K_0}} \Psi_{B^\star}^2 (\dx+\du)\sigma^2 T^{-2} \log 3T\\
        &+ 2 T^{-2} \norm{P_{K_0}} \norm{\theta_\star}_F^2 K_b^2 x_b^2 \log T+ 6\sqrt{T} \sqrt{\dtheta \du} + 2\du \gamma \sqrt{d(\hat\Phi, \Phi^\star)} T.
    \end{align*}
    In the setting of \Cref{thm: regret bound no exploration}, this reduces to 
    \begin{align*}
                R_2 &\leq T^{-1} \paren{\norm{Q} + 2 K_b^2} x_b^2 \log T + T^{-1}J(K_0) + 24\norm{P_{K_0}} \Psi_{B^\star}^2 (\dx+\du)\sigma^2 T^{-2} \log 3T\\
        &+ 2 T^{-2} \norm{P_{K_0}} \norm{\theta_\star}_F^2 K_b^2 x_b^2 \log T.
    \end{align*}
\end{lemma}
\begin{proof}  
    Let $t_{\mathsf{abort}}$ be the first time instance where $\norm{x_t}^2 \geq x_b^2 \log T$ or $\norm{\hat K_k} \geq K_b$.
    \begin{align*}
        R_2 &= \E\brac{\mathbf{1}\paren{\calE_{\mathsf{success}}^c} \sum_{t=\tau_1+1}^{T} \norm{x_t}_Q^2 + \norm{u_t}_R^2 } \\
        & \leq \E\brac{\mathbf{1}\paren{\calE_{\mathsf{success}}^c} \sum_{t=1}^{T} \norm{x_t}_Q^2 + \norm{u_t}_R^2 } \\
        &\leq \E\brac{\mathbf{1}(t_{\mathsf{abort}} \leq T) \sum_{t=1}^{t_{\mathsf{abort}} -1} \norm{x_t}_Q^2 + \norm{u_t}_R^2 } + \E\brac{\sum_{t=t_{\mathsf{abort}}}^{T} \norm{x_t}_Q^2 + \norm{u_t}_R^2 }.
    \end{align*} 
    By the fact that $u_t = \hat K_k x_t + \sigma_k g_t$ for $g_t \sim \calN(0,I)$, we have that $\norm{u_t}^2 \leq 2 \norm{\hat K_k x_t}^2  + 2 \norm{\sigma_u g_t}^2$. By the definition of $t_{\mathsf{abort}}$ and the fact that $R=I$, for $t < t_{\mathsf{abort}}$
    \begin{align*}
        \norm{x_t}^2_Q + 2\norm{\hat K_{k} x_t}_R^2 \leq x_b^2\log T \paren{\norm{Q} + 2 K_b^2}.
    \end{align*}
    Therefore 
    \begin{align*}
        &\E\brac{\mathbf{1}(t_{\mathsf{abort}} \leq T) \sum_{t=1}^{t_{\mathsf{abort}} -1} \norm{x_t}_Q^2 + \norm{u_t}_R^2 } \\
        &\leq \E\brac{\mathbf{1}(t_{\mathsf{abort}} \leq T) \sum_{t=1}^{t_{\mathsf{abort}} -1} x_b^2 \paren{\norm{Q} + 2 K_b^2} \log T + 2 \sum_{k=1}^{k_{\fin}} \sum_{t=\tau_{k-1}+1}^{\tau_k} 2 \sigma_k^2 \norm{g_t}^2 } \\
        &\leq \E\brac{\mathbf{1}(t_{\mathsf{abort}} \leq T) \sum_{t=1}^{t_{\mathsf{abort}} -1} x_b^2\paren{\norm{Q} + 2 K_b^2} \log T} + \E\brac{\mathbf{1}(t_{\mathsf{abort}} \leq T)  \sum_{k=1}^{k_{\fin}} \sum_{t=\tau_{k-1}+1}^{\tau_k} 2 \sigma_k^2 \norm{g_t}^2 } \\
         &\leq \E\brac{\mathbf{1}(t_{\mathsf{abort}} \leq T) T x_b^2 \paren{\norm{Q} + 2 K_b^2} \log T} + \E\brac{ \sum_{k=1}^{k_{\fin}} \sum_{t=\tau_{k-1}+1}^{\tau_k} 2 \sigma_k^2 \norm{g_t}^2 } \\
         &= \mathbf{P}(t_{\mathsf{abort}} \leq T) T x_b^2 \paren{\norm{Q} + 2 K_b^2} \log T +  \sum_{k=1}^{k_{\fin}} \sum_{t=\tau_{k-1}+1}^{\tau_k} 2 \du \sigma_k^2   \\
         &\leq T^{-1} x_b^2 \paren{\norm{Q} + 2 K_b^2} \log T +  \sum_{k=1}^{k_{\fin}}  2 (\tau_k - \tau_{k-1}) \du \sigma_k^2.
    \end{align*}
    For the second term, we have from \Cref{lem: expected cost bound}
    \begin{align*}
    \E\brac{\sum_{t=t_{\mathsf{abort}}}^{T} \norm{x_t}_Q^2 + \norm{u_t}_R^2 } &= \E\brac{\mathbf{1}(t_{\mathsf{abort}} \leq  T)\sum_{t=t_{\mathsf{abort}}}^{T} \E\brac{\norm{x_t}_Q^2 + \norm{u_t}_R^2 \vert t_{\mathsf{abort}}, x_{t_{\mathsf{abort}}}}} \\
    &\leq \E\brac{\mathbf{1}(t_{\mathsf{abort}} \leq  T) \paren{T J(K_0) + \norm{P_{K_0}} \norm{x_{t_{\mathsf{abort}}}}^2}}. %\\
    % &\leq T^{-1} J(K_0) + \norm{P_{K_0}} \E\brac{\mathbf{1}(t_{\mathsf{abort}} \leq  T) \norm{x_{t_{\mathsf{abort}}}}^2}.
    \end{align*}
    We have that $x_{t_{\mathsf{abort}}} = (A^\star + B^\star K) x_{t_{\mathsf{abort}}-1} + \bmat{I & \sigma_u B^\star} \bmat{w_t \\ g_t} $ for some $K$ satisfying $\norm{K} \leq K_b$ and $\sigma_u \leq 1$. Then we may write 
    \begin{align*}
        x_{t_{\mathsf{abort}}} &\leq 2 \norm{A^\star +B^\star K}^2 \norm{x_{t_{\mathsf{abort}}-1}}^2 + 2 \Psi_{B^\star}^2 \max_{1\leq t \leq T} \norm{\bmat{w_t \\ g_t}}^2 \\
        &\leq 2 \norm{\theta_\star}_F^2 K_b^2 x_b^2 \log T + 2 \Psi_{B^\star}^2 \max_{1\leq t \leq T} \norm{\bmat{w_t \\ g_t}}^2.
    \end{align*}
    Then 
    \begin{align*}
        \E\brac{\sum_{t=t_{\mathsf{abort}}}^{T} \norm{x_t}_Q^2 + \norm{u_t}_R^2 } &\leq \mathbf{P}(t_{\mathsf{abort}}\leq T) (TJ(K_0)+2 \norm{P_{K_0}} \norm{\theta_\star}_F^2 K_b^2 x_b^2 \log T) \\ 
        &+ 2\norm{P_{K_0}} \Psi_{B^\star}^2 \E\brac{\mathbf{1}\paren{t_{\mathsf{abort}}\leq T}\max_{1\leq t \leq T} \norm{\bmat{w_t \\ g_t}}^2}.
    \end{align*}
    \sloppy We have that $\mathbf{P}(t_{\mathsf{abort}}\leq T) (TJ(K_0)+2 \norm{P_{K_0}} \norm{\theta_\star}_F^2 K_b^2 x_b^2 \log T) \leq T^{-1}J(K_0) + 2 T^{-2} \norm{P_{K_0}} \norm{\theta_\star}_F^2 K_b^2 x_b^2 \log T $. Furthermore, by \Cref{lem: expected noise times ind}, 
    \begin{align*}
        2\norm{P_{K_0}} \Psi_{B^\star}^2\E\brac{\mathbf{1}\paren{t_{\mathsf{abort}}\leq T}\max_{1\leq t \leq T} \norm{\bmat{w_t \\ g_t}}^2} \leq 8\norm{P_{K_0}} \Psi_{B^\star}^2 (\dx+\du)\sigma^2 T^{-2} \log 9 T^3.
    \end{align*}
    Combining these results proves the first claim. 

    To reduce the bound to the expression specific to \Cref{thm: regret bound naive exploration}, we substitute in the bound $\sigma_k^2 \leq \gamma \sqrt{d(\hat\Phi,\Phi^\star)} + \frac{\sqrt{\du/\dtheta}}{\tau_{k}}$. To reduce the bound to the expression specific to \Cref{thm: regret bound no exploration}, we substitute $\sigma_k^2 = 0$. 
\end{proof}

\begin{lemma}[Adaptation of Lemma 35 of \cite{cassel2020logarithmic}]
    \label{lem: expected noise times ind}
    Suppose $\mathbf P (E) \leq \delta$. Then 
    \begin{align*}
        \E\brac{\mathbf{1}\paren{E}\max_{1\leq t \leq T} \norm{\bmat{w_t \\ g_t}}^2} \leq 4(\dx+\du) \sigma^2\delta \log\frac{9 T}{\delta}
    \end{align*}
\end{lemma}
\begin{proof}
    Note that 
    \begin{align*}
        \mathbf{P}\brac{\mathbf{1}(E) \max_{1\leq t \leq T} \norm{\bmat{w_t \\ g_t}}^2 > x} \leq \min\curly{\mathbf{P}(E), \mathbf{P}\paren{\max_{1\leq t \leq T} \norm{\bmat{w_t \\ g_t}}^2 > x}}.
    \end{align*}
    Therefore, the tail sum formula provides
     \begin{align*}
        \E\brac{\mathbf{1}\paren{E}\max_{1\leq t \leq T} \norm{\bmat{w_t \\ g_t}}^2} &= \int_0^\infty \mathbf{P}\brac{\mathbf{1}\paren{E}\max_{1\leq t \leq T} \norm{\bmat{w_t \\ g_t}}^2 > x} dx \\
        &\leq \int_0^{4(\dx+\du) \sigma^2 \log\frac{T}{\delta}} \mathbf{P}(E)dx  +\int_{4(\dx+\du) \sigma^2 \log\frac{T}{\delta}}^\infty \mathbf{P} \brac{\max_{1\leq t \leq T} \norm{\bmat{w_t \\ g_t}}^2 > x} dx \\
        &\leq  4(\dx+\du) \sigma^2\delta \log\frac{T}{\delta}  +\int_{4(\dx+\du) \sigma^2 \log\frac{T}{\delta}}^\infty 2T \exp\paren{-\frac{x}{4 (\dx+\du)\sigma^2}} dx\\
        &=4(\dx+\du) \sigma^2\delta \log\frac{T}{\delta}  +8 (\dx+\du) \sigma^2 \delta \\
        &=4(\dx+\du) \sigma^2\delta (\log\frac{T}{\delta}+\log 9) \\
        &= 4(\dx+\du) \sigma^2\delta \log\frac{9 T}{\delta},
    \end{align*}
    where the probability bound in the final inequality follows from the fact that $\norm{\bmat{w_t \\g_t}}^2$ is sub-exponential with parameter $4 (\dx+\du)\sigma^2$ (see Section 2.8 of \cite{vershynin2020high}). 
\end{proof}

\begin{lemma}
    \label{lem: r1 bound}
    Let $\calE_{\mathsf{success}} \gets \calE_{\mathsf{success},1}$ or $\calE_{\mathsf{success}} \gets \calE_{\mathsf{success},2}$ in the settings of \Cref{thm: regret bound naive exploration} and \Cref{thm: regret bound no exploration} respectively. Then in either setting, 
    \begin{align*}
        R_1 &= \E\brac{\mathbf{1}\paren{\calE_{\mathsf{success}}} \sum_{t=\tau_1+1}^{T} \norm{x_t}_Q^2 + \norm{u_t}_R^2 } \\
        &\leq \sum_{k=2}^{k_\fin} \bigg(\E\brac{\mathbf{1}(E_{k-1} )  142(\tau_{k}- \tau_{k-1}) \norm{P^\star}^8 \norm{\bmat{\hat A_{k-1} & \hat B_{k-1}} - \bmat{A^\star & B^\star}}_F^2 } \\ &+ (\tau_{k}- \tau_{k-1}) \calJ(K^\star)  + 4(\tau_{k}- \tau_{k-1}) \du\norm{P_{K^\star}}\sigma_k^2 \Psi_{B^\star}^2 + 2 x_b^2 \log T  \norm{P_{K^\star}}
        \bigg),
    \end{align*}
    where 
    \begin{align}
        \label{eq: epoch est error exp}
        E_{k} &= \curly{\norm{\bmat{\hat A_k & \hat B_k} - \bmat{A^\star & B^\star}}_F^2 \leq C_{\mathsf{est},1} \frac{\sigma^2 \sqrt{\dtheta \du} \norm{P_{K_0}} \log T}{\sqrt{\tau_{k}}} + \beta_1 \sqrt{d(\hat \Phi, \Phi^\star)}}
    \end{align}
    for $k=1, \dots, k_{\fin}-1$ in the setting of \Cref{thm: regret bound naive exploration} and 
    \begin{align}
        \label{eq: epoch est error no exp}
        E_{k} &= \curly{\norm{\bmat{\hat A_k & \hat B_k} - \bmat{A^\star & B^\star}}_F^2 \leq C_{\mathsf{est},2} \frac{\sigma^2  \dtheta   \log T}{ \tau_{k} \alpha^2 }+ \beta_2 d(\hat \Phi, \Phi^\star)^2 }
    \end{align}
    for $k=1, \dots, k_{\fin}-1$ in the setting of \Cref{thm: regret bound no exploration}.
\end{lemma}
\begin{proof}
    We have that 
    \begin{align*}
        \E\brac{\mathbf{1}\paren{\calE_{\mathsf{success}}} \sum_{t=\tau_1}^{T}\norm{x_t}_Q^2 + \norm{u_t}_R^2 } = \sum_{k=2}^{k_{\fin}} \mathbf{E}\brac{\mathbf{1}\paren{\calE_{\mathsf{success}}} \sum_{t=\tau_{k-1}+1}^{\tau_k}\norm{x_t}_Q^2 + \norm{u_t}_R^2 }.
    \end{align*}
    Let $S_k = \curly{\norm{x_{\tau_{k}+1}}^2 \leq x_b^2 \log T}$. We have that for any $k=1, \dots, k_{\fin}$, $\calE_{\mathsf{success}} \subseteq E_k \cap S_k$, and therefore $\mathbf{1}\paren{\calE_{\mathsf{success}}} \leq \mathbf{1}\paren{E_k \cap S_k}$. By the tower property, we may write that for any $k=2,\dots, k_{\fin}$,
    \begin{align*}
        &\mathbf{E}\brac{\mathbf{1}\paren{\calE_{\mathsf{success}}} \sum_{t=\tau_{k-1}+1}^{\tau_k}\norm{x_t}_Q^2 + \norm{u_t}_R^2 } \\&\leq \mathbf{E}\brac{ \mathbf{1}\paren{E_{k-1} \cap S_{k-1}} \E\brac{ \sum_{t=\tau_{k-1}+1}^{\tau_k}\norm{x_t}_Q^2 + \norm{u_t}_R^2  \vert x_{\tau_{k-1}+1}, \hat K_{k}}}.
    \end{align*}
    As $\hat K_k$ is stabilizing under $\calE_k$, we may invoke \Cref{lem: expected cost bound} to bound
    \begin{align*}
        &\mathbf{1}(E_{k-1} \cap S_{k-1}) \E\brac{ \sum_{t=\tau_{k-1}+1}^{\tau_k}\norm{x_t}_Q^2 + \norm{u_t}_R^2  \vert x_{\tau_{k-1}+1}, \hat K_{k}}  \\
        &\leq\mathbf{1}(E_{k-1} \cap S_{k-1}) \paren{(\tau_{k}- \tau_{k-1}) \calJ(\hat K_k) + 2(\tau_{k}- \tau_{k-1}) \du\norm{P_{\hat K_{k}}} \sigma_k^2 \Psi_{B^\star}^2 + \norm{x_{\tau_{k-1}+1}}^2
        \norm{P_{\hat K_{k}}}} \\
        &\overset{(i)}{\leq}\mathbf{1}(E_{k-1}  ) \paren{(\tau_{k}- \tau_{k-1}) \calJ(\hat K_k) + 4(\tau_{k}- \tau_{k-1}) \du\norm{P^\star} \sigma_k^2 \Psi_{B^\star}^2 + 2 x_b^2 \log T 
        \norm{P_{K^\star}}} \\
        &\overset{(ii)}{\leq}\mathbf{1}(E_{k-1} ) \bigg(142(\tau_{k}- \tau_{k-1}) \norm{P^\star}^8 \norm{\bmat{\hat A_{k-1} & \hat B_{k-1}} - \bmat{A^\star & B^\star}}_F^2\bigg) + (\tau_{k}- \tau_{k-1}) \calJ(K^\star) \\
        &\quad+ 4(\tau_{k}- \tau_{k-1}) \du\norm{P_{K^\star}} \sigma_k^2 \Psi_{B^\star}^2 + 2 x_b^2 \log T  \norm{P^\star}
    \end{align*}
    where inequality $(i)$ follows by using the bound from $S_{k-1}$ along with the fact that under the event $E_{k-1}$, the burn in time and assumptions on $d(\hat \Phi, \Phi^\star)$ ensure that $\norm{P_{\hat K_{k}}} \leq 2\norm{P^\star}$. Inequality $(ii)$ follows from the fact that under the burn-in time assumptions along with the event $E_{k-1}$, the condition of \Cref{lem: CE closeness} is satisfies, and $\calJ(\hat K) - \calJ(K^\star) l\leq 142 \norm{P^\star}^8 \norm{\bmat{\hat A_{k-1} & \hat B_{k-1}} - \bmat{A^\star & B^\star}}_F^2$.
    
    Taking the expectation and summing over $k=2,\dots, k_{\fin}$ proves the claim.
\end{proof}

\begin{lemma}
    \label{lem: r1 regret bound exp}
    In the setting of \Cref{thm: regret bound naive exploration}, 
    \begin{align*}
        R_1 - T \calJ(K^\star)
        &\leq  832  \sigma^2\sqrt{\dtheta \du}   \norm{P^\star}^8  C_{\mathsf{est},1} \norm{P_{K_0}}\Psi_{B^\star}^2 \sqrt{T} \log T  \\ &+ \paren{4 \du \norm{P^\star} \Psi_{B^\star}^2 \gamma\sqrt{d(\hat\Phi,\Phi^\star)} + 142 \norm{P^\star}^8 \beta_1 \sqrt{d(\hat \Phi, \Phi^\star)} } T      + 2x_b^2  \norm{P_{K^\star}} \log^2 T
         ,
    \end{align*} 
\end{lemma}
\begin{proof}
    By substituting the estimation error bound under the events $E_{k}$ in \eqref{eq: epoch est error exp}, and the bound on the exploration noise magnitude $\sigma_k^2 \leq  \frac{\sqrt{\dtheta}}{\sqrt{\du} \sqrt{\tau_1 2^{k-1}}} + \gamma\sqrt{d(\hat\Phi,\Phi^\star)}$, \Cref{lem: r1 bound} provides
     \begin{align*}
        R_1 &- T \calJ(K^\star)
        \leq  832  \sigma^2 \sqrt{\dtheta \du}   \norm{P^\star}^8  C_{\mathsf{est},1}   \norm{P_{K_0}}\Psi_{B^\star}^2 \sqrt{T} \log T  \\ &+ \paren{4 \du \norm{P^\star} \Psi_{B^\star}^2 \gamma\sqrt{d(\hat\Phi,\Phi^\star)} + 142 \norm{P^\star}^8 \beta_1 \sqrt{d(\hat \Phi, \Phi^\star)} } T      + 2 k_{\fin} x_b^2   \norm{P_{K^\star}} \log T
         ,
    \end{align*} 
    We have that $k_{\fin} = \log_2 T/\tau_1 \leq \log T$, providing the result in the lemma statement.
\end{proof}

\begin{lemma}
    \label{lem: r1 regret bound no exp}
    In the setting of  \Cref{thm: regret bound no exploration}, 
    \begin{align*}
        R_1 - T \calJ(K^\star)  
        % &\leq \sum_{k=2}^{k_\fin} \bigg(142(\tau_{k}- \tau_{k-1}) \norm{P^\star}^8 \paren{\frac{C_{\mathsf{est},2} \sigma^2 \dtheta \log T}{\tau_{k-1} \alpha^2} + \beta_2 d(\hat \Phi, \Phi^\star)^2} \bigg) \\ & +  2 k_{\fin} x_b  \norm{P_{K^\star}},
        &\leq  \paren{C_{\mathsf{est},2}142 \norm{P^\star}^8 \frac{ \sigma^2 \dtheta}{  \alpha^2} + x_b^2  \norm{P_{K^\star}} }\log^2 T+ 142 \norm{P^\star}^8 \beta_2 d(\hat \Phi, \Phi^\star)^2 T.
    \end{align*}
\end{lemma}
\begin{proof}
    By substituting the estimation error bound under the events $E_{k}$ from \eqref{eq: epoch est error no exp}, and the exploration noise magnitude $\sigma_k^2 = 0$, \Cref{lem: r1 bound} provides
     \begin{align*}
        R_1 - T \calJ(K^\star) \leq k_{\fin}\paren{C_{\mathsf{est},2}142 \norm{P^\star}^8 \frac{ \sigma^2 \dtheta}{  \alpha^2} + x_b^2  \norm{P_{K^\star}} }\log T+ 142 \norm{P^\star}^8 \beta_2 d(\hat \Phi, \Phi^\star)^2 T
         ,
    \end{align*}
    We have that $k_{\fin} = \log_2 T/\tau_1 \leq \log T$, providing the result in the lemma statement.
\end{proof}

\subsection{Proof of \Cref{thm: regret bound naive exploration}}

We consider the decomposition of the regret as 
\begin{align*}
    \E[\mathbf{R}_T] = R_1 + R_2 + R_3 - T \calJ(K^\star),
\end{align*}
with $R_1$, $R_2$, and $R_3$ as in \eqref{eq: regret decomposition}. Combining the bounds on $R_1-T\calJ(K^\star)$ from \Cref{lem: r1 regret bound exp}, on $R_2$ from \Cref{lem: r2 regret bound}, and on $R_3$ from \Cref{lem: r3 regret bound}, we find

\begin{align*}
    \E\brac{\mathbf{R}_T} &\lesssim  \tau_1 \max\curly{\dx,\du} \norm{P_{K_0}}\Psi_{B^\star}^2 \\
    &+ T^{-1}\paren{\paren{\norm{Q} +  K_b^2} x_b^2 \log T + J(K_0) + \norm{P_{K_0}} \Psi_{B^\star}^2 (\dx+\du)\sigma^2  \log T + \norm{P_{K_0}} \norm{\theta_\star}^2 K_b^2 x_b^2 \log T}\\
    &+ \sigma^2 \sqrt{\dtheta \du}   \norm{P^\star}^8  \norm{P_{K_0}}\Psi_{B^\star}^2 \sqrt{T} \log T  \\ &+ \paren{ \du \norm{P^\star} \Psi_{B^\star}^2 \gamma\sqrt{d(\hat\Phi,\Phi^\star)} +  \norm{P^\star}^8 \beta_1 \sqrt{d(\hat \Phi, \Phi^\star)}} T      +   x_b^2   \norm{P_{K^\star}} \log^2 T
\end{align*}
The first term may be modified by recalling that $\tau_1 = \tau_{\mathsf{warm\,up}} \log^2 T$. 
For the second term, note that by the fact that $T^{-1} \leq \tau_1^{-1}$, the following bound holds,
\begin{align*}
    &T^{-1}\paren{\paren{\norm{Q} +  K_b^2} x_b^2 \log T + J(K_0) + \norm{P_{K_0}} \Psi_{B^\star}^2 (\dx+\du)\sigma^2  \log T +  \norm{P_{K_0}} \norm{\theta_\star}^2 K_b^2 x_b^2 \log T} \\
    &\lesssim \paren{\norm{Q} + K_b^2 \paren{1 +\norm{P_{K_0}} \norm{\theta^\star}^2}}.
\end{align*}
Combining terms and substituting the definition of $\beta_1$ from \Cref{asmp: upper bound on representation error exp} leads to the theorem statement.

\subsection{Proof of \Cref{thm: regret bound no exploration}}
We consider the decomposition of the regret as 
\begin{align*}
    \E[\mathbf{R}_T] = R_1 + R_2 + R_3 - T \calJ(K^\star),
\end{align*}
with $R_1$, $R_2$, and $R_3$ as in \eqref{eq: regret decomposition}. Combining the bounds on $R_1-T\calJ(K^\star)$ from \Cref{lem: r1 regret bound no exp}, on $R_2$ from \Cref{lem: r2 regret bound}, and on $R_3$ from \Cref{lem: r3 regret bound}, we find
\begin{align*}
    \E\brac{\mathbf{R}_T} &\lesssim \tau_1 \max\curly{\dx,\du} \norm{P_{K_0}}\Psi_{B^\star}^2 \\
    &+ T^{-1} \paren{ J(K_0) + \paren{\norm{Q} + K_b^2 +\norm{P_{K_0}}\norm{\theta_\star}^2 K_b^2} x_b^2\log T+ \norm{P_{K_0}} \Psi_{B^\star} (\dx+\du) \sigma^2 \log T} \\
    &+\norm{P^\star}^8 \frac{ \sigma^2 \dtheta \log^2 T}{  \alpha^2} +  \norm{P^\star}^8 \beta_2 d(\hat \Phi, \Phi^\star)^2 T
         +   x_b^2  \norm{P_{K^\star}} \log^2 T.
\end{align*}
Substituting $\tau_1 = \tau_{\mathsf{warm\,up}} \log T \leq \tau_{\mathsf{warm\,up}} \log^2 T$ for $\tau_1$, and grouping terms, we find that

\begin{align*}
    \E\brac{\mathbf{R}_T} &\leq \paren{ \tau_{\mathsf{warm\,up}} \max\curly{\dx,\du} \norm{P_{K_0}}\Psi_{B^\star}^2 +  \norm{P^\star}^8 \frac{ \sigma^2 \dtheta }{  \alpha^2} +  x_b^2  \norm{P_{K^\star}}} \log^2 T \\
    &+ T^{-1} \paren{ J(K_0) + \paren{\norm{Q} + K_b^2 + \norm{P_{K_0}}\norm{\theta_\star}^2 K_b^2} x_b^2\log T+ \norm{P_{K_0}} \Psi_{B^\star} (\dx+\du) \sigma^2 \log T} \\
    &+ \norm{P^\star}^8 \beta_2 d(\hat \Phi, \Phi^\star)^2 T.
\end{align*}
Using the fact that $T^{-1} \leq \tau_1^{-1}$ we have the following bound: 
\begin{align*}
    &T^{-1} \paren{ \calJ(K_0) + \paren{\norm{Q} +K_b^2 +2 \norm{P_{K_0}}\norm{\theta_\star}^2 K_b^2} x_b^2\log T+ \norm{P_{K_0}} \Psi_{B^\star} (\dx+\du) \sigma^2 \log T} \\
    &\lesssim   \paren{\norm{Q} +  K_b^2 \paren{1 +\norm{P_{K_0}} \norm{\theta^\star}^2}}.
\end{align*}
Combining terms and substituting the definition of $\beta_2$ from \Cref{asmp: upper bound on representation error no exp} as well as the definition of $\varepsilon$ in \Cref{lem: CE closeness} provides the result in the theorem statement.

\section{Generalizations}
\label{s: generalizations}

\subsection{Logarithmic Regret with Known $A^\star$ or $B^\star$}
\label{s: log regret with known A or B}

\cite{cassel2020logarithmic} and \cite{jedra2022minimal} show that it is possible to attain logaritmic regret if either the $A^\star$ is known to the learner and $K^\star (K^\star)^\top \succ 0$, or the $B^\star$  is known to the learner. These settings cannot immediately be embedded into the formulation of \eqref{eq: dynamics}, because \eqref{eq: dynamics} requires that $\bmat{A^\star & B^\star}$ is a linear function of the unknown parameters. It can therefore only embed the settings where $A^\star$ and $B^\star$ are known up to scale. However, the model can easily be extended to
\begin{align*}
    x_{t+1} = \paren{\bmat{\bar A & \bar B} + \VEC^{-1} \paren{\Phi^\star \theta^\star}}\bmat{x_t \\ u_t} + w_t. 
\end{align*}
\begin{algorithm}
\label{alg: least squares affine model}
\begin{algorithmic}[1]
\State \textbf{Input:} Model structure estimate $\hat \Phi$, state data $x_{1:t+1}$, input data $u_{1:t}$
\State \textbf{Return: } $\hat \theta, \Lambda$, where %and $\Lambda_k$, where 
\begin{align*}
    \hat \theta = \Lambda^\dagger \paren{\sum_{s=1}^{t} \hat \Phi^\top \paren{ \bmat{x_s \\ u_s} \otimes I_{\dx}} \paren{x_{s+1} - \bmat{\bar A & \bar B}\bmat{x_t \\ u_t}}}, \mbox{ and } \Lambda = \sum_{s=1}^{t } \hat \Phi^\top \paren{ \bmat{x_s \\ u_s}\bmat{x_s \\ u_s} ^\top \otimes I_{\dx}}\hat \Phi 
\end{align*}
\end{algorithmic}
\end{algorithm}

If the learner has access to $\bar A$, $\bar B$, and an estimate $\hat \Phi$ for $\Phi^\star$, then the results of \Cref{thm: regret bound naive exploration} and \Cref{thm: regret bound no exploration} are unchanged by this generalization. In particular, the least squares algorithm may be modified to subtract out the known dynamics $\bar A$ and $\bar B$. This modified algorithm is shown in \Cref{alg: least squares affine model}. Propagating this change through the analysis does not change the results. 

Using the above generalization, we may recover logarithmic regret in the settings considered by  \cite{cassel2020logarithmic,jedra2022minimal}.
\begin{itemize}
    \item $B^\star$ known: We may set $\bmat{\bar A & \bar B} = \bmat{0 & B^\star}$ and $\theta^\star = \VEC A^\star$. We then set $\Phi^\star$ as the matrix with columns having a single element that is one, and the rest zero, such that $\VEC^{-1}\paren{\Phi^\star \theta^\star} = \bmat{A^\star & 0}$. As we assume the learner knows $B^\star$ exactly, and $\Phi^\star$ forms a complete basis for $A^\star$ we may simply choose $\hat\Phi = \Phi^\star$, and have $d(\hat\Phi, \Phi^\star) = 0$. To verify that we may achieve logarithmic regret in this setting using \Cref{thm: regret bound no exploration}, we must simply show that \Cref{asmpt: persisent excitation} is satisfied. To do so, note that for all $i=1,\dots, \dtheta$, $\Phi_i^B$ are zero, and $\hat \Phi_i^A$ are matrices each with a distinct unity element. Therefore, for any $K$, and any unit vector $v\in \R^{\dtheta}$,
    \begin{align*}
        \norm{\sum_{i=1} v_i \hat \Phi_i^A + \hat \Phi_i^B K}_F^2 = \norm{\sum_{i=1} v_i \hat \Phi_i^A}_F^2 = \sum_{i=1}^{\dtheta} v_i^2 = \norm{v}^2 = 1 \geq \frac{1}{3 \norm{P^\star}^{3/2}},
    \end{align*}
    where the last inequality follows by the fact that $\norm{P^\star} \geq 1$. This verifies that \Cref{asmpt: persisent excitation} holds in this setting.

    \item $A^\star$ known, $K^\star (K^\star)^\top \succeq \mu I$ and $K_0 K_0^\top \succeq \mu I$ for $\mu \geq \frac{1}{3\norm{P^\star}^{3/2}}$ \footnote{The assumption $K_0 K_0^\top \succeq \mu I$, and the lower bound on $\mu$ is not included in \cite{cassel2020logarithmic, jedra2022minimal}, but are needed here to verify \Cref{asmpt: persisent excitation}. The lower bound on $\alpha$ in \Cref{asmpt: persisent excitation} is used to find the probability bound on the success event $\calE_{\mathsf{success},2}$. It may be possible to remove this lower bound, but we leave that for future work.}: We may set $\bmat{\bar A & \bar B} = \bmat{A^\star & 0}$ and $\theta^\star = \VEC A^\star$. We then set $\Phi^\star$ as the matrix with columns having a single element that is one, and the rest zero, such that $\VEC^{-1}\paren{\Phi^\star \theta^\star} = \bmat{0 & B^\star}$.  As we assume the learner knows $A^\star$ exactly, and $\Phi^\star$ forms a complete basis for $A^\star$ we may simply choose $\hat\Phi = \Phi^\star$, and have $d(\hat \Phi, \Phi^\star) = 0$. To verify that we achieve logarithmic regret in this setting using \Cref{thm: regret bound no exploration}, we again show that \Cref{asmpt: persisent excitation} is satisfied. To do so, note that for all $i=1,\dots, \dtheta$, $\Phi_i^A$ are zero, and $\hat \Phi_i^B$ are matrices each with a distinct unity element. Therefore, for any $K$, and any unit vector $v\in \R^{\dtheta}$,
    \begin{align*}
        \norm{\sum_{i=1} v_i \hat \Phi_i^A + \hat \Phi_i^B K}_F^2 &= \norm{\sum_{i=1} v_i \hat \Phi_i^B K}_F^2 \\
        &= \trace\paren{(\sum_{i=1}^{\dtheta} v_i \hat \Phi_i^B)  K K^\top (\sum_{i=1}^{\dtheta} v_i \hat \Phi_i^B)^\top} \\&\geq \mu \norm{\sum_{i=1}^{\dtheta} v_i \hat \Phi_i^B}_F^2 \\&= \mu \geq \frac{1}{3\norm{P^\star}^{3/2}}.
    \end{align*}
\end{itemize}

With a different definition of regret measuring the distance between the incurred cost and the optimal cost of a set of controllers in hindsight, \cite{lale2020logarithmic} also show logarithmic regret in a partially observed setting. However, they operate under the assumption that the class of controllers is persistently exciting. Such a condition is not satisfied by the optimal state feedback LQR controller in general. In particular, the covariates under the policy $u_t= K^\star x_t$ are given by 
\begin{align*}
    \sum_{s=1}^t \bmat{x_s \\ u_s} \bmat{x_s \\ u_s}^\top  = \bmat{I \\ K^\star} \sum_{s=1}^t x_s x_s^\top \bmat{I \\ K^\star}^\top, 
\end{align*}
which is degenerate. For this reason, the structure through which the unknown parameters enter the dynamics become important. This is why we only can get $\log T$ regret in the setting of \Cref{s: no exp}. Future work could extend the analysis of this work to the partially observed setting to examine under which conditions on system structure it is possible to satisfy the persistency of excitation condition assumed by \cite{lale2020logarithmic}. See \cite{tsiamis2023statistical} for further discussion.

\end{document}